\documentclass[journal,comsoc]{IEEEtran}
\usepackage{color} 
\definecolor{mygreen}{RGB}{28,172,0} 
\definecolor{mylilas}{RGB}{170,55,241}

\usepackage{float}

\usepackage{multirow}
\usepackage{hyperref}
\usepackage{amsmath,amssymb,amsthm,bbm}

\usepackage{algorithm,algorithmic}

\usepackage{url}
\usepackage{color} 
 \usepackage{booktabs}

\usepackage{float}
\usepackage{graphicx}
\usepackage{caption}

\usepackage{subfigure}

\usepackage{cite}

\usepackage{epstopdf}
\newtheorem{definition}{Definition}
\newtheorem{lemma}{Lemma}

\newtheorem{proposition}{Proposition}
\newtheorem{claim}{Claim}
\makeatletter

\newcommand{\Rmnum}[1]{\expandafter\@slowromancap\romannumeral #1@}
\makeatother

\usepackage[T1]{fontenc}

\usepackage{amsmath}
\usepackage[cmintegrals]{newtxmath}
\usepackage{enumitem}

 \usepackage{pgfplots}
  \pgfplotsset{compat=newest}
  \usetikzlibrary{plotmarks}
  \usepgfplotslibrary{patchplots}
  \usepackage{grffile}
  \usepackage{amsmath}

\begin{document}
%
\title{Optimized Portfolio Contracts for Bidding the Cloud}

\author{Yang Zhang, Arnob Ghosh, and Vaneet Aggarwal
\thanks{Y. Zhang, A. Ghosh, and V. Aggarwal are with the School of Industrial Engineering, Purdue University, West Lafayette IN 47907, email: \{zhan1925, ghosh39, vaneet\}@purdue.edu. }}

\maketitle


\begin{abstract}

Amazon EC2 provides two most popular pricing schemes--i) the {\em costly} on-demand instance where the job is guaranteed to be completed, and ii) the {\em cheap} spot instance where a job may be interrupted. We consider a user can select a combination of on-demand and spot instances to finish a task. Thus he needs to find the optimal  bidding price for the spot-instance, and the portion of the job to be run on the on-demand instance. We formulate the problem as an optimization problem and seek to find the optimal solution. We consider three bidding strategies: one-time requests with expected guarantee and one-time requests with penalty for incomplete job and violating the deadline, and persistent requests. Even without a penalty on incomplete jobs, the optimization problem turns out to be non-convex. Nevertheless, we show that the portion of the job to be run on the on-demand instance is at most half. If the job has a higher execution time or smaller deadline, the bidding price is higher and vice versa. Additionally, the user never selects the on-demand instance if the execution time is smaller than the deadline.

The numerical results  illustrate the sensitivity of the effective portfolio to several of the parameters involved in the model. Our empirical analysis on the Amazon EC2 data  shows that our strategies can be employed on the real instances, where the expected total cost of the proposed scheme decreases over 45\%  compared to the baseline strategy.

\end{abstract}

\begin{IEEEkeywords}
Cloud pricing, spot instance, on-demand instance, optimization.
\end{IEEEkeywords}

%
\IEEEpeerreviewmaketitle

\section{Introduction}

\IEEEPARstart{C}{loud}  computing is projected to increase to \$162 billion in 2020. The latest quarterly results released from Amazon shows that Amazon Web Services (AWS) realized 43\% year-to-year growth, making contribution to 10\% of consolidated revenue and 89\% of consolidated operating income \cite{web3}. However, the success story of the CSPs inherently depends on the user's participation. The cloud service provider's (CSP's) prices affect the users' behavior and the profit of the CSP.  CSPs provide different pricing plans to meet customers' service requirements which we describe in the following.

\subsection{Cloud Pricing Schemes}

The most popular pricing schemes broadly adopted are: usage-based pricing, auction-based pricing, and volume-discount pricing \cite{lz16}.
Among the above, the most popular ones are the usage based and the auction-based pricing.  In the usage-based pricing, which is also known as pay-as-you-go, asks a fixed price per instance per hour and remains constant (static) over a long time. This type of pricing scheme is commonly implemented in Amazon \cite{web}, Google \cite{google}, Windows Azure \cite{azure}, etc. For example, Amazon EC2 on-demand instance provides fixed price short term service with no up-front payment or long-term commitment. The user will certainly get the resource on on-demand instance\cite{web}.  

On the contrary, in the auction-based pricing ({\it e.g.}, Amazon EC2 spot pricing),  users bid for the service, and the CSP sets a dynamic threshold to decide the successful bids based on the demand and bids. In each time slot, the bids that are above the spot price (which is decided by the CSP)  will be accepted, and others will be rejected.  The users pay the spot price\footnote{Hence, it has the similarity with the generalized second price auction.}. Although a user can bid a relatively lower price for the spot instance compared to the price it has to pay for the on-demand instance, the job may be interrupted when the bid is below a threshold \cite{web}. Typical job types like word counting, multimedia processing, etc. can be run using auction-based pricing strategies.

There are two types of spot instance requests: one-time requests and persistent requests. Specifically, the user bids with the instance type, bid price, etc., and the instance will start when the bid is higher than the spot price. When the user's bid price is lower than the spot price, the job will be interrupted and action taken afterwards relies on the request type: the interrupted job will be resumed when the bid price is above the spot price again if it is a persistent request and will be terminated permanently otherwise (i.e., if it is a one-time request). Figure \ref{fig20} depicts this procedure.

\begin{figure}[htbp]
\centering
\includegraphics[scale=1]{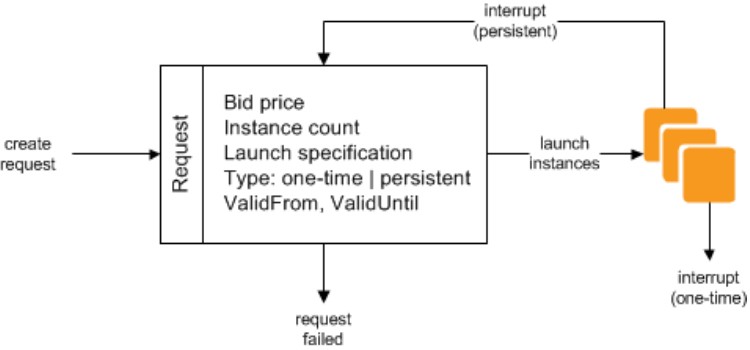}
\caption{Spot Instance Requests \cite{web2}  }
\label{fig20}
\end{figure}

\subsection{Research Challenges and Contributions}
The user will be likely to distribute its job over on-demand and spot instances. This is because in the on-demand instance, the user will be able to complete the job. In the spot market, the job may be interrupted. However, the user can pay less. Most of the existing literature considers the profit maximizing spot pricing from the CSP's perspective \cite{jin,feng,wang13,xu13,sha12}. It is also imperative to investigate the optimal decision of the users. The users needs to select a price to maximize in the spot-instance. The user also needs to select the portion to be run on the on-demand instance. In this paper, we propose a method that enables the users to decide how to make decisions in order to minimize the expected cost  while completing the job within the deadline. 

\begin{color}{black}
The closest work to ours is \cite{lz15}, which is motivated by Amazon EC2's auction-based spot pricing, they modeled the CSP's spot price setting and derive user's optimal bidding strategies. However, in the cloud infrastructure, more and more cloud jobs are requested for data analysis, such as web logs analysis, weather forecast analysis, finance  analysis, scientific simulation, etc. Most of them have hard deadlines, which can be predefined by the companies, or application providers. Failing to do meet the deadline, it may incur a penalty\cite{dli16}. For example, weather prediction is carried out by exploiting complex mathematical models based on the historical data such as temperature, atmospheric pressure,  humidity, etc. The Environmental Modeling Center runs the Global Forecast System model for 16 days into the future \cite{brown08}. If the computing job misses the deadline, some extreme weather may not be predicted in time, awareness and preparedness for the severe weather will be missed, resulting in large loss of human life,  social welfare, and financial resources if the severe weather happens. \cite{lz15} did not consider the deadline or the penalty incurred when a job misses the deadline. However, the deadline and penalty considered above impacts the bidding prices.

We consider that the users can access both the on-demand instances and the spot instances. On-demand instances exploit  the pay-as-you-go pricing scheme which guarantees the availability of the instances and there is no interruption of the job. Unlike on-demand instance, the spot instance uses the auction-based pricing scheme. The user bids for the spot instances, but its job will be interrupted when its bid is below the current spot price \cite{jin}. Thus, there is no guarantee that a job can be finished before its deadline if a user selects the spot instance. However, the user is likely to pay less for the spot instance. We consider that a job can be split into independent chunks, which can be processed on different machines in a parallel manner. We have seen many similar workloads in the real world, for example, word counting, multimedia processing, etc. \cite{lz16, wang155, tan12}.  

A user possessing the jobs, which can be run parallelly, may want to know whether a combination of on-demand and spot instances can be used to minimize the total cost while finishing the job before deadline. The user now has to select the portion of job to be completed via the on-demand instance and the spot instance. Additionally, the user has to select the bidding price for the participation in the spot instance \footnote{Though we consider the cloud computing market, our model can be applied to other markets. For example, in the  Display Advertising market of Internet, the spots are allocated in a two-stage process. In the first market, the publisher (e.g., Google's DoubleClick, OpenX, and Yahoo!'s Right Media) promises to deliver a contracted number of impressions within a fixed time slots (over a day); the second market (spot market) runs an auction to allocate the displays in every time frame (in an hour), where the advertisers arrive and bid for the displays \cite{chen}. The spot market is operated if an advertiser requires certain spots in the current time frame.  Thus, the advertiser has to select how much to bid in the spot market, and how much to buy fixed impressions in the first market.}. We propose an economic portfolio model for computing the optimal behaviors when it comes to how to allocate the job with known fixed deadlines to on-demand and spot instances and how much to bid for the spot instance  \footnote{Our approach can be applied in a MapReduce setting. Suppose we fix the number of instances $M$ to run for each job apriori. We need to split each job into two sub-jobs, and decide whether to run on spot or on-demand instance. Each sub-job will be run on $M$ instances.}.

 We suppose that each job has a fixed deadline, and a fixed execution time, which is the total time required to complete the job without any interruption. For example, suppose a job requires 30 minutes to finish. If it starts and gets interrupted after 10 minutes, we still need 20 minutes to execute the job. We consider two request mechanisms: one-time requests and persistent requests. Recall that with a one-time request, if a user's job is interrupted, it will not be resumed on the spot instance. Thus, the user's job may not be finished before the deadline by placing one-time requests. We consider two bidding strategies in one-time requests. The first one considers the user wants to finish its job before the deadline in an expected sense, and we denote this strategy as {\em one-time requests with expected guarantee} (OTR-EG) (Section \ref{wo-penalty}). However, it may not pay a penalty if it is incomplete or misses the deadline. Subsequently, we consider a strategy, where a user pays a penalty if the job is incomplete or misses the deadline, and we denote this strategy as {\em one-time requests with penalty} (OTR-P) (Section \ref{wi-penalty}). Finally, we consider the bidding strategy by placing persistent requests, and denote it as {\em persistent request} (PR) (Section \ref{sec:pers}), where the interrupted job can be resumed when the bid price is higher than the spot price again.\end{color}

\begin{color}{black} Our analysis shows that, in terms of one-time requests, {\em when the deadline is smaller than the execution time, the user should select the on-demand instances. The optimal bidding price in OTR-EG will decrease first and then increase with the deadline, while the optimal bidding price in OTR-P will increase with the deadline.} We also show the optimal bidding prices on spot instance in OTR-P increase with the increase of the penalty coefficients, and very small or very large penalty coefficient for incomplete jobs will lead to a slower increase of bidding price. However, {\em when the deadline is larger than the execution time, the user will solely depend on spot instance to finish the job, and the optimal bidding prices for OTR-EG and OTR-P do not change with the increase of the deadline.}\end{color}

We, subsequently, consider the case where a user places a persistent request for the spot instance. In the persistent request, unlike the one-time request, an interrupted job will be resumed when the bid price exceeds the spot price again. A lower bid can reduce the cost of executing the job on the spot instance, while the number of interruptions may be increased, so dose the total idle time, total recovery time and total completion time, which may exceed the deadline. Thus, it is not apriori clear that how much portion of the job should be run on the spot-instance, and what the bidding price will be  if we want to finish the job before the deadline in expectation. {\em Our result shows that the persistent request reduces the expected cost of the user as compared to the one-time-requests. Similar to the one-time-request, only when the deadline is smaller than the execution time, the user selects the on-demand instances.} Note that we did not consider any penalty based approach in the persistent request. This is mainly because in the persistent request, the interrupted job is not discarded and thus, it will finish unlike the one-time request. 


The main contributions of this paper can be summarized as follows: 

\begin{itemize}[leftmargin=*]
\item \textbf{User's optimal or local optimal bidding strategies}: For the one-time request and persistent request job, we formulate the cost minimization problem as an optimization problem. The problem turns out to be non-convex. Nevertheless, we  find analytical expression for the optimal solutions for the one-time request without penalty and the persistent request.  However, for the one-time request with  penalty, we provide algorithms for solving the proposed non-convex problem. 

\item \textbf{Analytical Results}: Our analytical result shows that only when the deadline is smaller than the execution time, the user should select the on-demand instances. We show a threshold type behavior for one-time request. When the penalty is above a certain threshold, the user opts for the on-demand instances. However, below the threshold, the portion of the job that is run on the on-demand instance becomes independent of the penalty parameters. Our result shows that the persistent requests reduce the expected cost of the user compared to the one-time-request.
\item \textbf{Numerical Evaluation}: We, empirically, evaluate the impact of different parameters on the portion of the job should be run on the spot instances, and the bidding price. Our result shows that the expected cost, and the portion of the job that is run on the on-demand instance decreases with the increase in the deadline. The bidding price in the persistent request instance decreases with the increase in the deadline. However, the bidding price in the one-time request increases with the increase in the deadline in the one-time request. 

\item \textbf{Real time Data}: Using the real time data, we show the strength of our approach compared to the baseline strategies readily employed by the users. Specifically, we compute the optimal bidding strategy in the spot-instance, and the optimal portion of the job should be run on the on-demand instance. Finally, we show that the user's cost is reduced using our approach compared to the baseline ones.

\end{itemize}

\subsection{Related Literature}
The genre of works can be divided based on the topics they considered.

\textbf{Portfolio Contract}: This type of portfolio contract has been practiced and studied in many other contexts especially in procurement, e.g., Hewlett-Packard (HP) uses a portfolio approach for procurement of electricity or memory products \cite{hp}. Motivated by that practice, the procurement has been studied in multi-period \cite{sim05} and single-period \cite{qi10} settings. However, the above portfolio contracts did not study the cloud spot market, the deadline, and the execution time.


\textbf{Deadline-based cloud scheduling}: Deadline-based resource allocation has been considered in many cloud research works. While resource allocation approaches are utilized in the cloud context, which aims to  meet the jobs' deadlines and utilize the cloud resource more efficiently \cite{li16} or minimize the total execution cost \cite{maria14, rod13}, they only consider from the CSP's perspective. In this paper, we develop a model to optimize the bidding strategies of the user. Although in \cite{lz15}, optimal one-time request and persistent request bidding strategies are proposed, they do not consider the deadline of the user's job, which may be not practical \cite{li16}. In this paper, we not only consider one-time request without penalty and persistent request bidding strategies, we also include one-time request with penalty model to balance the finished job and penalty for the unfinished job or late completed job. This model can be applied to the type of the job with a soft deadline, which is a deadline when it is unmet, dose not lead to computation useless \cite{abb88,zhou17}. 
 

\textbf{Game Theory, Auctions and Bidding}: Game Theory has been used to model the interactions between CSPs and users to reach an equilibrium \cite{mak11, ash09, danilo11, ard13}. In distributed resource allocation games, auctions have been proposed to be a solution \cite{kut99,song13}, including to ensure truthful bidding in Amazon spot pricing \cite{qwang13}.  

The remainder of this paper is organized as follows. Section \ref{system} introduces the system model. In Section \ref{User Bidding Strategies}, we present three types of bidding strategies: one-time request without penalty (Section III-A) and with penalty (Section III-B), and persistent request (Section III-C). In Section \ref{simulation}, extensive simulation results show the benefits of each strategy. We test our proposed model and results using Amazon spot price history in Section \ref{real}. Finally,  Section \ref{con} concludes this paper. We relegate all the proofs in Appendix.

\section{System Model}\label{system}
We consider a CSP, which can provide two types of computing instances: on-demand instance and spot instances. On-demand instance can guarantee the availability, but the price is fixed and high. In order to provide a reduced-cost service, the CSP also offers spot instance, which may terminate unpredictably since the price fluctuates based on availability and demand, and update spot price in every certain time period, {\it e.g.}, every 5 minutes. The users can run its job on the spot instance as long as  the bid price exceeds the spot price. 

We consider a user can select a combination of on-demand and spot instance to finish a  task.  In other words, the user decides the portion of the job to be run on the on-demand instance and the rest in the spot instance. The spot price is much lower than the on-demand price for every instance type \cite{web}. However, the spot market cannot guarantee that the task is run continuously if her bidding price is not high enough, which means the task may be interrupted and takes extra time to get recovered, so that the task may take longer time to get finished. Therefore, the user should balance the proportion of the job she runs on on-demand instance, with the bidding price in an spot market to run the rest of the job on spot instance. This paper aims to provide a framework to help users to decide how much to run at on-demand instances and how much to bid for spot instances with the objective to minimize the total cost, subject to the constraint that deadline has to be satisfied.


\begin{color}{black}We consider a series of discrete time slots $t \in \{ 1, 2, \cdots\}$ and denote the spot price at time slot $t$ as $\pi(t)$. We assume the spot prices $\pi(t)$ are i.i.d,  upper-bounded by the on-demand price $\bar{\pi}$ for the same instance type and lower-bounded by the marginal cost of the instance $\underline{\pi}$, which is very small and closed to 0 \cite{xu13}, that is, $\underline{\pi} \leq \pi(t) \leq \bar{\pi}$. We use $F_{\pi}$ to denote the cumulative distribution function (CDF) of spot price $\pi(t)$, which is heavy-tailed \cite{lz15}, corresponding to the probability density function (PDF) $f_{\pi}$. We use $p$ to denote the user's bid price. $F_{\pi}(p)$ gives the probability that $p\geq \pi(t)$, that is, the user's bid gets accepted. We assume $f_{\pi}$ monotonically decreases, thus $F''_{\pi}(p) = f'_{\pi}(p) < 0$, i.e., $F_{\pi}(p)$ is concave in $p$, which is consistent with the observations and findings in \cite{lz15}.\end{color}

Suppose a user has a certain job $J$, which can be split and run on different machines. First, the user would like to purchase on-demand instances to ensure a certain desired level of finished job in the future; say, $q$ portion of the total amount of task run on on-demand instance. And the rest portion of the job $(1-q)$ will be run on spot instances. Then the user needs to decide how much to bid ($p$) to the spot market. The strategy of the user is to decide $p$ and $q$. More formally, we define the strategy of a user in the following

\begin{definition}
The strategy of a user is the vector $\textbf{x}=(q,p)$. 
\end{definition}
A user decides $(q,p)$ while minimizing the expected cost.   Figure \ref{fig10} depicts the major considerations that we need to incorporate in the bidding and resource allocation decisions graphically. 

\begin{figure}[htbp]
\centering
\includegraphics[scale=0.4]{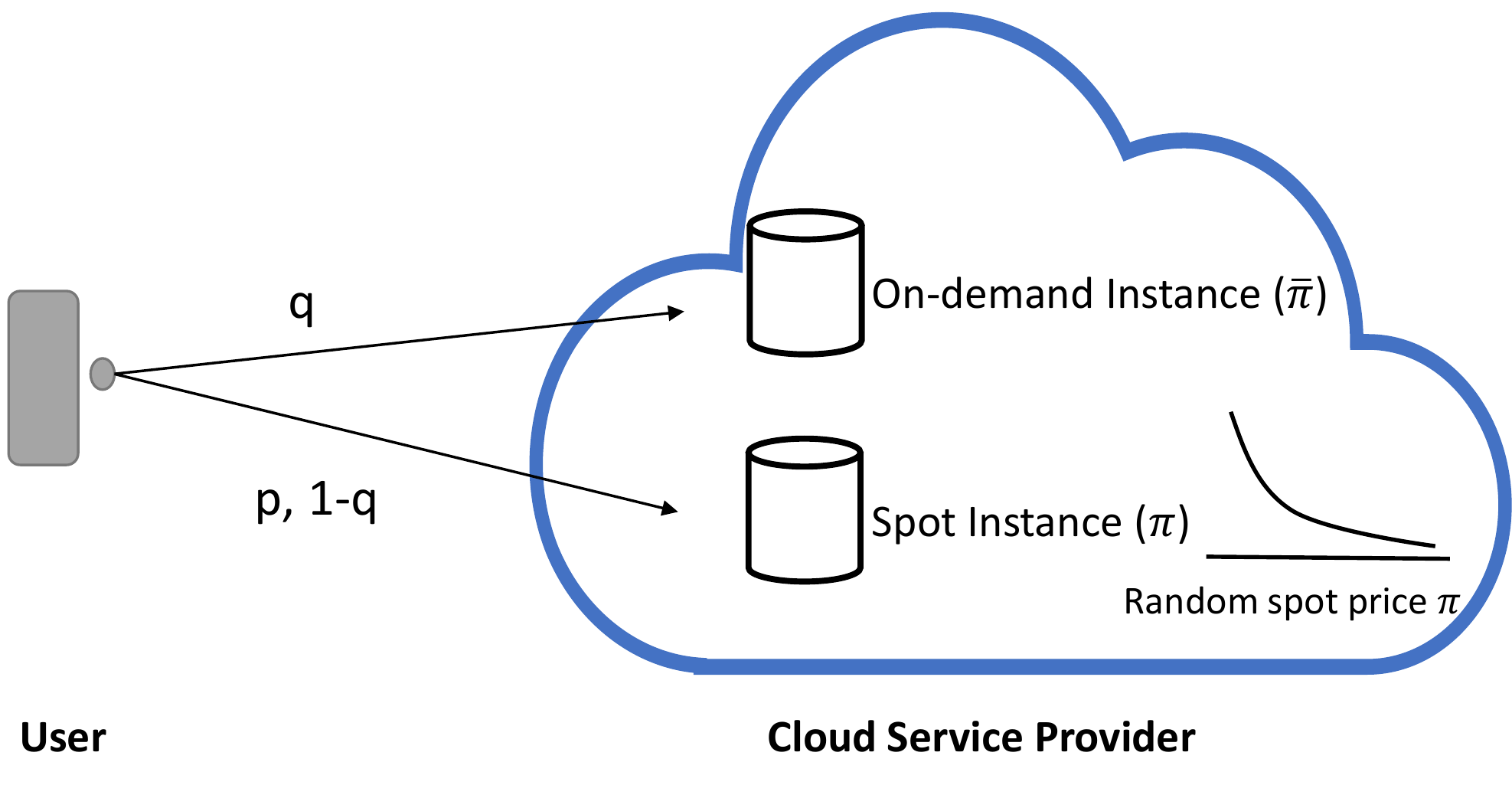}
\caption{User Decision Model}
\label{fig10}
\end{figure}

\begin{color}{black} We consider three bidding strategies: OTR-EG, OTR-P, and PR, where the first two strategies can be used if the user places a one-time requests, and the third strategy will play a role when the persistent requests are placed.  \end{color}
The problem is to design optimal portfolio of contracts and bidding strategies in different settings, so that the expected total cost is minimized, subject to the deadline constraints. We use $\textbf{x}^{*}=(q^{*},p^{*})$ to denote user's optimal decisions. Then we will investigate the extent of the benefits that can be accrued by managing a portfolio of contracts instead of sticking to on-demand instance contract.

Our notations are summarized in Table \ref{tab1}. We not only consider the job's characteristics such as its execution time $t_{e}$, total completion time $T$, the recovery time for writing and transferring the data saved after interruption $t_{r}$, but also include the deadlines $t_{s}$ on the job completion times. 

\begin{table}[H]
\centering
\caption{Key terms and symbols}
\label{tab1}
\begin{tabular}{|l|l|}
\hline
\textbf{Symbol}   & \textbf{Definition}                                    \\ \hline
$p$               & User bid price                                         \\ \hline
$q$               & The portion of job that will run on on-demand instance \\ \hline
$\pi$             & Spot price                                             \\ \hline
$\bar{\pi}$       & On-demand price                                        \\ \hline
$\underline{\pi}$ & Minimum spot price                                     \\ \hline
$t_{k}$           & Length of one time slot                                \\ \hline
$T$               & Total job completion time                              \\ \hline
$t_{s}$           & Deadline of the job                                    \\ \hline
$t_{e}$           & Job execution time (w/o interruptions)                 \\ \hline
$t_{r}$           & Recovery time from an interruption                     \\ \hline
$c_I$               & Penalty coefficient for incomplete job                                  \\ \hline
$c_s$              & Penalty coefficient for late completed job                                    \\  \hline
\end{tabular}
\end{table}

\section{User Bidding Strategies}\label{User Bidding Strategies}

In this section, we first consider OTR-EG, subsequently, OTR-P, and finally, PR for a single instance on each machine type.

\subsection{OTR-EG}\label{wo-penalty}

In one-time request, a job on spot instance will not be resumed as soon as the job is interrupted, the user's objective is to minimize the total cost. However, the job has to be completed before the deadline. In order to make sure that there exists at least one feasible solution such that  the job can be finished before deadline, we assume that $t_{e} \leq 2t_{s}$. The factor 2 comes from the fact that the smallest time a job can be completed when the half of the job is run on the on-demand instance, and the rest in the spot instance. Thus, for a feasible solution, $0.5t_e\leq t_s$. 

First, we compute the expected amount of time that a job will continue running without any interruptions if the user bids $p$. Note that when a user bids the price $p$, its bid will only be accepted if the spot price $\pi$ is lower than $p$. Thus, the probability that the bid will be accepted at an instance with probability $1-F_{\pi}(p)$.  Thus, we have the following
\begin{lemma}\cite{lz15}
The expected amount of time that a job will continue running without any interruptions is: 
\begin{equation}
t_{u}(p)=t_{k}\sum_{i=1}^{\infty}i F_{\pi}(p)^{i-1}(1-F_{\pi}(p)) = \frac{t_{k}}{1-F_{\pi}(p)}
\end{equation}
\end{lemma}
In order to guarantee that the job that runs on spot instance can be finished, that is, the expected amount of time that a job will keep running must exceed its execution time, we need the following constraint:
\begin{equation}\label{P1-1}
(1-q)t_{e}\leq \frac{t_{k}}{1-F_{\pi}(p)}. 
\end{equation}

Now, we compute the expected time tn for a job to enter the system when the user bids p. Note that a job can only enter the system if the spot price is lower than the bid price. The random variable that the bid gets into the system follows a Geometric distribution. Thus we get the following term:

\begin{equation}\label{tn}
t_{n}=t_{k}\sum_{i=1}^{\infty}i (1-F_{\pi}(p))^{i}F_{\pi}(p)= t_{k}(\frac{1}{F_{\pi}(p)}-1)
\end{equation}
 From \eqref{tn}, we notice that $t_{n}$ monotonically decreases with $p$. Thus we would intuitively expect that in order to finish the job $(1-q)t_{e}$ before deadline, the user should bid more to shorten the expected amount of time to enter the system $t_{n}$. 

In order to finish the job that run on spot instance, the deadline need to be longer than the summation of the expected time to enter the system and the required execution time $(1-q)t_{e}$, that is, 
\begin{equation}\label{P1-3}
t_{n}+(1-q)t_{e} \leq t_{s}. 
\end{equation}

\begin{lemma}\cite{lz15}
The expected price that a user must pay to use an instance in each time slot on spot instance, or the expected value of all possible spot prices that are no more than $p$ is  
\begin{equation}
 \mathbb{E}(\pi | \pi \leq p) = \frac{ \int_{\underline{\pi}}^{p}x f_{\pi}(x)dx}{F_{\pi}(p)} 
 \end{equation}

\end{lemma}

\begin{lemma}\label{lemma3}
$ \mathbb{E}(\pi | \pi \leq p)$ monotonically increases with $p$ and not larger than $\frac{\bar{\pi}+\underline{\pi}}{2}$. 
\end{lemma}
\begin{proof}
	The proof is provided in Appendix \ref{prolemma3}.
\end{proof}

If the user's bid price is $p$, the user has to pay for the spot instance is the expected spot price $\mathbb{E}(\pi | \pi \leq p)$. The user puts $(1-q)$ fraction of the job on the spot market. Thus the user's expected cost for running the job on the spot instance is $(1- q)t_{e}\mathbb{E}(\pi | \pi \leq p) $. Recall that on the on-demand instance, a user has to pay the price $\bar{\pi}$. The user's cost for running $q$ fraction of the job on the on-demand instance is $ q t_{e}\bar{\pi}$. Thus the total expected cost of running the job is

 $$ q t_{e}\bar{\pi} + (1- q)t_{e}\mathbb{E}(\pi | \pi \leq p). $$ 
 
 The user also has to make sure that its job is completed before the deadline. In other words, the expected time the job will take to finish must be smaller than the deadline $t_s$. The total time a job takes in the on-demand instance is $qt_e$ and in the spot-instance is given by $t_{n}+(1-q)t_e$.  Thus, the total time to complete the job is  $\max\{qt_e, t_{n}+(1-q)t_e\}$. Hence, the user is solving the following problem:

\begin{align}
\begin{split}
{\text{(P1) \quad   min}} &
\quad  \Phi_{1}(p,q) = q t_{e}\bar{\pi} + \frac{ (1- q)t_{e}  \int_{\underline{\pi}}^{p}x f_{\pi}(x)dx}{F_{\pi}(p)}
  \end{split}
\label{obj1} 
\\[2ex]
\text{subject to}\qquad & (\ref{P1-1}), \quad (\ref{P1-3}) \nonumber
\\
& q t_{e} \leq t_{s}
\label{P1-2}
\\
& \underline{\pi} \leq  p \leq \bar{\pi}
\label{P1-4}
\\
& 0 \leq q \leq 1
\label{P1-5}
\end{align}

The objective function (\ref{obj1}) aims to minimize the expected total cost running on on-demand and spot instance.
In order to guarantee that the job can be completed before deadline, we include constraints (\ref{P1-2}) and (\ref{P1-3}), which represent that the both of maximum job completion time on each instance including the job execution time and time to enter the system (if any) should not exceed the deadline.  The constraint in (\ref{P1-4}) denotes the upper and lower bound of the bidding price in the spot instance.

\begin{claim}\label{cla1}
When $t_{s}< t_{e} \leq 2t_{s}$, $q^{*} \leq \frac{1}{2} \leq \frac{t_{s}}{t_{e}}$ and $F_\pi(p^{*}) \geq \frac{1}{2}$. 
\end{claim}
\begin{proof}
	The proof is provided in Appendix \ref{claim1}. 
\end{proof}
The above theorem shows that $q^{*}$ is at most half. Thus, at most half of the job is put on the on-demand instance.
Intuitively, on-demand price is larger than spot price, if the user wants to minimize his total cost, he will run as much job as possible on spot instance. However, if the deadline is smaller than the execution time, the user may have to opt for on-demand instance as the user has to complete the job before the deadline. The above claim shows that the fraction of the job that will be run on on-demand instance never exceeds half. The results show that the bidding price in the spot market has to be at least the median of the distribution.   


\begin{proposition}\label{pro1}
When $\frac{t_{e}}{2} < t_{s} < t_{e}$, the optimal bid price for a one-time request is
\begin{equation}
 p^{*}=\max \{\psi_{1}^{-1} (t_{k}\bar{\pi}), \psi_{2}^{-1} (0)\},
\end{equation}
  where $\psi_{1}^{-1}(.)$ is the inverse function of 
 \begin{equation}
\psi_{1}(p)  =\frac{ 2t_{k} \int_{\underline{\pi}}^{p}x f_{\pi}(x)dx  }{  F_{\pi}(p) } + 2pt_{s}F_{\pi}(p) - pt_{s} - (t_{s}+t_{k}) \int_{\underline{\pi}}^{p}x f_{\pi}(x)dx
\end{equation}
 and $\psi_{2}^{-1}(.)$ is the inverse function of 
 \begin{equation}
 \psi_{2}(p)  =   (t_{s}+t_{k})F_{\pi}(p) - (t_{s}+t_{k})F_{\pi}(p)^{2} - t_{k}
 \end{equation}
 with $F_{\pi}(p) \geq \frac{1}{2}$.  Further the optimal portion of the job to run on on-demand instance is 
 \begin{equation}
q^{*} = 1- \frac{  t_{s}-t_{k}( \frac{1}{F_{\pi}(p^{*})} -1  )   }{t_{e}}. 
\end{equation}
\end{proposition}
\begin{proof}
	The proof is provided in Appendix \ref{apdxpro1}. 
\end{proof}

Proposition \ref{pro1} implies that the portion of job that runs on the on-demand instance $q^{*}$ decreases as the deadline $t_{s}$ increases. Intuitively, as the deadline increases, a user can be more likely to run the job in the spot instance as the job can be more likely to be finished using spot instances instead of the on-demand instance resulting into a lower cost. The above proposition also shows that the optimal bidding price $p^{*}$ takes the maximum value of two functions. The intuition is that with certain portion of job running on spot instance, lower price can decrease the total cost, however, in order to guarantee that the spot instance can continue running without any interruption, the price cannot get too low.

Note that though the optimization problem is non-convex, we still obtain the optimal strategy. $q^{*}$ is non-zero, however it is less than half. 


\begin{proposition}\label{pro2}
When $ t_{s} > t_{e}$, the optimal bid price for a one-time request is
\begin{equation}
 p^{*}=F_\pi^{-1}(1-\frac{t_{k}}{t_{e}}).
 \end{equation}
Further, the optimal portion of the job to run on on-demand instance is 
 \begin{equation}
q^{*} = 0.
\end{equation}
\end{proposition}
\begin{proof}
	The proof is provided in Appendix \ref{apdxpro2}. 
\end{proof}


We can observe from Proposition \ref{pro2} that when $t_{s} > t_{e}$, all of the job will be run on spot instance, and the optimal bid price $p^{*}$ does not depend on the deadline $t_{s}$, but instead {\em increases as the number of time slots that are needed to complete the job, $t_{e}/t_{k}$ increase}.  This increase in the bid price with the $t_e/t_k$  is intuitive because more consecutive time slots are required to complete the job, and thus, a higher bid is needed.

\subsection{OTR-P}\label{wi-penalty}

In section \ref{wo-penalty}, we consider that if the job is interrupted, it can not continue. However, there are some possibilities that the job will not get completed before deadline or get interrupted before completion. In this section, we consider the scenario where a user has to incur a penalty when the job is not completed before the deadline. Note that in the one-time request there can be two possible ways the job may not be completed before the deadline: i)  {\em{The job is incomplete}},  and  ii) {\em{The job is late}}. We now define each of them.

\begin{definition}
Incomplete Job: the job is interrupted before its completion.
\end{definition}

\begin{definition}
Late Job: the job is completed (i.e. it is never interrupted), however the total time it takes is greater than the deadline. For example, when the time to enter the system is long, the job may get completed beyond the deadline. 
\end{definition}

In section \ref{wo-penalty}, we put a constraint where we consider the expected time for completing the job is less than the deadline. However, as the spot price is random, the job that we run on the spot instance may not be completed (as it is one-time request) or may be completed after the deadline. In this section, we compute the optimal solution where we put penalty for the job which is incomplete or late.


If the job is not completed (i.e., the case (i) holds), there is a penalty $c_I$ associated with the unfinished portion of the job\footnote{If the job is not complete, one may need on-demand instances or incur penalty for the unfinished job.}. If the job is not interrupted, however, it is completed after the deadline,  there is another penalty $c_s$ for the portion of the job that is completed after the deadline. We also assume $c_{s} \leq c_{I}$, which means the completed job will have a lower penalty than that of incomplete one. We denote  the total number of slots needed to complete the job in the spot instance is $K(q) = \frac{(1-q)t_{e}}{t_{k}}$. 

The user wants to minimize the expected cost which also consists of the expected penalty for incomplete jobs.  We begin by finding the expected total penalty and then formulate the optimization problem before deriving the user's optimal bid price. We now compute the expressions.

\begin{definition}
Let $L(p,q)$ be the expected time by which a completed job is late, {\em i.e.}, $L(p,q)=(t_c-t_s)^{+}$ where $t_c$ is the time to complete the job when the user's strategy is $(p,q)$. 
\end{definition}

\begin{lemma}\label{the1}
\begin{equation}
\begin{split}
L(p,q) = t_{k}(1-F_\pi(p))^{\frac{t_{s}}{t_{k}}-K(q)+1}F_{\pi}(p)^{K(q)-2} 
\end{split}
\end{equation}
\end{lemma}

 Recall that $F_{\pi}(p)$ represents the probability that $p \geq \pi$, that is, the request starts to run or continues running (we denote it as {\em ``success''}); and $1-F_{\pi}(p)$ is the probability that the request fails or get terminated (we label it as {\em ``failure''}). The intuition behind Lemma \ref{the1} is from when the user places the bid, a Bernoulli trial is ``conducted'' in each time slot. In order to guarantee the portion of job $(1-q)t_{e}$ can get completed on spot instance, a fixed number $K(q) = \frac{(1-q)t_{e}}{t_{k}}$ statistically independent Bernoulli trials' results need to be``success'' successively, which happens with probability $F_{\pi}(p)^{K(q)-1}$. On the other hand, the random variable that the bid gets the first ``success'' follows a Geometric distribution. For example, when the bid dose not win until the $M$th time slots, the probability is $(1-F_{\pi}(p))^{M-1}F_{\pi}(p)$. When the number of time slots that the bid spends without getting accepted is larger than $\frac{t_{s}}{t_{k}}- K(q)$, the job may be completed but late. Considering all the possibilities of the late but completed job, we have the expression in Lemma \ref{the1}.  



\begin{definition}
Let $EC(p,q)$ be the portion of the job that is completed on the spot instance.
\end{definition}

\begin{lemma}\label{the2}
 The expected portion of the job that can be completed when the user bids the price $p$
\begin{equation}
\begin{split}
EC(p,q) = \frac{1-F_{\pi}(p)^{K(q)}}{1-F_{\pi}(p)}t_{k}
\end{split}
\end{equation}
\end{lemma}

The intuition behind Lemma \ref{the2} is in the Bernoulli process, the expected completed job is from the first ``success'' to the job interruption (the first ``failure''  from the first ``success'') or the job completion.

\begin{definition}
Let $EI(p,q)$ be the portion of the job that is incomplete in the spot instance when the user bids the price $p$. 
\end{definition}

$EI(p,q)$  is simply the difference between the total portion of job running on spot instance and the expected portion of job that can be completed with bid price $p$ on spot instance, thus we can obtain

\begin{lemma}\label{the3}
\begin{equation}
\begin{split}
EI(p,q) = (1-q)t_{e}-EC = (1-q)t_{e}-\frac{1-F_{\pi}(p)^{K(q)}}{1-F_{\pi}(p)}t_{k}
\end{split}
\end{equation}
\end{lemma}

Considering the penalty for incomplete job and completed but late job, the user solves the following optimization problem:

\begin{align}
\begin{split}
{\text{(P2) \quad  min \quad}}  U & =  q t_{e}\bar{\pi} +  \frac{\int_{\underline{\pi}}^{p}x f_{\pi}(x)dx}{F_{\pi}(p)} EC(p,q)\\
&+c_{I}EI(p,q) +c_{s}L(p,q)
\end{split}
\\
\text{subject to}&  \quad (\ref{P1-2}),\quad  (\ref{P1-4}), \quad (\ref{P1-5})
\nonumber 
\\
& (1-q)t_{e} \leq t_{s}
\label{P2-4}
\end{align}

{\bf Solution Method:} Note that if $c_I$, $c_s$ are large, $q$ and the price should increase in order to avoid hefty penalty. Although the constraints in (P2) are linear, the objective function is non-convex. Thus problem (P2) is non-convex. Unlike the problem in the OTR-EG, we cannot have any closed form for (P2). We use the successive convex approximation based  algorithm  \cite{kojima2001complexity}, which iteratively solves approximate convex relaxation of the problem. The algorithm is stated in Algorithm \ref{alg11}. Let $\tilde{U}$ as the approximation of the objective function $U$, which is the first order approximation of $U$, that is,

\vspace*{-0.15in}
\begin{equation}
\tilde{U}(\textbf{x};\textbf{y}) = \sum_{i=1}^{2} (\nabla_{x_{i}}U(\textbf{y})^{T}(x_{i}-y(i)) + \frac{\tau}{2}(x_{i}-y_{i})^{2}  ).
\end{equation}

Instead of solving $U$, we solve $\tilde{U}$ iteratively, which is shown in Algorithm \ref{alg11}. When the difference between two successive objective values is smaller than $\epsilon = 10^{-5}$, the iteration stops.

\begin{algorithm}\label{alg11}
 \caption{Successive Convex Approximation  Algorithm to solve (P2)}
 \begin{algorithmic}[1]\label{alg11}
 \renewcommand{\algorithmicrequire}{\textbf{Input:}}
 \renewcommand{\algorithmicensure}{\textbf{Output:}}
 \REQUIRE $\nu = 0$, $k = 0$, $\gamma \in (0, 1]$, $\epsilon > 0$, $\textbf{x}^{0}=(q^{0}, p^{0})$ such that $\textbf{x}^{0}$ is the solution of OTR-EG.
 \ENSURE  $\hat{\textbf{x}}(\textbf{x}^{\nu})$
 \WHILE{obj(k) - obj(k-1) $\geq \epsilon$}
 \STATE // solve for $\textbf{x}^{\nu+1}$ with given $\textbf{x}^{\nu}$.
 \STATE $\textbf{Step 1:}$ Compute $\hat{\textbf{x}}(\textbf{x}^{\nu})$, the solution of  $\hat{\textbf{x}}(\textbf{x}^{\nu})= \text{argmin}  \tilde{U}(\textbf{x}, \textbf{x}^{\nu}) $, s.t.  (\ref{P1-2}), (\ref{P1-4}), (\ref{P1-5}), and (\ref{P2-4}), solved using CVX.
 \STATE $\textbf{Step 2:}$ $\textbf{x}^{\nu+1} = \textbf{x}^{\nu}+\gamma^{\nu}(\hat{\textbf{x}}(\textbf{x}^{\nu}) - \hat{\textbf{x}}^{\nu})$.
 \STATE //update index
 \STATE  $\textbf{Step 3:}$ $\nu \gets \nu+1 $. 
 \ENDWHILE
 \end{algorithmic} 
 \end{algorithm}

\subsection{PR}\label{sec:pers}
In section \ref{wo-penalty} and \ref{wi-penalty}, we consider the one-time request job. We now consider a job that places a persistent spot instance request, where the job can be interrupted and recovered upon resuming when the bid price is above the spot price. 

We, first, compute the total time $T$ for completing a job in the PR in spot instance. The expected running time is $TF_{\pi}(p)$ with bidding price $p$, and the associated expected idle time is $(1-F_{\pi}(p))T$. The expected number of idle-to-running transitions in $T/t_{k}$ time slots is $\frac{T}{t_{k}}F_{\pi}(p)(1-F_{\pi}(p))$. We incur a recovery time every time there is a transition from the ideal state to the running state. Thus, $TF_{\pi}(p)=( \frac{T}{t_{k}}F_{\pi}(p)(1-F_{\pi}(p)) )t_{r}+(1-q)t_{e}$, we get $TF_{\pi}(p) = \frac{(1-q)t_{e}}{  1- \frac{t_{r}}{t_{k}}(1-F_{\pi}(p))   }$ \cite{lz15}. 
\begin{lemma}
The total time including the recovery, execution and idle time is 

\begin{equation}
\begin{split}
T & = \frac{(1-q)t_{e}}{  1- \frac{t_{r}}{t_{k}}(1-F_{\pi}(p))   }\frac{1}{F_{\pi}(p)}
\end{split}
\end{equation}
\end{lemma}
Note that as $F_{\pi}(p)$ increases the time decreases. 
\begin{align}
\begin{split}
{\text{(P3) \quad   min \quad}} &
\Phi_{3}(p,q) = q t_{e}\bar{\pi} +\frac{ (1- q)t_{e} }{1-  \frac{t_{r}}{t_{k}} (1-F_{\pi}(p) )   } \frac{ \int_{\underline{\pi}}^{p}x f_{\pi}(x)dx}{F_{\pi}(p)}
  \end{split}
\label{green} 
\\[2ex]
\text{subject to}\qquad &(\ref{P1-2}),\quad  (\ref{P1-4}), \quad (\ref{P1-5})
\nonumber 
\\
&   \frac{(1-q)t_{e}}{  1- \frac{t_{r}}{t_{k}}(1-F_{\pi}(p))   }\frac{1}{F_{\pi}(p)}  \leq t_{s} 
\label{green-constraint-1} 
\\
& t_{r} < \frac{t_{k}}{2(1-F_{\pi}(p))}
\label{green-constraint-20}
\end{align}

The expected total time including the recovery, execution and idle time should be smaller than the deadline, so we get constraint (\ref{green-constraint-1}). The constraint in (\ref{green-constraint-20}) guarantees that the recovery time is sufficiently small such that the job's running time is finite \cite{lz15}. 

\begin{claim}\label{cla2}
In PR, when $\frac{t_{e}}{2} < t_{s} < t_{e}$, $F_{\pi}(p^{*}) \geq \frac{1}{2}$. 
\end{claim}
\begin{proof}
	The proof is provided in Appendix \ref{claim2}.
\end{proof}

The above claim shows that in the spot instance a user's bid should exceed the median value of the distribution.

\begin{proposition}\label{pro3}
When $\frac{t_{e}}{2} < t_{s} \leq t_{e}$, the optimal bid price for a PR is
\begin{equation}
 p^{*}=\bar{\pi}
 \end{equation}

Further, the optimal portion of the job to run on on-demand instance is 
\begin{equation}
q^{*} = 1- \frac{t_{s}  }{t_{e}}
\end{equation}
\end{proposition}
\begin{proof}
	The proof is provided in Appendix \ref{apdxpro3}.
\end{proof}

The above proposition entails that when $\frac{t_e}{2}<t_s\leq t_e$, the user bids the highest possible value in the spot instance. The user also runs a portion of the job in the on-demand instance. The portion of the job that is run on the on-demand instance is given by $$   q^{*} = 1- \frac{t_{s}  }{t_{e}}.   $$ If $t_e$ is large, $q^{*}$ is higher.  Proposition \ref{pro3} thus implies that surprisingly, {\em  the bidding price does not change as the deadline changes}, and the optimal portion of the job that run on on-demand instance is decreasing with the deadline. The bidding price at the spot instance is the maximum possible bidding price. This is intuitive, because when $t_e$ is large, the bid price has to be large.


\begin{proposition}\label{pro4}
When $ t_{s} > t_{e}$, the optimal bid price for a PR is
\begin{equation}
 p^{*}=\psi_{3}^{-1} ( \frac{t_{e}}{t_{s}}  ).
 \end{equation}

where $\psi_{3}^{-1}(.)$ is the inverse function of 

$$  \psi_{3}(p) = F_{\pi}(p)[ 1-\frac{t_{r}}{t_{k}}(1-F_{\pi}(p))] $$

The optimal portion of the job to run on on-demand instance is 
\begin{equation}
q^{*} = 0 
\end{equation}
\end{proposition}
\begin{proof}
	The proof is provided in Appendix \ref{apdxpro4}.
\end{proof}

This shows that similar to the one-time request, the portion of the job that is run on the on-demand instance is $0$ when $t_s>t_e$. Thus, a job will be run on the on-demand only when $t_s\leq t_e$. Also note that $\psi_3(\cdot)$ is an increasing function. Hence, as the ratio $\dfrac{t_e}{t_s}$ increases the bidding price also increases. 

\begin{color}{black}
\begin{lemma}\label{lemmaaa}
When $t_{s}< t_{e} \leq 2t_{s}$, the difference of optimal portions of job to run on on-demand instance with OTR-EG and PR is bounded by $\frac{t_{k}}{t_{e}}$.  
\end{lemma}
\end{color}

\begin{proof}
	The proof is provided in Appendix \ref{apdxlemmaaa}.
\end{proof}

\section{Numerical Studies}\label{simulation}
In this section, we present computational results that illustrate the sensitivity of the expected total cost, the corresponding bid price and the portion of job that run on on-demand instance in terms of the different parameters used in the model. We specifically focus on the impact of the deadline, penalty coefficient, and the recovery time. \begin{color}{black}  Note that we use closed-form solutions for our OTR-EG and PR, and convex approximation algorithm for the OTR-P.\end{color}

\subsection{Distribution of Spot Price }
We, first,  introduce the spot price probability density function. In many applications including the cloud spot market, the prices follow a Pareto distribution \cite{lz15}. The PDF of the spot price is chosen to be

\begin{equation}\label{pdf}
f_{\pi}(\pi)=\frac{\alpha (\frac{1}{\bar{\pi}-\underline{\pi}})^{\alpha}}{\theta(\frac{1}{\bar{\pi}-\pi})^{\alpha+1}}
\end{equation}

which behaves like a Pareto Distribution, the random variable $\pi$ is bounded by $\underline{\pi}$ and $\bar{\pi}$. 

\subsection{Simulation Set Up}
In this section, we consider a job that needs one hour ( i.e., $t_{e}=1h=3600s$) , the deadline $t_{s}$ is $2000s$, the length of one time slot is 5 minutes (i.e., $t_{k}=5min = 300s$), and the recovery time $t_{r}=10s$. We assume that  the PDF of the  price in spot instances is drawn from the distribution shown  in (\ref{pdf}). We set $\alpha=3 $ and $\theta = 0.983$. The on-demand price $\bar{\pi}$ is 0.35 and the provider's marginal cost of running a spot instance $\underline{\pi}$ is 0.0321.

We use the above parameters to do the numerical evaluations and illustrate the tradeoff of different bidding strategies by comparing the bidding prices, portion of job that runs on on-demand instance, expected total cost and percentage of late job. 

\begin{color}{black}  Recall that $x^{*} = (q^{*}, p^{*})$ is the user's optimal decisions. We use closed-form solutions for our OTR-EG and PR, and convex approximation algorithm for the OTR-P for incomplete job and violating the deadline. we generate the random spot price $\pi_{t}$ for each time slot according to the distribution of spot price we introduced, a one-hour count-down program will be run for each bidding strategy based on the associated optimal solution $x^{*} = (q^{*}, p^{*})$, that is, $q^{*}$ portion of the job will be run on on-demand instance, and the rest will be run on spot instance with bidding price $p^{*}$. Note that after the job starts running, if the random spot price $\pi_{t}$ is higher than the bid price $p^{*}$, the job with one-time request will get interrupted and not resumed. However, the job with PR will get resumed, where a recovery time will be added, when the random spot price $\pi_{t}$ is lower than its bid price $p^{*}$ again. 

When we consider the expected total cost, to be consistent and get insight, besides running cost, penalty for incomplete job and violating the deadline will be added for the strategies OTR-EG and OTR-P; penalty for violating the deadline will be added for PR. We run the simulation for 1000 times, and an average is taken to get the expected total cost for each bidding strategy. 

\end{color}

\begin{color}{black}
\subsection{Comparison among different request mechanisms}
In this subsection, we compare the cost, bidding price in different scenarios for a fixed $t_{e}=3600$s. We vary the deadline $t_{s}$, from 1850s to 8000s in the steps of 50s.

\begin{color}{black}
\begin{figure*}[htpb]
	\centering
	\subfigure[Impact of the deadline on bidding price]{
%
%
\begin{tikzpicture}[scale=.43]

\begin{axis}[%
width=2.3in,
height=1.7in,
at={(0.758in,0.5in)},
scale only axis,
xmin=1850,
xmax=8000,
xlabel style={font=\color{white!15!black}},
xlabel={Deadline (s)},
ymin=0.05,
ymax=0.35,
ylabel style={font=\color{white!15!black}},
ylabel={Bidding Price on Spot Instance (\$)},
axis background/.style={fill=white},
legend style={at={(0.55,0.55)}, anchor=south west, legend cell align=left, align=left, draw=white!15!black}
]
\addplot [color=black, line width=2.0pt]
  table[row sep=crcr]{%
1850	0.110033727524895\\
1900	0.108716128896642\\
1950	0.107459035387309\\
2000	0.106258050796762\\
2050	0.105109208421712\\
2100	0.104008918890404\\
2150	0.102953925597062\\
2200	0.101941266481765\\
2250	0.100968240458705\\
2300	0.100032379747182\\
2350	0.0991314236044884\\
2400	0.0982632977401605\\
2450	0.0977676298078499\\
2500	0.0983545956358372\\
2550	0.098927654139942\\
2600	0.0994874443779059\\
2650	0.100034561185277\\
2700	0.100569559400895\\
2750	0.101092957425199\\
2800	0.101605240419274\\
2850	0.102106863213086\\
2900	0.102598252803134\\
2950	0.103079810627713\\
3000	0.103551914619782\\
3050	0.104014921088761\\
3100	0.104469166294392\\
3150	0.104914968037716\\
3200	0.105352626944112\\
3250	0.105782427729235\\
3300	0.106204640328093\\
3350	0.106619520904368\\
3400	0.107027312791324\\
3450	0.10742824734716\\
3500	0.10782254470773\\
3550	0.10821041457348\\
3600	0.108592056808196\\
3650	0.108781617008749\\
3700	0.108781616\\
3750	0.108781616\\
3800	0.108781616\\
3850	0.108781616\\
3900	0.108781616\\
3950	0.108781616\\
4000	0.108781616\\
4050	0.108781616\\
4100	0.108781616\\
4150	0.108781616\\
4200	0.108781616\\
4250	0.108781616\\
4300	0.108781616\\
4350	0.108781616\\
4400	0.108781616\\
4450	0.108781616\\
4500	0.108781616\\
4550	0.108781616\\
4600	0.108781616\\
4650	0.108781616\\
4700	0.108781616\\
4750	0.108781616\\
4800	0.108781616\\
4850	0.108781616\\
4900	0.108781616\\
4950	0.108781616\\
5000	0.108781616\\
5050	0.108781616\\
5100	0.108781616\\
5150	0.108781616\\
5200	0.108781616\\
5250	0.108781616\\
5300	0.108781616\\
5350	0.108781616\\
5400	0.108781616\\
5450	0.108781616\\
5500	0.108781616\\
5550	0.108781616\\
5600	0.108781616\\
5650	0.108781616\\
5700	0.108781616\\
5750	0.108781616\\
5800	0.108781616\\
5850	0.108781616\\
5900	0.108781616\\
5950	0.108781616\\
6000	0.108781616\\
6050	0.108781616\\
6100	0.108781616\\
6150	0.108781616\\
6200	0.108781616\\
6250	0.108781616\\
6300	0.108781616\\
6350	0.108781616\\
6400	0.108781616\\
6450	0.108781616\\
6500	0.108781616\\
6550	0.108781616\\
6600	0.108781616\\
6650	0.108781616\\
6700	0.108781616\\
6750	0.108781616\\
6800	0.108781616\\
6850	0.108781616\\
6900	0.108781616\\
6950	0.108781616\\
7000	0.108781616\\
7050	0.108781616\\
7100	0.108781616\\
7150	0.108781616\\
7200	0.108781616\\
7250	0.108781616\\
7300	0.108781616\\
7350	0.108781616\\
7400	0.108781616\\
7450	0.108781616\\
7500	0.108781616\\
7550	0.108781616\\
7600	0.108781616\\
7650	0.108781616\\
7700	0.108781616\\
7750	0.108781616\\
7800	0.108781616\\
7850	0.108781616\\
7900	0.108781616\\
7950	0.108781616\\
8000	0.108781616\\
};
\addlegendentry{OTR- EG}

\addplot [color=blue, dashdotted, line width=2.0pt]
  table[row sep=crcr]{%
1850	0.206027381647682\\
1900	0.210133421183954\\
1950	0.214552317296198\\
2000	0.218624346216263\\
2050	0.22307521332696\\
2100	0.227120072649957\\
2150	0.231590683870257\\
2200	0.235616313938895\\
2250	0.240099124932982\\
2300	0.24458908771301\\
2350	0.248604446622568\\
2400	0.253099478045252\\
2450	0.257103295850817\\
2500	0.261601219193367\\
2550	0.266097652487751\\
2600	0.270100808086757\\
2650	0.274602328381144\\
2700	0.279098498347174\\
2750	0.283098940658191\\
2800	0.28759839274533\\
2850	0.291600497787045\\
2900	0.296095959308023\\
2950	0.30110018015566\\
3000	0.304101114665732\\
3050	0.308101053959692\\
3100	0.317099847190676\\
3150	0.317099823597073\\
3200	0.317099800003469\\
3250	0.317600087694378\\
3300	0.318600655902257\\
3350	0.318100349012991\\
3400	0.31810033050848\\
3450	0.317600004035702\\
3500	0.317099658441835\\
3550	0.317599962206364\\
3600	0.318600557863134\\
3650	0.318600557863134\\
3700	0.327108582656051\\
3750	0.327102932957975\\
3800	0.327111217661668\\
3850	0.327102292558315\\
3900	0.327114084761922\\
3950	0.327098748729554\\
4000	0.327057313440273\\
4050	0.327065223376709\\
4100	0.327098483505019\\
4150	0.327101958857826\\
4200	0.327107327251517\\
4250	0.327100878417338\\
4300	0.327074834813659\\
4350	0.327069188202599\\
4400	0.327071325519257\\
4450	0.327127543024274\\
4500	0.327096300987914\\
4550	0.327096484441568\\
4600	0.327100199489121\\
4650	0.327100619563311\\
4700	0.327102534538227\\
4750	0.327123925675868\\
4800	0.327101638330397\\
4850	0.327099724158849\\
4900	0.3270993416019\\
4950	0.327098638794492\\
5000	0.327100326039238\\
5050	0.327100265181619\\
5100	0.327099713955852\\
5150	0.327099072326287\\
5200	0.327098737290889\\
5250	0.327098832262468\\
5300	0.327097809933924\\
5350	0.327096004126339\\
5400	0.327093527603332\\
5450	0.327091949208085\\
5500	0.327090969693642\\
5550	0.327095608205646\\
5600	0.327090317238394\\
5650	0.327086127112208\\
5700	0.327087636305296\\
5750	0.327089948715351\\
5800	0.327084309877343\\
5850	0.327127696717587\\
5900	0.327115825504987\\
5950	0.327101239183121\\
6000	0.327122931867526\\
6050	0.327121893336132\\
6100	0.327106332885454\\
6150	0.327085920629351\\
6200	0.327086210903338\\
6250	0.327098707703058\\
6300	0.327101783116659\\
6350	0.32709992940653\\
6400	0.327096735834179\\
6450	0.327095513768334\\
6500	0.327099526196088\\
6550	0.327100045144128\\
6600	0.327095294886858\\
6650	0.327101117323021\\
6700	0.327107776897216\\
6750	0.327111742846533\\
6800	0.327113873670725\\
6850	0.327118430390703\\
6900	0.327129148212503\\
6950	0.327127700228331\\
7000	0.327088079662897\\
7050	0.327150455859752\\
7100	0.327142034212683\\
7150	0.327098918778377\\
7200	0.327052828655357\\
7250	0.327137858666183\\
7300	0.327096341743611\\
7350	0.327101597636983\\
7400	0.327113052403095\\
7450	0.327103508218333\\
7500	0.327102552165704\\
7550	0.327101311011278\\
7600	0.327100146562062\\
7650	0.327099243068927\\
7700	0.327098507292256\\
7750	0.327096767671305\\
7800	0.327094655874246\\
7850	0.327093850931156\\
7900	0.327095191640586\\
7950	0.32710198328145\\
8000	0.327099881201927\\
};
\addlegendentry{OTR-P}

\addplot [color=red, dashed, line width=2.0pt]
  table[row sep=crcr]{%
1850	0.35\\
1900	0.35\\
1950	0.35\\
2000	0.35\\
2050	0.35\\
2100	0.35\\
2150	0.35\\
2200	0.35\\
2250	0.35\\
2300	0.35\\
2350	0.35\\
2400	0.35\\
2450	0.35\\
2500	0.35\\
2550	0.35\\
2600	0.35\\
2650	0.35\\
2700	0.35\\
2750	0.35\\
2800	0.35\\
2850	0.35\\
2900	0.35\\
2950	0.35\\
3000	0.35\\
3050	0.35\\
3100	0.35\\
3150	0.35\\
3200	0.35\\
3250	0.35\\
3300	0.35\\
3350	0.35\\
3400	0.35\\
3450	0.35\\
3500	0.35\\
3550	0.35\\
3600	0.18528346191978\\
3650	0.14930586447781\\
3700	0.134417707633972\\
3750	0.128465226221085\\
3800	0.121595855116844\\
3850	0.116191871023178\\
3900	0.111738940834999\\
3950	0.107955833637714\\
4000	0.104671218752861\\
4050	0.101772694766521\\
4100	0.0991823083043098\\
4150	0.0968438207626343\\
4200	0.094715091085434\\
4250	0.0927639071345329\\
4300	0.0909648781657219\\
4350	0.0892977673768997\\
4400	0.0877460897326469\\
4450	0.0862962782382965\\
4500	0.0849370396971703\\
4550	0.0836587862610817\\
4600	0.0824534459471703\\
4650	0.0813140268266201\\
4700	0.0802345033347607\\
4750	0.0792096457362175\\
4800	0.0782348306417465\\
4850	0.077306041008234\\
4900	0.0764196577072143\\
4950	0.0755725163698196\\
5000	0.0747617368519306\\
5050	0.0739847990274429\\
5100	0.0732393533051014\\
5150	0.072523353266716\\
5200	0.0718348851323128\\
5250	0.0711722435534\\
5300	0.0705338558197022\\
5350	0.0699182818591595\\
5400	0.0693242142379284\\
5450	0.0687504402637482\\
5500	0.0681958609342575\\
5550	0.0676594151437283\\
5600	0.0671401744246483\\
5650	0.0666372671544552\\
5700	0.0661498596072197\\
5750	0.0656771749019623\\
5800	0.0652185498476029\\
5850	0.0647732833564282\\
5900	0.0643407690823078\\
5950	0.0639204196274281\\
6000	0.0635117139130831\\
6050	0.0631141308605671\\
6100	0.0627271872878075\\
6150	0.0623504568576813\\
6200	0.061983494284749\\
6250	0.061625911128521\\
6300	0.0612773189485073\\
6350	0.0609373766750097\\
6400	0.0606057432383299\\
6450	0.060282096517086\\
6500	0.0599661522865295\\
6550	0.0596576073735952\\
6600	0.0593561965018511\\
6650	0.059061654394865\\
6700	0.0587737536728382\\
6750	0.0584922574818134\\
6800	0.0582169384419918\\
6850	0.0579475786477327\\
6900	0.0576839886158705\\
6950	0.0574259788632393\\
7000	0.0571733599066734\\
7050	0.0569259517371654\\
7100	0.0566836027681827\\
7150	0.0564461424648762\\
7200	0.0562134192407131\\
7250	0.0559853004574776\\
7300	0.0557616155803203\\
7350	0.0555422603935003\\
7400	0.0553270927846432\\
7450	0.0551159990638495\\
7500	0.0549088465929031\\
7550	0.0547055311560631\\
7600	0.0545059390634298\\
7650	0.0543099660992622\\
7700	0.0541175080478191\\
7750	0.0539284796416759\\
7800	0.0537427766650915\\
7850	0.0535603138506413\\
7900	0.0533810059309006\\
7950	0.0532047581642866\\
8000	0.0530315042316914\\
};
\addlegendentry{PR}

\end{axis}
\end{tikzpicture}%
		\label{fig:1.1}
	}
	\subfigure[Impact of the deadline on portion of job run on on-demand instance]{
%
%
\begin{tikzpicture}[scale=.43]

\begin{axis}[%
width=2.3in,
height=1.7in,
at={(0.758in,0.58in)},
scale only axis,
xmin=1800,
xmax=3800,
xlabel style={font=\color{white!15!black}},
xlabel={Deadline (s)},
ymin=0,
ymax=0.5,
ylabel style={font=\color{white!15!black}},
ylabel={Portion of job run on on-demand instance},
axis background/.style={fill=white},
legend style={legend cell align=left, align=left, draw=white!15!black}
]
\addplot [color=black, line width=2.0pt]
  table[row sep=crcr]{%
1850	0.493314950579016\\
1900	0.479823165522218\\
1950	0.466333141017925\\
2000	0.452844564875549\\
2050	0.439357156503412\\
2100	0.425870663599268\\
2150	0.412384859176516\\
2200	0.398899538878259\\
2250	0.385414518822423\\
2300	0.371929633098531\\
2350	0.358444732459164\\
2400	0.344959682160414\\
2450	0.33130801635635\\
2500	0.317138743214011\\
2550	0.302982660529829\\
2600	0.288838842898803\\
2650	0.274706450900876\\
2700	0.260584721147437\\
2750	0.246472957768494\\
2800	0.232370525014784\\
2850	0.218276840797884\\
2900	0.204191371093089\\
2950	0.190113625021147\\
3000	0.17604315052405\\
3050	0.161979530542977\\
3100	0.147922379688499\\
3150	0.133871341227733\\
3200	0.119826084467606\\
3250	0.105786302387889\\
3300	0.0917517095376705\\
3350	0.0777220401608878\\
3400	0.0636970465094957\\
3450	0.0496764973303107\\
3500	0.0356601765131922\\
3550	0.0216478818410318\\
3600	0.00763942390085437\\
3650	0\\
3700	0\\
3750	0\\
3800	0\\
};
\addlegendentry{ORT- EG}

\addplot [color=blue, dashdotted, line width=2.0pt]
  table[row sep=crcr]{%
1850	0.486111111111111\\
1900	0.472222222222222\\
1950	0.458333333333333\\
2000	0.444444444444444\\
2050	0.430555555555556\\
2100	0.416666666666667\\
2150	0.402777777777778\\
2200	0.388888888888889\\
2250	0.375\\
2300	0.361111111111111\\
2350	0.347222222222222\\
2400	0.333333333333333\\
2450	0.319444444444444\\
2500	0.305555555555556\\
2550	0.291666666666667\\
2600	0.277777777777778\\
2650	0.263888888888889\\
2700	0.25\\
2750	0.236111111111111\\
2800	0.222222222222222\\
2850	0.208333333333333\\
2900	0.194444444444444\\
2950	0.180555555555556\\
3000	0.166666666666667\\
3050	0.152777777777778\\
3100	0.138888888888889\\
3150	0.125\\
3200	0.111111111111111\\
3250	0.0972222222222222\\
3300	0.0833333333333334\\
3350	0.0694444444444444\\
3400	0.0555555555555556\\
3450	0.0416666666666666\\
3500	0.0277777777777778\\
3550	0.0138888888888888\\
3600	0\\
3650	0\\
3700	0\\
3750	0\\
3800	0\\
};
\addlegendentry{OTR- P}
\addplot [color=red, dashed, line width=2.0pt]
  table[row sep=crcr]{%
1850	0.486111111111111\\
1900	0.472222222222222\\
1950	0.458333333333333\\
2000	0.444444444444444\\
2050	0.430555555555556\\
2100	0.416666666666667\\
2150	0.402777777777778\\
2200	0.388888888888889\\
2250	0.375\\
2300	0.361111111111111\\
2350	0.347222222222222\\
2400	0.333333333333333\\
2450	0.319444444444444\\
2500	0.305555555555556\\
2550	0.291666666666667\\
2600	0.277777777777778\\
2650	0.263888888888889\\
2700	0.25\\
2750	0.236111111111111\\
2800	0.222222222222222\\
2850	0.208333333333333\\
2900	0.194444444444444\\
2950	0.180555555555556\\
3000	0.166666666666667\\
3050	0.152777777777778\\
3100	0.138888888888889\\
3150	0.125\\
3200	0.111111111111111\\
3250	0.0972222222222222\\
3300	0.0833333333333334\\
3350	0.0694444444444444\\
3400	0.0555555555555556\\
3450	0.0416666666666666\\
3500	0.0277777777777778\\
3550	0.0138888888888888\\
3600	0\\
3650	0\\
3700	0\\
3750	0\\
3800	0\\
};
\addlegendentry{PR}
\end{axis}
\end{tikzpicture}%
		\label{fig:1.2}
	}
	\subfigure[Impact of the deadline on ratio of job portion run on on-demand instance compared to OTR-EG]{
%
%
\begin{tikzpicture}[scale=.43]

\begin{axis}[%
width=2.3in,
height=1.7in,
at={(0.874in,0.58in)},
scale only axis,
unbounded coords=jump,
xmin=1800,
xmax=3800,
xlabel style={font=\color{white!15!black}},
xlabel={Deadline (s)},
ymin=0,
ymax=1.1,
ylabel style={font=\color{white!15!black}, align=center},
ylabel={Ratio of job portion run on on-demand instance\\[1ex]w.r.t. OTR-EG},
axis background/.style={fill=white},
legend style={at={(0.1,0.1)}, anchor=south west, legend cell align=left, align=left, draw=white!15!black}
]
\addplot [color=blue, dashdotted, line width=2.0pt]
  table[row sep=crcr]{%
1850	0.985397078561172\\
1900	0.984158865502621\\
1950	0.98284529453101\\
2000	0.981450322952616\\
2050	0.97996709324618\\
2100	0.978387811795221\\
2150	0.976703602993761\\
2200	0.974904333011963\\
2250	0.972978395172442\\
2300	0.970912449494032\\
2350	0.968691100131538\\
2400	0.966296499479977\\
2450	0.964191714880979\\
2500	0.963475961526912\\
2550	0.962651348287146\\
2600	0.961705063591809\\
2650	0.96062137610343\\
2700	0.95938088349605\\
2750	0.957959498878918\\
2800	0.956327065182143\\
2850	0.954445430728229\\
2900	0.952265727016442\\
2950	0.949724437348832\\
3000	0.946737582067401\\
3050	0.943191879033396\\
3100	0.938930871592025\\
3150	0.933732334744883\\
3200	0.927269814454708\\
3250	0.919043581519045\\
3300	0.908248290448683\\
3350	0.893497446807773\\
3400	0.872184168653308\\
3450	0.838760156329364\\
3500	0.778957943954698\\
3550	0.641581887358771\\
3600	0\\
3650	nan\\
3700	nan\\
3750	nan\\
3800	nan\\
};
\addlegendentry{OTR-P/OTR-EG}

\addplot [color=red, dashed, line width=2.0pt]
  table[row sep=crcr]{%
1850	0.985397078561172\\
1900	0.984158865502621\\
1950	0.98284529453101\\
2000	0.981450322952616\\
2050	0.97996709324618\\
2100	0.978387811795221\\
2150	0.976703602993761\\
2200	0.974904333011963\\
2250	0.972978395172442\\
2300	0.970912449494032\\
2350	0.968691100131538\\
2400	0.966296499479977\\
2450	0.964191714880979\\
2500	0.963475961526912\\
2550	0.962651348287146\\
2600	0.961705063591809\\
2650	0.96062137610343\\
2700	0.95938088349605\\
2750	0.957959498878918\\
2800	0.956327065182143\\
2850	0.954445430728229\\
2900	0.952265727016442\\
2950	0.949724437348832\\
3000	0.946737582067401\\
3050	0.943191879033396\\
3100	0.938930871592025\\
3150	0.933732334744883\\
3200	0.927269814454708\\
3250	0.919043581519045\\
3300	0.908248290448683\\
3350	0.893497446807773\\
3400	0.872184168653308\\
3450	0.838760156329364\\
3500	0.778957943954698\\
3550	0.641581887358771\\
3600	0\\
3650	nan\\
3700	nan\\
3750	nan\\
3800	nan\\
};
\addlegendentry{PR/OTR-EG}

\end{axis}
\end{tikzpicture}%
		\label{fig:1.3}
	}
	\subfigure[Impact of the deadline on expected total cost]{
		\input{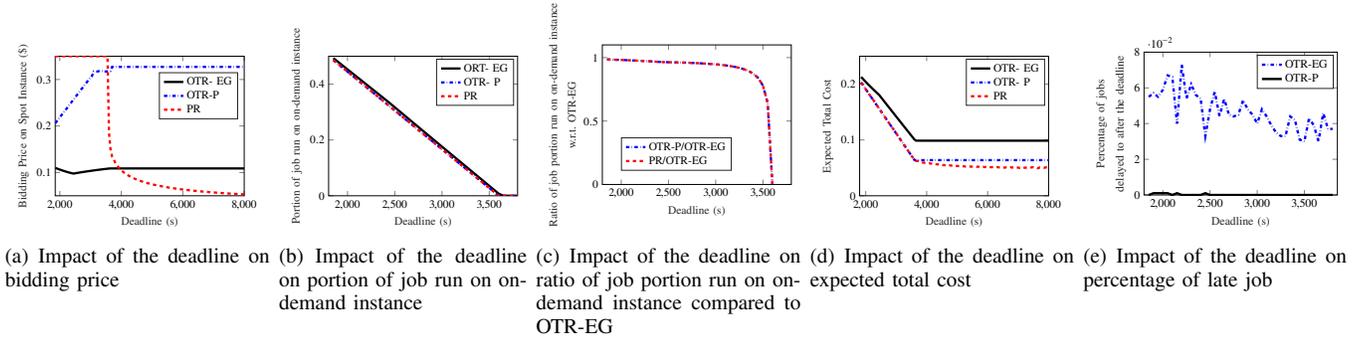}
		\label{fig:1.4}
	}
	\subfigure[Impact of the deadline on percentage of late job]{
%
%
\begin{tikzpicture}[scale=.44]

\begin{axis}[%
width=2.3in,
height=1.7in,
at={(0.999in,0.58in)},
scale only axis,
xmin=1800,
xmax=3850,
xlabel style={font=\color{white!15!black}},
xlabel={Deadline (s)},
ymin=0,
ymax=0.08,
ylabel style={font=\color{white!15!black}, align=center},
ylabel={Percentage of jobs\\[1ex]delayed to after the deadline},
axis background/.style={fill=white},
legend style={legend cell align=left, align=left, draw=white!15!black}
]
\addplot [color=blue, dashdotted, line width=2.0pt]
  table[row sep=crcr]{%
1850	0.055\\
1900	0.057\\
1950	0.055\\
2000	0.059\\
2050	0.067\\
2100	0.066\\
2150	0.04\\
2200	0.073\\
2250	0.054\\
2300	0.062\\
2350	0.056\\
2400	0.054\\
2450	0.033\\
2500	0.048\\
2550	0.058\\
2600	0.044\\
2650	0.054\\
2700	0.045\\
2750	0.044\\
2800	0.045\\
2850	0.053\\
2900	0.048\\
2950	0.045\\
3000	0.04\\
3050	0.048\\
3100	0.044\\
3150	0.039\\
3200	0.035\\
3250	0.033\\
3300	0.033\\
3350	0.044\\
3400	0.042\\
3450	0.031\\
3500	0.03\\
3550	0.043\\
3600	0.039\\
3650	0.03\\
3700	0.046\\
3750	0.037\\
3800	0.037\\
};
\addlegendentry{OTR-EG}

\addplot [color=black, line width=2.0pt]
  table[row sep=crcr]{%
1850	0\\
1900	0.001\\
1950	0.001\\
2000	0.001\\
2050	0.001\\
2100	0\\
2150	0.001\\
2200	0\\
2250	0\\
2300	0\\
2350	0\\
2400	0\\
2450	0.001\\
2500	0\\
2550	0\\
2600	0\\
2650	0\\
2700	0\\
2750	0\\
2800	0\\
2850	0\\
2900	0\\
2950	0\\
3000	0\\
3050	0\\
3100	0\\
3150	0\\
3200	0\\
3250	0\\
3300	0\\
3350	0\\
3400	0\\
3450	0\\
3500	0\\
3550	0\\
3600	0\\
3650	0\\
3700	0\\
3750	0\\
3800	0\\
};
\addlegendentry{OTR-P}

\end{axis}
\end{tikzpicture}%
		\label{fig:1.5}
	}
	\caption{\small Impact of Deadline with $t_{e}=3600s$, $c_{I}=\bar{\pi}/3$, $c_{s}=\bar{\pi}/10$, and $t_{r}=10s$ }
	\label{fig:1}
\end{figure*}
\end{color}

 \begin{color}{black}  From Fig. \ref{fig:1.1}, we can see that as the deadline increases, the bidding price in OTR-EG decreases first, and then increases. From Proposition 1 (cf. (10)) we can see that the optimal bid price is determined by the minimum of two terms: the first term (cf. (11)) is to get the trade-off between the bid price and portion of job running on on-demand instance, and the second term is to guarantee that the job can continue running without interruption, and finish before the deadline (cf. (12)). Thus, the bid price decreases for a while since lower price can lead to lower total cost while the job will not get any interruption. After that we see the bid price increases with the deadline. The reason is that with a longer deadline, the portion of job that runs on spot instance becomes bigger (see Fig.~\ref{fig:1.2}). In order to guarantee that the job running on spot instance can continue running without any interruption, the bid price should be higher. 
 
 \begin{color}{black}
Fig. \ref{fig:1.1} also shows that when the deadline $t_{s}$ is smaller than the execution time $t_{e}$, the user needs to bid with the upperbound of the spot price $\bar{\pi}$ for PR, and bid lowest for OTR-EG compared to PR and OTR-P. Note that with PR, if a job is interrupted, there will be a recovery time $t_{r}$ before it get resumed. The user not only needs to pay for it, he also needs to allocate more $t_{r}$ amount of job to on-demand instance in order to finish the job. That is to say, although we do not put any penalty on our PR model, $(\pi+\bar{\pi}) t_{r}$ will be added to the total cost for each interruption, which is even larger than the cost to run $t_{r}$ on on-demand instance, i.e., $\bar{\pi} t_{r}$. Thus the users needs to bid the upper-bound of the spot price. In terms of OTR-P, the user needs to find a trade-off between the penalty and bid price, thus the bid price is higher than that of OTR-EG but lower than that of PR.\end{color} When the deadline $t_{s}$ is larger than the execution time $t_{e}$, the user's optimal strategy is to rely on spot instance: there will be no job running on the on-demand instance (see Figure 3b) and the optimal bid price will not be impacted by the increase of deadline (cf. (14)). 

Another interesting observation from Fig. \ref{fig:1.1} is that the bidding price is higher in PR for smaller deadline. It decreases and gets lower than that of OTR-EG and OTR-P when the deadline is longer than the execution time. This is because unlike one-time request, in PR, the interrupted job can get resumed in the spot instances and a recovery time will be included, which will induce more cost when the execution time is shorter than the deadline. However, when the execution time is longer than the deadline, with the increase of deadline, smaller bidding price can save the cost while guaranteeing the job can be completed before the deadline. \end{color}

\begin{color}{black}  Figure 3b shows that, as we would expect from Lemma 8, when the deadline is smaller than the execution time, the differences among the portions of job that runs on on-demand instance by using OTR-EG and PR are minimal, which is not larger than $\frac{t_{k}}{t_{e}} = \frac{300}{3600} \approx 0.0833$. We can also see that the user will put the same portion of the job $\frac{t_{s}}{t_{e}}$ running on the spot instance by using OTR-P and PR. The intuition is that because spot instance price is not higher than that of on-demand instance, the user tends to run as much job as possible on the spot instance, which is $\frac{t_{s}}{t_{e}}$. \end{color}

\begin{color}{black} In order to show the difference more clearly, we plot Figure 3c to show the ratio of portion of job runs on on-demand instance compared to OTR-EG. The comparison shows that when the execution time is smaller than the deadline, the user put less job on on-demand instance by using OTR-P and PR compared to OTR-EG, and the portion of the job put on on-demand instance decreases much faster for OTR-P and PR compared to OTR-EG. 
\end{color}

\begin{color}{black} 

Figure 3d shows that when the deadline $t_{s}$ is shorter than the execution time $t_{e}$, the expected total costs obtained with different strategies are decreasing. This is intuitive, as the deadline increases, the portions of the job will be run on on-demand instance, whose price is not less than that of spot instance, are getting smaller (which can be verified in Figure 3b),  thus the expected total costs are becoming lower.

However, when the deadline $t_{s}$ is longer than the execution time $t_{e}$, the user's optimal strategy is to rely only on spot instance to finish the job and there will be no job running on on-demand instance (see Figure 3b, Proposition 2 and Proposition 4). Note that the bidding price obtained from PR decreases with the deadline and gets smaller than the bid prices in OTR-EG and OTR-P, which are not impacted by the increase of the deadline (see Figure 3a). Therefore, the expected total costs obtained from PR is decreasing while that of OTR-EG and OTR-P do not change. Another interesting observation is that the expected total cost of PR is lower than that of the other two bidding strategies, indicating that users can further lower the total running cost by using PR.
  \end{color}

\begin{color}{black}
In Figure 3e, we plot the percentage of jobs that are delayed to after the deadline for OTR-EG and OTR-P. We can see that with OTR-EG, the percentage of late jobs are decreasing with the increase of deadline and there are always some job being delayed after deadline. However, for OTR-P, almost all the jobs can be finished before deadline. Thus, adding penalty can reduce the fraction of jobs which are delayed.
\end{color}

\end{color}

\subsection{OTR-P: Impact of the Penalty Parameters $c_s$, and $c_I$}

Recall that the penalty corresponding to the incomplete job in the spot instance is $c_I$ and the penalty corresponding to the portion of the job completed after the deadline is $c_s$. 
We, now, evaluate the impact of $c_I$ and $c_s$ on the bidding prices and the expected total cost (including penalty). 

We assume that the execution time is 3600s. We consider three different scenarios of the deadlines: i) 2100s,  ii) 2700s and iii) 3300s. 

\subsubsection{Impact of penalty parameter $c_{s}$}
In order to see the impact of penalty coefficient for late but completed job $c_{s}$ on the bidding price and expected total cost, we fix the execution time without any interruptions $t_{e}$ as 3600s and penalty coefficient for incomplete job $c_{I}$ as $\bar{\pi}/3$ , and change $c_{s}$ from 0 to $c_{I}=\bar{\pi}/3$ in the step of 0.005.

We observe for Figure \ref{fig3} that both the bid prices and expected total cost  increase with the increase in the penalty coefficient $c_{s}$. The higher bidding price is due to the fact that there is a penalty due to the completed but late job. However, compared to the bid price, the expected total cost is relatively less sensitive to the change of penalty coefficient $c_{s}$ (see Figure \ref{fig:3.2}). Finally, from Figure \ref{fig:3.2} we also observe that  the total cost decreases with the deadline. This is because a smaller portion of the job that runs  on the on-demand instance when the deadline is higher. 
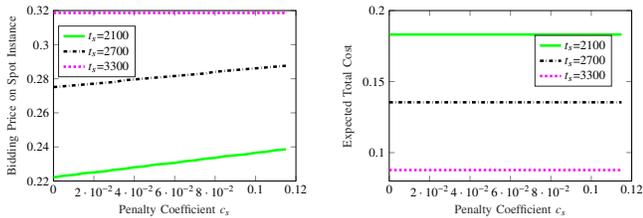
\begin{figure}[htbp]
	\centering
	\subfigure[Impact of the Penalty Coefficient $c_{s}$ on Bidding Price]{
%
%
\definecolor{mycolor1}{rgb}{1.00000,0.00000,1.00000}%
\begin{tikzpicture}[scale=.47]

\begin{axis}[%
width=2.7in,
height=1.9in,
at={(0.758in,0.559in)},
scale only axis,
xmin=0,
xmax=0.12,
xlabel style={font=\color{white!15!black}},
xlabel={Penalty Coefficient $c_{s}$},
ymin=0.22,
ymax=0.32,
ylabel style={font=\color{white!15!black}},
ylabel={Bidding Price on Spot Instance},
axis background/.style={fill=white},
legend style={at={(0.02,0.61)}, anchor=south west, legend cell align=left, align=left, draw=white!15!black}
]
\addplot [color=green, line width=2.0pt]
  table[row sep=crcr]{%
0	0.222131849087204\\
0.005	0.223074645749137\\
0.01	0.223614865115977\\
0.015	0.224563681640495\\
0.02	0.22510043268563\\
0.025	0.225634902012781\\
0.03	0.226589647357344\\
0.035	0.227120072649957\\
0.04	0.228079424536832\\
0.045	0.228608922933981\\
0.05	0.229571746130589\\
0.055	0.230098008887377\\
0.06	0.230625706464554\\
0.065	0.231590748697199\\
0.07	0.232113897106943\\
0.075	0.233084286035511\\
0.08	0.233606757206276\\
0.085	0.234579476458951\\
0.09	0.235098799593957\\
0.095	0.235615094028376\\
0.1	0.236592856110018\\
0.105	0.237109203238159\\
0.11	0.238088522845883\\
0.115	0.238603359251864\\
};
\addlegendentry{$t_{s}$=2100}

\addplot [color=black, dashdotted, line width=2.0pt]
  table[row sep=crcr]{%
0	0.275102398881797\\
0.005	0.275599893145128\\
0.01	0.276099970578959\\
0.015	0.276600018103093\\
0.02	0.277100026873178\\
0.025	0.2776000089603\\
0.03	0.278099974448443\\
0.035	0.279098498347174\\
0.04	0.279598425338177\\
0.045	0.280098272889781\\
0.05	0.280598790549048\\
0.055	0.281098861011519\\
0.06	0.281598931094829\\
0.065	0.282098991159489\\
0.07	0.282599050502484\\
0.075	0.283099091681255\\
0.08	0.284098429577433\\
0.085	0.284598546158523\\
0.09	0.285098693644105\\
0.095	0.285598835359443\\
0.1	0.28609898610122\\
0.105	0.286599101397444\\
0.11	0.287099208445807\\
0.115	0.2875985799346\\
};
\addlegendentry{$t_{s}$=2700}

\addplot [color=mycolor1, dotted, line width=2.0pt]
  table[row sep=crcr]{%
0	0.318600642430586\\
0.005	0.318600644355111\\
0.01	0.318600646279635\\
0.015	0.318600648204159\\
0.02	0.318600650128684\\
0.025	0.318600652053208\\
0.03	0.318600653977732\\
0.035	0.318600655902257\\
0.04	0.318600657826781\\
0.045	0.318600659751306\\
0.05	0.31860066167583\\
0.055	0.318600663600354\\
0.06	0.318600665524879\\
0.065	0.318600667449403\\
0.07	0.318600669373927\\
0.075	0.318600671298452\\
0.08	0.318600673222976\\
0.085	0.318600675147501\\
0.09	0.318600677072025\\
0.095	0.318600678996549\\
0.1	0.318600680921074\\
0.105	0.318600682845598\\
0.11	0.318600684770122\\
0.115	0.318600686694647\\
};
\addlegendentry{$t_{s}$=3300}

\end{axis}
\end{tikzpicture}%
		\label{fig:3.1}
	}
	\subfigure[Impact of the Penalty Coefficient $c_{s}$ on Expected Total Cost]{
%
%
\definecolor{mycolor1}{rgb}{1.00000,0.00000,1.00000}%
\begin{tikzpicture}[scale=.47]

\begin{axis}[%
width=2.7in,
height=1.9in,
at={(0.74in,0.545in)},
scale only axis,
xmin=0,
xmax=0.12,
xlabel style={font=\color{white!15!black}},
xlabel={Penalty Coefficient $c_{s}$},
ymin=0.08,
ymax=0.2,
ylabel style={font=\color{white!15!black}},
ylabel={Expected Total Cost},
axis background/.style={fill=white},
legend style={at={(0.6,0.555)}, anchor=south west, legend cell align=left, align=left, draw=white!15!black}
]
\addplot [color=green, line width=2.0pt]
  table[row sep=crcr]{%
0	0.183100462090976\\
0.005	0.183100572934571\\
0.01	0.183100678319433\\
0.015	0.183100778664401\\
0.02	0.183100873760414\\
0.025	0.183100964390287\\
0.03	0.18310105019892\\
0.035	0.183101131833979\\
0.04	0.183101209284461\\
0.045	0.183101282717072\\
0.05	0.183101352543117\\
0.055	0.183101418506512\\
0.06	0.183101481263819\\
0.065	0.183101540552832\\
0.07	0.183101596820212\\
0.075	0.183101650100947\\
0.08	0.183101700474724\\
0.085	0.183101748296989\\
0.09	0.183101793328326\\
0.095	0.183101836092963\\
0.1	0.183101876388651\\
0.105	0.183101914537778\\
0.11	0.183101950581949\\
0.115	0.183101984561689\\
};
\addlegendentry{$t_{s}$=2100}

\addplot [color=black, dashdotted, line width=2.0pt]
  table[row sep=crcr]{%
0	0.135417487259426\\
0.005	0.135417488171346\\
0.01	0.135417489024792\\
0.015	0.135417489822958\\
0.02	0.135417490568794\\
0.025	0.135417491265707\\
0.03	0.135417491916561\\
0.035	0.135417492521224\\
0.04	0.135417493083993\\
0.045	0.135417493608626\\
0.05	0.135417494097907\\
0.055	0.135417494553231\\
0.06	0.135417494976924\\
0.065	0.135417495371125\\
0.07	0.135417495737911\\
0.075	0.135417496078147\\
0.08	0.135417496392482\\
0.085	0.135417496683524\\
0.09	0.135417496953574\\
0.095	0.135417497204114\\
0.1	0.135417497435718\\
0.105	0.135417497650398\\
0.11	0.135417497848532\\
0.115	0.135417498031801\\
};
\addlegendentry{$t_{s}$=2700}

\addplot [color=mycolor1, dotted, line width=2.0pt]
  table[row sep=crcr]{%
0	0.0877324997885286\\
0.005	0.0877324997889014\\
0.01	0.0877324997892744\\
0.015	0.0877324997896472\\
0.02	0.0877324997900203\\
0.025	0.087732499790393\\
0.03	0.0877324997907661\\
0.035	0.0877324997911389\\
0.04	0.0877324997915119\\
0.045	0.0877324997918847\\
0.05	0.0877324997922575\\
0.055	0.0877324997926306\\
0.06	0.0877324997930033\\
0.065	0.0877324997933764\\
0.07	0.0877324997937492\\
0.075	0.0877324997941222\\
0.08	0.087732499794495\\
0.085	0.0877324997948678\\
0.09	0.0877324997952408\\
0.095	0.0877324997956136\\
0.1	0.0877324997959867\\
0.105	0.0877324997963594\\
0.11	0.0877324997967322\\
0.115	0.0877324997971053\\
};
\addlegendentry{$t_{s}$=3300}

\end{axis}
\end{tikzpicture}%
		\label{fig:3.2}
	}	
	\caption{\small Impact of Penalty Coefficient $c_{s}$ with $t_{e}=3600s$ and $c_{I}=\bar{\pi}/3$}
	\label{fig3}
\end{figure}

\subsubsection{Impact of penalty parameter $c_{I}$}

In order to investigate the impact of penalty coefficient for unfinished job $c_{I}$ on the bidding price and expected total cost, we fix the execution time without any interruptions $t_{e}$ as 3600s and penalty coefficient for late completed job $c_{s}$ as $\bar{\pi}/10$ , and vary $c_{I}$ from $c_{s}=\bar{\pi}/10$ to $\bar{\pi}/2$ in the step of $0.005$. 

We note from Fig.~\ref{fig:311} that the optimal bidding prices on spot instance increase with the increase in the penalty coefficient $c_{I}$. That is intuitive, since a larger penalty coefficient for unfinished job implies that a higher penalty for the unfinished job. Hence, the user has to bid a higher price for spot instance.  However, very small or very large $c_{I}$ will lead to a slower increase of bidding price. When $c_I$ is small and below a certain threshold $T_{1}$, the bidding price does not increase much because of the lower penalty. On the other hand, when $c_I$ exceeds a threshold $T_{2}$, the bid price becomes closer to the upper bound. Thus, increase in $c_I$ does not increase the bid price as rapidly as $c_I$ in between of $T_{1}$ and $T_{2}$. Fig.~\ref{fig:311} also suggests that with longer deadline, the bid price is higher. This is because more portion of job will be run on spot instance with longer deadline, then in order to avoid the penalty for incomplete job or completed but late job, higher price needs to be set. 

We notice that, in Fig.~\ref{fig:312}, the expected total cost is increasing with the penalty coefficient $c_{I}$ for different deadlines. The total cost decreases with the deadline. Intuitively, longer deadline will introduce less portion of job running on the on-demand instance, the price of which is much higher than the spot price. 

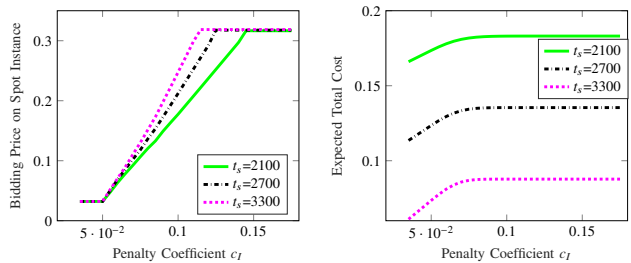
\begin{figure}[htbp]
	\centering
	\subfigure[Impact of the Penalty Coefficient $c_{I}$ on Bidding Price]{
%
%
\definecolor{mycolor1}{rgb}{1.00000,0.00000,1.00000}%
\begin{tikzpicture}[scale=.55]

\begin{axis}[%
width=2.3in,
height=2in,
at={(0.758in,0.559in)},
scale only axis,
xmin=0.02,
xmax=0.18,
xlabel style={font=\color{white!15!black}},
xlabel={Penalty Coefficient $c_{I}$},
ymin=0,
ymax=0.35,
ylabel style={font=\color{white!15!black}},
ylabel={Bidding Price on Spot Instance},
axis background/.style={fill=white},
legend style={at={(0.58,0.03)}, anchor=south west, legend cell align=left, align=left, draw=white!15!black}
]
\addplot [color=green, line width=2.0pt]
  table[row sep=crcr]{%
0.035	0.0321000002039577\\
0.04	0.0321000009983002\\
0.045	0.0321000016188483\\
0.05	0.0321000048747399\\
0.055	0.0474107270492772\\
0.06	0.0636124081265234\\
0.065	0.0782622262061372\\
0.07	0.0921499301878548\\
0.075	0.106610908281389\\
0.08	0.120633573315623\\
0.085	0.13295247114423\\
0.09	0.14910620508288\\
0.095	0.16348453637121\\
0.1	0.177444868408388\\
0.105	0.192354570602285\\
0.11	0.207192795273196\\
0.115	0.222129972007098\\
0.12	0.237109653830079\\
0.125	0.252103066947768\\
0.13	0.267103388713411\\
0.135	0.282098947864243\\
0.14	0.297097853593362\\
0.145	0.316600062234687\\
0.15	0.316600062234687\\
0.155	0.316600062234687\\
0.16	0.316600062234687\\
0.165	0.316600062234687\\
0.17	0.316600062234687\\
0.175	0.316600062234687\\
};
\addlegendentry{$t_{s}$=2100}

\addplot [color=black, dashdotted, line width=2.0pt]
  table[row sep=crcr]{%
0.035	0.0321000002039575\\
0.04	0.0321000009982986\\
0.045	0.0321000016189535\\
0.05	0.0321000049000936\\
0.055	0.050159663072432\\
0.06	0.0690350943528282\\
0.065	0.0856144999767723\\
0.07	0.102147665977322\\
0.075	0.118020461872036\\
0.08	0.136616559241731\\
0.085	0.153522601822533\\
0.09	0.172909861559047\\
0.095	0.19252440931945\\
0.1	0.212211893623951\\
0.105	0.232120476455319\\
0.11	0.252103977922087\\
0.115	0.272100838085114\\
0.12	0.292100266115519\\
0.125	0.317619385381683\\
0.13	0.317620945489366\\
0.135	0.317622505597174\\
0.14	0.317624065705109\\
0.145	0.31762562581317\\
0.15	0.317627185921357\\
0.155	0.31762874602967\\
0.16	0.31763030613811\\
0.165	0.317631866246675\\
0.17	0.317633426355367\\
0.175	0.317634986464184\\
};
\addlegendentry{$t_{s}$=2700}

\addplot [color=mycolor1, dotted, line width=2.0pt]
  table[row sep=crcr]{%
0.035	0.0321000002039575\\
0.04	0.0321000009982986\\
0.045	0.0321000016189536\\
0.05	0.0321000049001211\\
0.055	0.0517259644121438\\
0.06	0.0709012373250775\\
0.065	0.0894543092968436\\
0.07	0.110207746755737\\
0.075	0.130192328334349\\
0.08	0.151700248730894\\
0.085	0.173095752696228\\
0.09	0.197974225729643\\
0.095	0.222634147708724\\
0.1	0.247605846938721\\
0.105	0.272598871858734\\
0.11	0.298098880281857\\
0.115	0.318600655902257\\
0.12	0.318600655902257\\
0.125	0.318600655902257\\
0.13	0.318600655902257\\
0.135	0.318600655902257\\
0.14	0.318600655902257\\
0.145	0.318600655902257\\
0.15	0.318600655902257\\
0.155	0.318600655902257\\
0.16	0.318600655902257\\
0.165	0.318600655902257\\
0.17	0.318600655902257\\
0.175	0.318600655902257\\
};
\addlegendentry{$t_{s}$=3300}

\end{axis}
\end{tikzpicture}%
		\label{fig:311}
	}
	\subfigure[Impact of the Penalty Coefficient $c_{I}$ on Expected Total Cost]{
%
%
\definecolor{mycolor1}{rgb}{1.00000,0.00000,1.00000}%
\begin{tikzpicture}[scale=.55]

\begin{axis}[%
width=2.3in,
height=2in,
at={(0.758in,0.559in)},
scale only axis,
xmin=0.02,
xmax=0.18,
xlabel style={font=\color{white!15!black}},
xlabel={Penalty Coefficient $c_{I}$},
ymin=0.06,
ymax=0.2,
ylabel style={font=\color{white!15!black}},
ylabel={Expected Total Cost},
axis background/.style={fill=white},
legend style={at={(0.63,0.58)}, anchor=south west, legend cell align=left, align=left, draw=white!15!black}
]
\addplot [color=green, line width=2.0pt]
  table[row sep=crcr]{%
0.035	0.166008333340436\\
0.04	0.168508333356323\\
0.045	0.171008333351517\\
0.05	0.173508333330586\\
0.055	0.175883155136418\\
0.06	0.177973719732732\\
0.065	0.1796708998941\\
0.07	0.180931939990213\\
0.075	0.181803697424196\\
0.08	0.18236704847138\\
0.085	0.182709217323389\\
0.09	0.182904437401202\\
0.095	0.183008960100937\\
0.1	0.183061334218256\\
0.105	0.183085791198642\\
0.11	0.183096325198046\\
0.115	0.183100461355602\\
0.12	0.183101914694348\\
0.125	0.183102359492\\
0.13	0.183102473423484\\
0.135	0.183102496400658\\
0.14	0.183102499706738\\
0.145	0.183102499771559\\
0.15	0.183102499771559\\
0.155	0.183102499771559\\
0.16	0.183102499771559\\
0.165	0.183102499771559\\
0.17	0.183102499771559\\
0.175	0.183102499771559\\
};
\addlegendentry{$t_{s}$=2100}

\addplot [color=black, dashdotted, line width=2.0pt]
  table[row sep=crcr]{%
0.035	0.113508333340436\\
0.04	0.116841666689656\\
0.045	0.120175000018185\\
0.05	0.123508333330572\\
0.055	0.126698490637689\\
0.06	0.129501918984841\\
0.065	0.131747462591991\\
0.07	0.133350672700111\\
0.075	0.134366639063702\\
0.08	0.134939228477415\\
0.085	0.135223659337609\\
0.09	0.135348362934722\\
0.095	0.135396154443969\\
0.1	0.135411950365053\\
0.105	0.135416339745549\\
0.11	0.135417319358924\\
0.115	0.135417481671872\\
0.12	0.135417499061189\\
0.125	0.13541749977599\\
0.13	0.135417499776086\\
0.135	0.135417499776182\\
0.14	0.135417499776278\\
0.145	0.135417499776373\\
0.15	0.135417499776469\\
0.155	0.135417499776565\\
0.16	0.13541749977666\\
0.165	0.135417499776756\\
0.17	0.135417499776851\\
0.175	0.135417499776947\\
};
\addlegendentry{$t_{s}$=2700}

\addplot [color=mycolor1, dotted, line width=2.0pt]
  table[row sep=crcr]{%
0.035	0.0610083333404361\\
0.04	0.0651750000229894\\
0.045	0.0693416666848517\\
0.05	0.0735083333305717\\
0.055	0.0775266478424211\\
0.06	0.0810815935471031\\
0.065	0.08389427907868\\
0.07	0.0858157950740408\\
0.075	0.0869215999481186\\
0.08	0.0874476445539194\\
0.085	0.0876511126620964\\
0.09	0.0877141240641775\\
0.095	0.0877294038770653\\
0.1	0.0877321505826372\\
0.105	0.0877324787312422\\
0.11	0.08773249956997\\
0.115	0.0877324997911389\\
0.12	0.0877324997911389\\
0.125	0.0877324997911389\\
0.13	0.0877324997911389\\
0.135	0.0877324997911389\\
0.14	0.0877324997911389\\
0.145	0.0877324997911389\\
0.15	0.0877324997911389\\
0.155	0.0877324997911389\\
0.16	0.0877324997911389\\
0.165	0.0877324997911389\\
0.17	0.0877324997911389\\
0.175	0.0877324997911389\\
};
\addlegendentry{$t_{s}$=3300}

\end{axis}
\end{tikzpicture}%
		\label{fig:312}
	}	
	\caption{\small Impact of Penalty Coefficient $c_{I}$ with $t_{e}=3600s$ and $c_{s}=\bar{\pi}/10$}
	\label{fig4}
\end{figure}

\subsection{PR}

In this subsection, we numerically evaluate the bidding prices, and the expected total cost for PR. We  again set $t_e$ at  3600s. We investigate the variation of the deadlines and the recovery time on the bidding prices in the spot instances, and the expected total cost.  Figure \ref{fig:4} suggests that the bid price on spot instance (Fig. \ref{fig:4.1}) and expected total cost (Fig. \ref{fig:4.3}) increases with the recovery time. Intuitively, as the recovery time increases the more time is required to recover  a job after it is interrupted. Thus, the bidding price is higher. Hence, the expected cost is also higher.  Fig. ~\ref{fig:4.1}  and Fig.~\ref{fig:4.3} show that as the deadline increases, the bidding price and the expected cost decreases in the PR scenario. This is  consistent with Proposition \ref{pro3} which shows that bid price decreases as the deadline increases. 
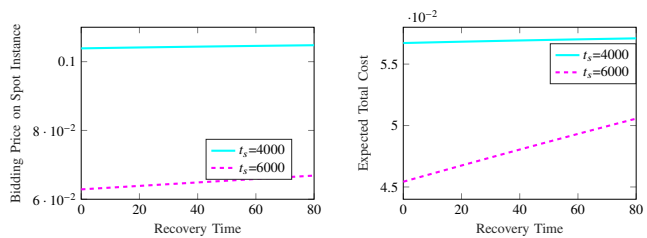
\begin{figure}[htbp]
	\centering
	\subfigure[Impact of the Recovery Time on Bidding Price]{
%
%
\definecolor{mycolor1}{rgb}{0.00000,1.00000,1.00000}%
\definecolor{mycolor2}{rgb}{1.00000,0.00000,1.00000}%
\begin{tikzpicture}[scale =0.53]

\begin{axis}[%
width=2.3in,
height=1.7in,
at={(0.758in,0.58in)},
scale only axis,
xmin=0,
xmax=80,
xlabel style={font=\color{white!15!black}},
xlabel={Recovery Time},
ymin=0.06,
ymax=0.11,
ylabel style={font=\color{white!15!black}},
ylabel={Bidding Price on Spot Instance},
axis background/.style={fill=white},
legend style={at={(0.72,0.35)}, anchor=north, legend cell align=left, align=left, draw=white!15!black}
]
\addplot [color=mycolor1, line width=1.5pt]
  table[row sep=crcr]{%
0	0.103861595082283\\
1	0.103875237870216\\
2	0.103888804864883\\
3	0.103902296066284\\
4	0.103915749371052\\
5	0.103929126882553\\
6	0.103942428600788\\
7	0.10395569242239\\
8	0.103968880450726\\
9	0.103981992685795\\
10	0.103995067024231\\
11	0.104008027672768\\
12	0.104020988321304\\
13	0.104033873176575\\
14	0.104046682238579\\
15	0.104059415507317\\
16	0.104072148776054\\
17	0.10408473045826\\
18	0.104097312140465\\
19	0.104109855926037\\
20	0.104122323918343\\
21	0.104134678220749\\
22	0.104147032523155\\
23	0.104159311032295\\
24	0.104171551644802\\
25	0.104183716464043\\
26	0.10419584338665\\
27	0.104207894515991\\
28	0.104219869852066\\
29	0.104231845188141\\
30	0.104243744730949\\
31	0.104255568480492\\
32	0.104267354333401\\
33	0.104279064393043\\
34	0.104290774452686\\
35	0.104302408719063\\
36	0.104313967192173\\
37	0.10432548776865\\
38	0.104336970448494\\
39	0.104348377335072\\
40	0.104359746325016\\
41	0.104371039521694\\
42	0.104382294821739\\
43	0.104393512225151\\
44	0.10440469173193\\
45	0.104415795445442\\
46	0.104426861262321\\
47	0.104437889182568\\
48	0.104448841309547\\
49	0.104459793436527\\
50	0.104470669770241\\
51	0.104481470310688\\
52	0.104492270851135\\
53	0.104502995598316\\
54	0.104513682448864\\
55	0.104524331402779\\
56	0.104534904563427\\
57	0.104545439827442\\
58	0.104555937194824\\
59	0.104566396665573\\
60	0.104576818239689\\
61	0.104587164020538\\
62	0.104597509801388\\
63	0.104607779788971\\
64	0.104618011879921\\
65	0.104628206074238\\
66	0.104638324475288\\
67	0.104648442876339\\
68	0.104658485484123\\
69	0.104668528091907\\
70	0.104678494906425\\
71	0.10468842382431\\
72	0.104698352742195\\
73	0.104708167970181\\
74	0.104717983198166\\
75	0.104727760529518\\
76	0.104737499964237\\
77	0.10474716360569\\
78	0.104756827247143\\
79	0.104766415095329\\
80	0.104776002943516\\
};
\addlegendentry{$t_{s}$=4000}

\addplot [color=mycolor2, dashed, line width=1.5pt]
  table[row sep=crcr]{%
0	0.0628723787635565\\
1	0.062922894975543\\
2	0.0629734206616879\\
3	0.0630239368736744\\
4	0.063074453085661\\
5	0.0631249882459641\\
6	0.0631755044579506\\
7	0.0632260301440954\\
8	0.0632765558302402\\
9	0.0633270815163851\\
10	0.0633776072025299\\
11	0.0634281234145164\\
12	0.063478639626503\\
13	0.0635291558384895\\
14	0.0635796720504761\\
15	0.0636301787883043\\
16	0.0636806855261326\\
17	0.0637311827898026\\
18	0.0637816705793143\\
19	0.0638321583688259\\
20	0.0638826366841793\\
21	0.0639331055253744\\
22	0.0639835743665695\\
23	0.0640340242594481\\
24	0.0640844741523266\\
25	0.0641349050968885\\
26	0.0641853360414505\\
27	0.0642357480376959\\
28	0.064286150559783\\
29	0.0643365436077118\\
30	0.0643869271814823\\
31	0.0644372918069363\\
32	0.0644876469582319\\
33	0.064537983161211\\
34	0.0645883098900318\\
35	0.0646386271446943\\
36	0.0646889254510403\\
37	0.0647392048090696\\
38	0.0647894652187824\\
39	0.0648397161543369\\
40	0.0648899481415749\\
41	0.0649401611804962\\
42	0.064990355271101\\
43	0.0650405304133892\\
44	0.0650906866073609\\
45	0.0651408238530159\\
46	0.0651909421503544\\
47	0.0652410414993763\\
48	0.0652911219000816\\
49	0.0653411644041538\\
50	0.065391206908226\\
51	0.065441211515665\\
52	0.0654911971747875\\
53	0.0655411638855934\\
54	0.0655911116480827\\
55	0.0656410215139389\\
56	0.0656909124314785\\
57	0.0657407844007015\\
58	0.0657906184732914\\
59	0.0658404335975647\\
60	0.0658902202993631\\
61	0.065939988052845\\
62	0.0659897084355354\\
63	0.0660394288182259\\
64	0.0660890923559666\\
65	0.0661387558937073\\
66	0.0661883720606566\\
67	0.0662379692792893\\
68	0.0662875191271305\\
69	0.0663370500266552\\
70	0.0663865619778633\\
71	0.06643602655828\\
72	0.0664854721903801\\
73	0.0665348704516888\\
74	0.0665842497646809\\
75	0.0666335911810398\\
76	0.0666828947007656\\
77	0.0667321792721748\\
78	0.0667814259469509\\
79	0.0668306252509356\\
80	0.0668798056066036\\
};
\addlegendentry{$t_{s}$=6000}

\end{axis}
\end{tikzpicture}%
		\label{fig:4.1}
	}
	\subfigure[Impact of the Recovery Time on Expected Total Cost]{
%
%
\definecolor{mycolor1}{rgb}{0.00000,1.00000,1.00000}%
\definecolor{mycolor2}{rgb}{1.00000,0.00000,1.00000}%
\begin{tikzpicture}[scale = 0.53]

\begin{axis}[%
width=2.3in,
height=1.7in,
at={(0.7in,0.57in)},
scale only axis,
xmin=0,
xmax=80,
xlabel style={font=\color{white!15!black}},
xlabel={Recovery Time},
ymin=0.044,
ymax=0.058,
ylabel style={font=\color{white!15!black}},
ylabel={Expected Total Cost},
axis background/.style={fill=white},
legend style={at={(0.63,0.67)}, anchor=south west, legend cell align=left, align=left, draw=white!15!black}
]
\addplot [color=mycolor1, line width=1.5pt]
  table[row sep=crcr]{%
0	0.0567138462663825\\
1	0.0567196017095761\\
2	0.0567253255733878\\
3	0.056731018014397\\
4	0.05673668633743\\
5	0.0567423234953833\\
6	0.0567479296447945\\
7	0.0567535120074111\\
8	0.0567590636191697\\
9	0.0567645846365708\\
10	0.0567700821983467\\
11	0.0567755424688125\\
12	0.0567809864682511\\
13	0.0567864003886261\\
14	0.0567917843863931\\
15	0.0567971386179897\\
16	0.0568024768727617\\
17	0.0568077719860594\\
18	0.05681305132492\\
19	0.0568183081464828\\
20	0.0568235358183258\\
21	0.0568287278184156\\
22	0.0568339043385083\\
23	0.05683905212293\\
24	0.0568441779237875\\
25	0.0568492752466275\\
26	0.0568543507884117\\
27	0.0568593981098497\\
28	0.0568644173673453\\
29	0.0568694216330494\\
30	0.0568743980373956\\
31	0.0568793467368036\\
32	0.0568842742629108\\
33	0.0568891743418242\\
34	0.0568940597702444\\
35	0.0568989179541731\\
36	0.0569037490500592\\
37	0.0569085594518819\\
38	0.056913349233505\\
39	0.0569181122863389\\
40	0.0569228549218097\\
41	0.056927571086405\\
42	0.0569322670365281\\
43	0.0569369428461294\\
44	0.0569415985891786\\
45	0.0569462283223017\\
46	0.0569508381918956\\
47	0.0569554282719844\\
48	0.0569599927017822\\
49	0.056964543452525\\
50	0.0569690687561397\\
51	0.0569735687692486\\
52	0.0569780552981036\\
53	0.0569825167397258\\
54	0.05698695902059\\
55	0.05699138221488\\
56	0.0569957806819986\\
57	0.0570001602659739\\
58	0.0570045210410533\\
59	0.0570088630815061\\
60	0.0570131864616222\\
61	0.0570174856784042\\
62	0.0570217719883061\\
63	0.0570260343385625\\
64	0.0570302783808633\\
65	0.0570345041896111\\
66	0.0570387063994122\\
67	0.0570428959918536\\
68	0.0570470621892675\\
69	0.0570512158632919\\
70	0.0570553463463044\\
71	0.0570594590976703\\
72	0.0570635594667733\\
73	0.0570676317037617\\
74	0.0570716917076969\\
75	0.0570757342784053\\
76	0.0570797594905533\\
77	0.0570837622815456\\
78	0.0570877530281808\\
79	0.0570917215581392\\
80	0.0570956781382331\\
};
\addlegendentry{$t_{s}$=4000}

\addplot [color=mycolor2, dashed, line width=1.5pt]
  table[row sep=crcr]{%
0	0.0454265717387325\\
1	0.0454929197646928\\
2	0.0455592340969428\\
3	0.0456255073045731\\
4	0.0456917422518506\\
5	0.0457579451066656\\
6	0.045824101778825\\
7	0.0458902218060172\\
8	0.0459563012428722\\
9	0.0460223395143844\\
10	0.0460883360452447\\
11	0.0461542869909439\\
12	0.0462201950788469\\
13	0.0462860597329944\\
14	0.04635188037718\\
15	0.0464176532350467\\
16	0.0464833809646411\\
17	0.0465490598241906\\
18	0.0466146892893253\\
19	0.0466802719658669\\
20	0.0467458041645278\\
21	0.0468112853612944\\
22	0.0468767181093864\\
23	0.0469420957133653\\
24	0.0470074237860378\\
25	0.0470726957035347\\
26	0.0471379170078344\\
27	0.0472030811471233\\
28	0.0472681906226053\\
29	0.0473332449122336\\
30	0.0473982434942658\\
31	0.0474631829321656\\
32	0.0475280656564611\\
33	0.0475928882680117\\
34	0.0476576531618431\\
35	0.0477223598182375\\
36	0.0477870048945503\\
37	0.0478515879273828\\
38	0.0479161084540497\\
39	0.0479805687801956\\
40	0.0480449656396469\\
41	0.0481092985717678\\
42	0.0481735671166764\\
43	0.0482377708152581\\
44	0.0483019092091742\\
45	0.0483659818408733\\
46	0.0484299882536008\\
47	0.0484939279914111\\
48	0.0485578005991761\\
49	0.0486216004652039\\
50	0.0486853374891403\\
51	0.0487490009420983\\
52	0.0488125955300064\\
53	0.0488761208022681\\
54	0.0489395763091647\\
55	0.0490029566758106\\
56	0.0490662664581011\\
57	0.0491295052092931\\
58	0.0491926676744342\\
59	0.0492557582960728\\
60	0.0493187742651647\\
61	0.0493817175436658\\
62	0.0494445807120161\\
63	0.0495073750384183\\
64	0.0495700862463092\\
65	0.0496327277346083\\
66	0.0496952877135786\\
67	0.0497577726064844\\
68	0.0498201753581639\\
69	0.0498825022750656\\
70	0.0499447529246833\\
71	0.0500069204359217\\
72	0.0500690109388864\\
73	0.0501310176859975\\
74	0.0501929466893172\\
75	0.0502547933909914\\
76	0.0503165574875164\\
77	0.0503782427269669\\
78	0.0504398446773944\\
79	0.0505013610556436\\
80	0.0505627975167733\\
};
\addlegendentry{$t_{s}$=6000}

\end{axis}
\end{tikzpicture}%
		\label{fig:4.3}
	}
	\caption{\small Impact of Recovery Time with $t_{e}=3600s$. }
	\label{fig:4}
\end{figure}

\section{Data-Driven Evaluation}\label{real}
To verify our models and comparing with baseline algorithms, we collect three sets of data for three instance types: r3.large, r4.16xlarge and d2.2xlarge in the US Eastern region, whose on-demand prices are \$0.1660, \$4.2560 and \$1.38 respectively\cite{ec2ins}.  We collect the Amazon EC2 spot price history for the three months from (July 9 - October 9, 2017). 
The empirical PDF and associated estimated PDF of these prices are shown by the black dots and blue line respectively in Figure \ref{fig:ama}. We observe that they approximately follow exponential functions, which are consistent among different instance types, though the spot prices are different.

We consider a job that needs one hour (i.e., $t_{e}=3600s$) to be executed without interruption. The one hour time periods are Oct.10, 3:00pm-4:00pm, Oct.12, 9:40am-10:40am, and Oct.14, 1:30pm-2:30pm for r3.large, r4.16xlarge, and d2.2xlarge, respectively. We first examine the optimal bid prices using three different bidding strategies (OTR-EG and with OTR-P ($t_{s}=\bar{\pi}/10$ and $t_{I} = \bar{\pi}/3$), and PR with recovery time 10s and 50s ) that are derived in Section \ref{User Bidding Strategies} on Amazon EC2 spot instances. {\em The strategies are summarized in Table \ref{rsp1} for $t_{s}=2000s$}. We consider a model where the user is price taker, one single user's action does not impact the distribution of the spot price. 

\begin{color} {black} We show that our approach outperforms two baseline algorithms, where we use one single type of instance to finish the job. Specifically, for smaller value of deadlines, we split the job into two sub-jobs of equal size, and each corresponding to one instance request. For Baseline I, each of them will be run on one on-demand instance. For Baseline II, we consider they only use the spot market. We adopt the strategy that has been proposed in \cite{lz15}. Since these two sub-jobs are requesting for the same types of spot instance, the bidding prices are same for both of them. For larger value of the deadlines, we run the whole job solely on one on-demand instance and one spot instance for Baseline I and Baseline II respectively. \end{color}



\subsubsection{Smaller Value of the Deadlines}
For shorter deadline with $t_{s}= 2000s$, we determine the optimal bid prices and optimal portion of job running on on-demand instance for different bidding strategies (Table \ref{rsp1}) to the associated spot instance (r3.large, r4.16xlarge and d2.2xlarge). We set the deadline at $t_s=2000s$ and compare our algorithm with two baseline algorithms. We use our proposed methods to calculate $p$ and $q$. Specifically, from Table~\ref{rsp1}, we observe that the price of PR is the highest, the price for OTR-P is medium, and the price of OTR-EG is lowest, which are consistent among different instance types and consistent with our simulation results in Section \ref{simulation}. 

Fig.~\ref{fig:6} compares the job completion time, the completed job and total cost for different instance types with different bidding strategies  and the baseline algorithms when the deadline $t_{s}$ is $2000s$. 

Fig.~\ref{fig:6.4} shows that the total cost running with OTR-EG is almost equal to that of the bidding strategy OTR-P and with PR. Thus, the penalty does not increase the cost, yet, increases the portion of the completed jobs. Our numerical results show that {\em the penalty mechanism reduces the cost by almost $50\%$ for all instances compared to the baseline I algorithm.} {\em And compared to other methods, the method where we put penalty for the incomplete job and violating the deadline can achieve the minimum cost but is able to get the job done before the deadline (see Fig.~\ref{fig:6.3} and Fig.~\ref{fig:6.2}).}

Fig.~\ref{fig:6.3} shows that r3.large job is interrupted when using OTR-EG and baseline II. In contrast, {\em the r3.large job with  penalty and PR bidding strategies are not interrupted.} However, for r4.16xlarge and d2.2xlarge instances, none of experiments are interrupted. Thus, {\em the penalty does not affect the rate of completion of jobs}. Since in the baseline I algorithm, all the jobs are put in the on-demand instance, thus, the job is never interrupted. {\em Note that our penalty based approach is able to achieve similar completion rate of the baseline algorithm I, however, at a smaller cost.which can be verified in Fig.~\ref{fig:6.4}.}

From Fig.~\ref{fig:6.3}, we observe that none of the job is interrupted by using different methods for job r4.16xlarge. We compare the job completion time of r4.16xlarge in Fig.~\ref{fig:6.2}. The results show {\em the job completion time by using our algorithms is longer than that of baseline I but not beyond the deadline (the red line), while the job completion time is beyond the deadline by using baseline II.} Recall for baseline I, we split the job to two sub-jobs with the same sizes and each of them will be run on on-demand instance without any interruption, which means the job completion time is 0.5 hour. There are two reasons that the job completion times are longer by using our proposed algorithms for job r4.16xlarge compared to baseline I: i) more than 50\% of the job is allocated to spot instance, and ii) there may be some time to enter the system. 


\begin{figure*}[htbp]
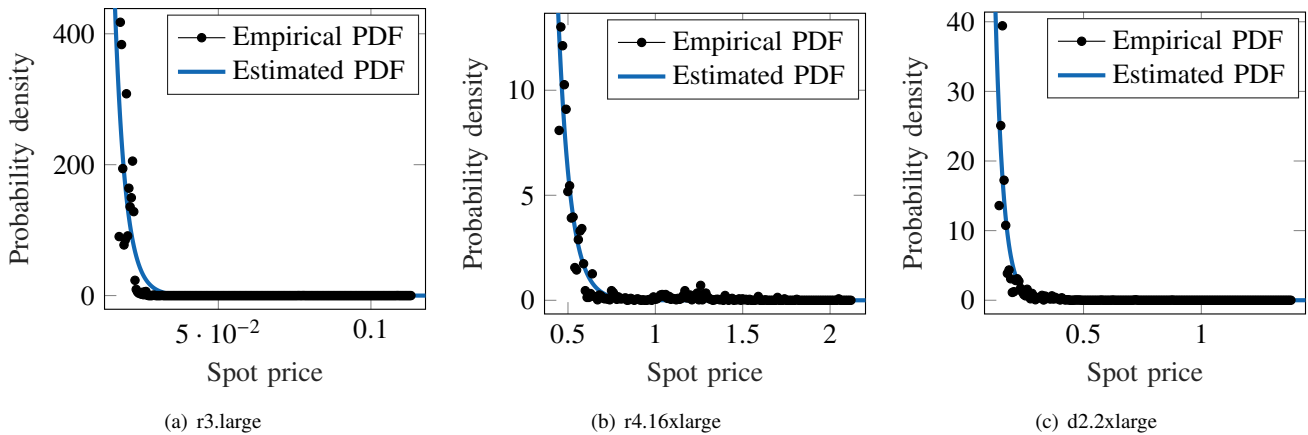

	\centering
	\subfigure[r3.large]{
		\input{r3.large.tex}
		\label{fig:5.1}
	}
	\subfigure[r4.16xlarge ]{
		\input{r4.16xlarge.tex}
		\label{fig:5.2}
	}
	\subfigure[d2.2xlarge ]{
		\input{d2.2xlarge.tex}
		\label{fig:5.3}
	}
	\caption{ Fitting the probability density function of Amazon spot price in the US Eastern region, the best fits are exponential functions with equation $a\exp(bx)$. The fitted parameter values, which are with 95\% confidence bounds, are (a, b) = (44350, -285.7), (8126, -14.39), and (1571, -28.84) for Fig. \ref{fig:5.1}, Fig. \ref{fig:5.2}, and Fig. \ref{fig:5.3}, respectively.}
	\label{fig:ama}
\end{figure*}

\begin{table*}[htbp]
\centering
\caption{Optimal bid prices for a single instance with $t_{s}=2000s$}
\label{rsp1}
\begin{tabular}{|c|c|c|c|c|c|c|c|c|c|c|c|c|}
\hline
\multirow{2}{*}{\begin{tabular}[c]{@{}c@{}}Instance \\ Types\end{tabular}} & \multirow{2}{*}{$\bar{\pi}$} & \multirow{2}{*}{$\underline{\pi}$} & \multicolumn{2}{c|}{\begin{tabular}[c]{@{}c@{}}OTR-EG\end{tabular}} & \multicolumn{2}{c|}{\begin{tabular}[c]{@{}c@{}}OTR-P\end{tabular}} & \multicolumn{2}{c|}{\begin{tabular}[c]{@{}c@{}}PR with $t_{r} = 10s$\end{tabular}} & \multicolumn{2}{c|}{Baseline I} & \multicolumn{2}{c|}{Baseline II} \\ \cline{4-13} 
                                                                           &                              &                                    & p                                               & q                                              & p                                              & q                                            & p                                              & q                                              & p                 & q           & p                  & q           \\ \hline
r3.large                                                                   & \$0.1660                     & \$0.0173                           & \$0.04258                                       & 0.444505                                       & \$0.08813                                      & 0.444444                                     & \$0.166                                        & 0.444444                                       & \$0.1660          & 0.5         & \$0.02357          & 0.5         \\ \hline
r4.16xlarge                                                                & \$4.2560                     & \$0.4343                           & \$1.0666                                        & 0.4445                                         & \$1.8375                                       & 0.444444                                     & \$4.2560                                       & 0.444444                                       & \$4.2560          & 0.5         & \$0.5588           & 0.5         \\ \hline
d2.2xlarge                                                                 & \$1.38                       & \$0.138                            & \$0.3538                                        & 0.44461                                        & \$0.83833                                      & 0.444444                                     & \$1.38                                         & 0.444444                                       & \$1.38            & 0.5         & \$0.2              & 0.5         \\ \hline
\end{tabular}
\end{table*}

\subsubsection{Larger Value of the Deadlines}

We now discuss our results when $t_{s} = 4000s$. When $t_{s} > 3600s$, from Section \ref{User Bidding Strategies} (Proposition \ref{pro2} and Proposition \ref{pro4}), we know that the whole one-hour job will be run on spot instance. The optimal bid prices are shown in Table \ref{rsp2}. 


Unlike the bidding prices with smaller value of the deadlines in Table \ref{rsp1}, Table \ref{rsp2} shows that the bidding prices with PR are the lowest, while the bidding prices for OTR-P are the highest in different instance types. As we also expect from Section \ref{User Bidding Strategies}'s analysis for PR, longer recovery times (50s) gives higher bidding prices than that of shorter recovery times (10s), which are also consistent with the findings in \cite{lz15}. 

From Figure \ref{fig:7.4} we can see that compared to the other two bidding strategies, OTR-P does not introduce more cost but can finish all the job in time without any interruption or penalty (see Figure \ref{fig:7.3} and Figure \ref{fig:7.2}, verifying the reliability of that bidding strategy. 

Figure \ref{fig:7.3} shows that with OTR-EG bidding strategy, it cannot finish more job compared to PR (e.g., r3.large), that is because the job may be interrupted after running some time, while the same job will be recovered from interruption with PR bidding strategy even with lower bidding prices (see Table \ref{rsp2}). On the other hand, there may  be a delay in finishing job in the PR. The reason behind that is the job may not be interrupted by using OTR-EG because of the higher bidding price, while PR job may be interrupted, which induces that the finished job can be more by using OTR-EG compared to PR strategy within a specific time period. Note that our proposed strategy outperforms the baseline II, as the baseline II strategy gives rise incomplete jobs for r3.large. Our proposed strategy for PR as well as for OTR-EG always results in completed jobs. 

Note that the job completion time is the summation of the time to enter the system and running time. In Figure~\ref{fig:7.2}, we can observe that for the jobs running on r3.large and d2.2xlarge, the completion time OTR-EG strategies is not longer than that of OTR-P strategy, that is because the bid price of OTR-EG is lower than that of OTR-P (see Table \ref{rsp2}), which means the job OTR-EG has higher probability to get interrupted, thus the running time is shorter. However, since the time to enter the system is longer with a lower bidding price, the job running on r4.16xlarge with OTR-EG bidding strategy has  less job completion time compared to that of OTR-P strategy. 

Another interesting result in Figure~\ref{fig:7.2}, unlike our expectation (higher recovery time may induce higher completion time and higher total cost),  is that within specific deadline, in insance d2.2xlarge, the PR with longer recovery time completes more job compared to that with shorter recovery time, that is because the longer recovery time yields higher bidding price (see Table \ref{rsp2}),  contributing to less interruption and less time to complete the job.


\begin{table*}[ht]
\centering
\caption{Optimal bid prices for a single instance with $t_{s}=4000s$}
\label{rsp2}
\begin{tabular}{|c|c|c|c|l|c|l|c|l|c|l|c|l|}
\hline
\begin{tabular}[c]{@{}c@{}}Instance \\ Types\end{tabular} & $\bar{\pi}$ & $\underline{\pi}$ & \multicolumn{2}{c|}{\begin{tabular}[c]{@{}c@{}}OTR-EG\end{tabular}} & \multicolumn{2}{c|}{\begin{tabular}[c]{@{}c@{}}OTR-P\end{tabular}} & \multicolumn{2}{c|}{\begin{tabular}[c]{@{}c@{}}PR with $t_{r} = 10s$\end{tabular}} & \multicolumn{2}{c|}{\begin{tabular}[c]{@{}c@{}}PR with $t_{r}=50s$\end{tabular}} & \multicolumn{2}{c|}{Baseline II} \\ \hline
r3.large                                                  & \$0.1660    & \$0.0173          & \multicolumn{2}{c|}{\$0.026}                                                                     & \multicolumn{2}{c|}{\$0.090299}                                                               & \multicolumn{2}{c|}{\$0.025463}                                                                 & \multicolumn{2}{c|}{\$0.025855}                                                               & \multicolumn{2}{c|}{\$0.026}     \\ \hline
r4.16xlarge                                               & \$4.2560    & \$0.4343          & \multicolumn{2}{c|}{\$0.606983}                                                                  & \multicolumn{2}{c|}{\$1.8808}                                                                 & \multicolumn{2}{c|}{\$0.596374}                                                                 & \multicolumn{2}{c|}{\$ 0.604158}                                                              & \multicolumn{2}{c|}{\$0.606983}  \\ \hline
d2.2xlarge                                                & \$1.38      & \$0.138           & \multicolumn{2}{c|}{\$0.3538}                                                                    & \multicolumn{2}{c|}{\$0.83833}                                                                & \multicolumn{2}{c|}{\$0.218868
}                                                                     & \multicolumn{2}{c|}{\$0.222752}                                                                   & \multicolumn{2}{c|}{\$0.3538}    \\ \hline
\end{tabular}
\end{table*}


\begin{figure*}
\centering     
\subfigure[Total cost for different instance types]{\label{fig:6.4}\includegraphics[width=57mm]{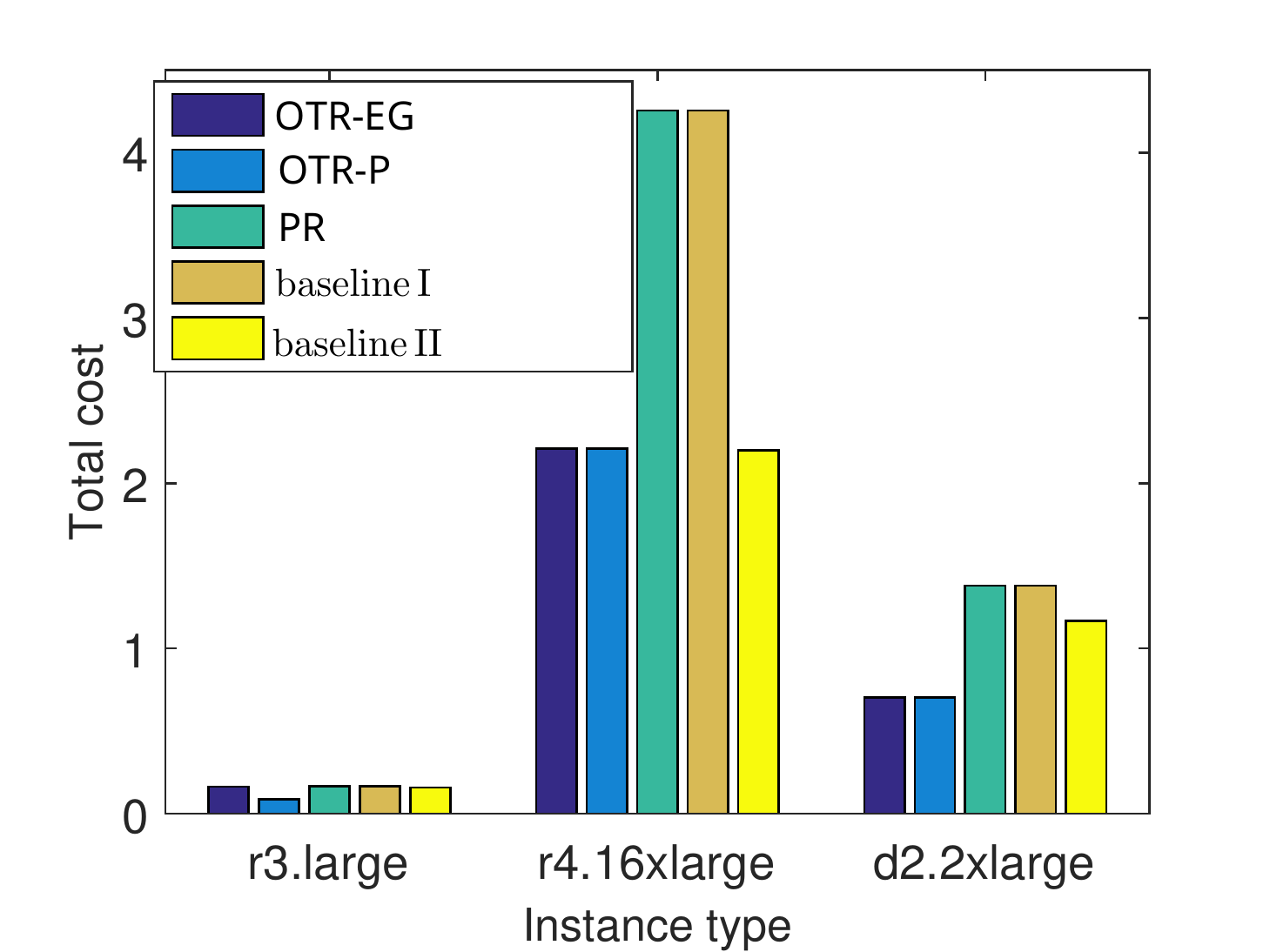}}
\subfigure[Completed job for different instance types]{\label{fig:6.3}\includegraphics[width=57mm]{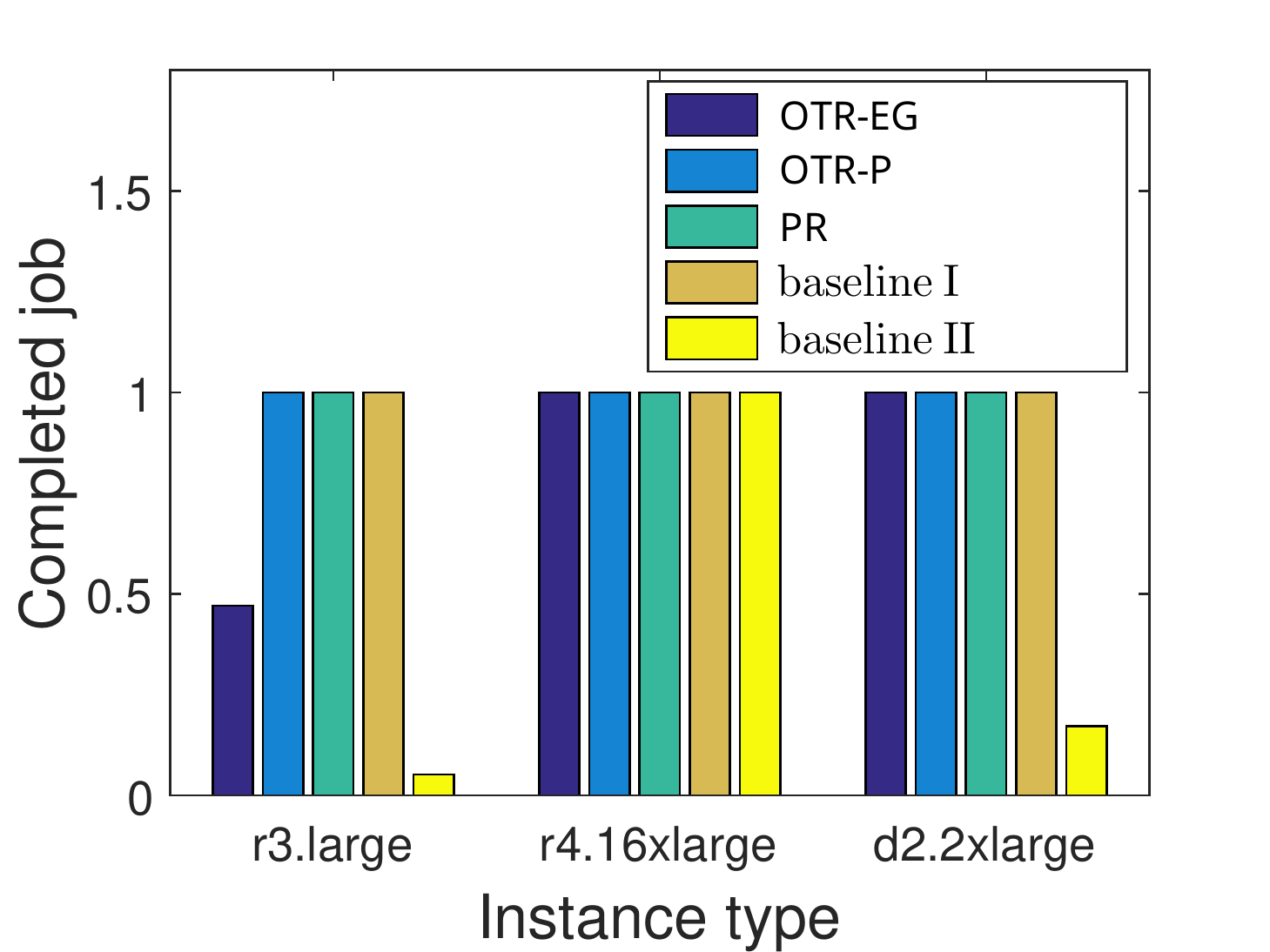}}
\subfigure[Job completion time for different instance types]{\label{fig:6.2}\includegraphics[width=57mm]{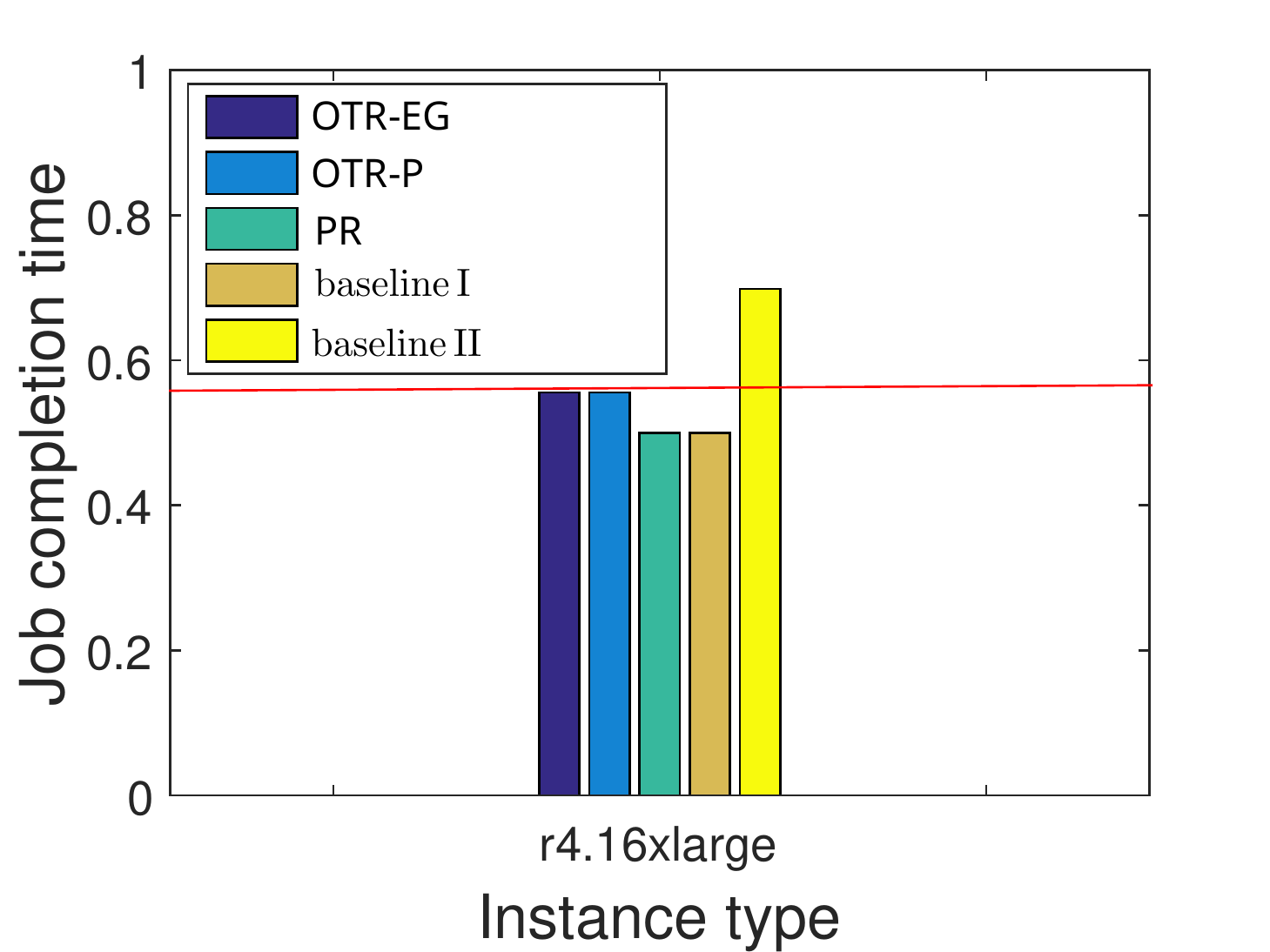}}
\caption{Comparison among different bidding strategies for smaller value of deadlines}
\label{fig:6}
\end{figure*}

\begin{figure*}
\centering     
\subfigure[Total cost for different instance types]{\label{fig:7.4}\includegraphics[width=57mm]{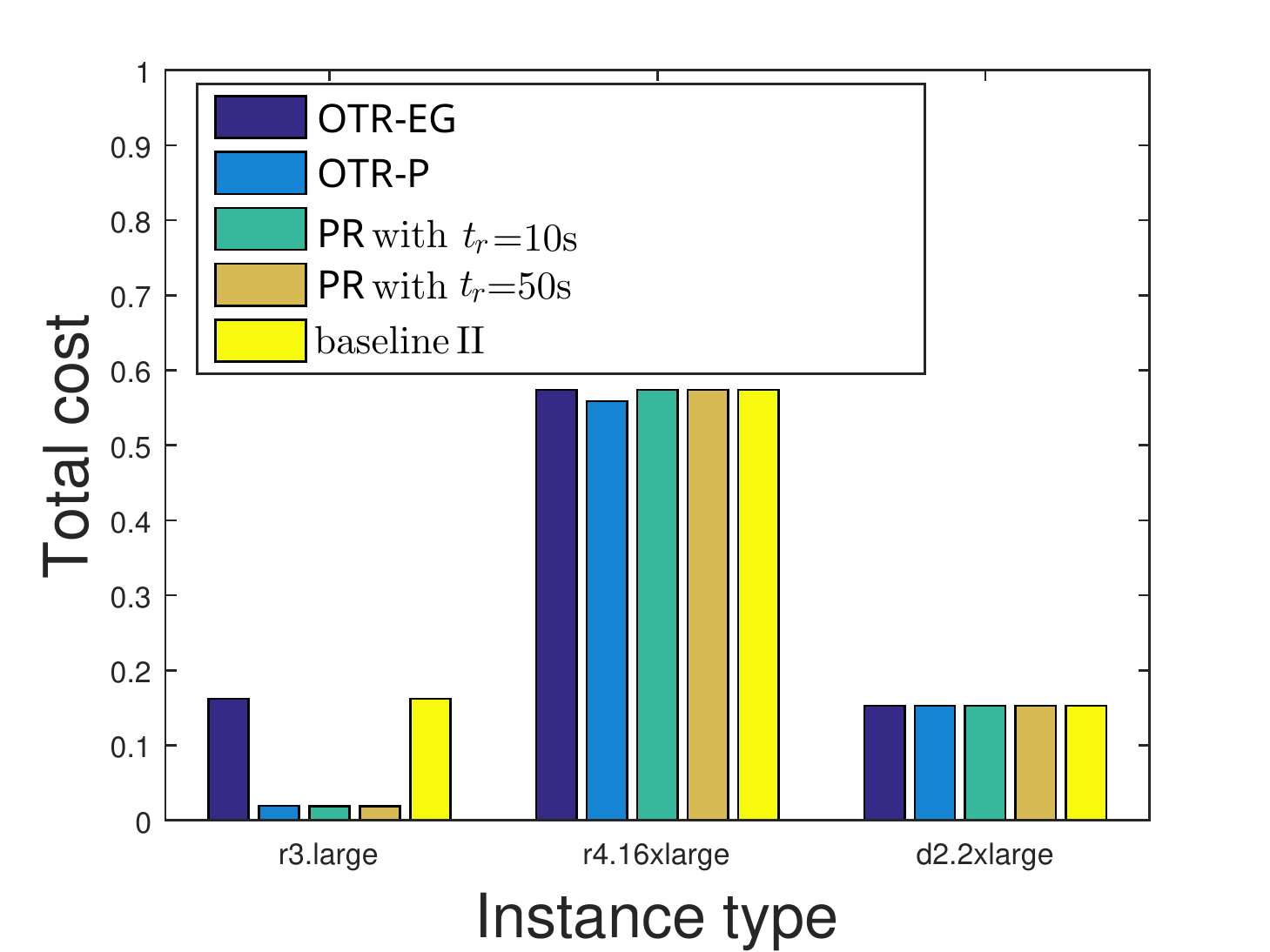}}
\subfigure[Completed job for different instance types]{\label{fig:7.3}\includegraphics[width=57mm]{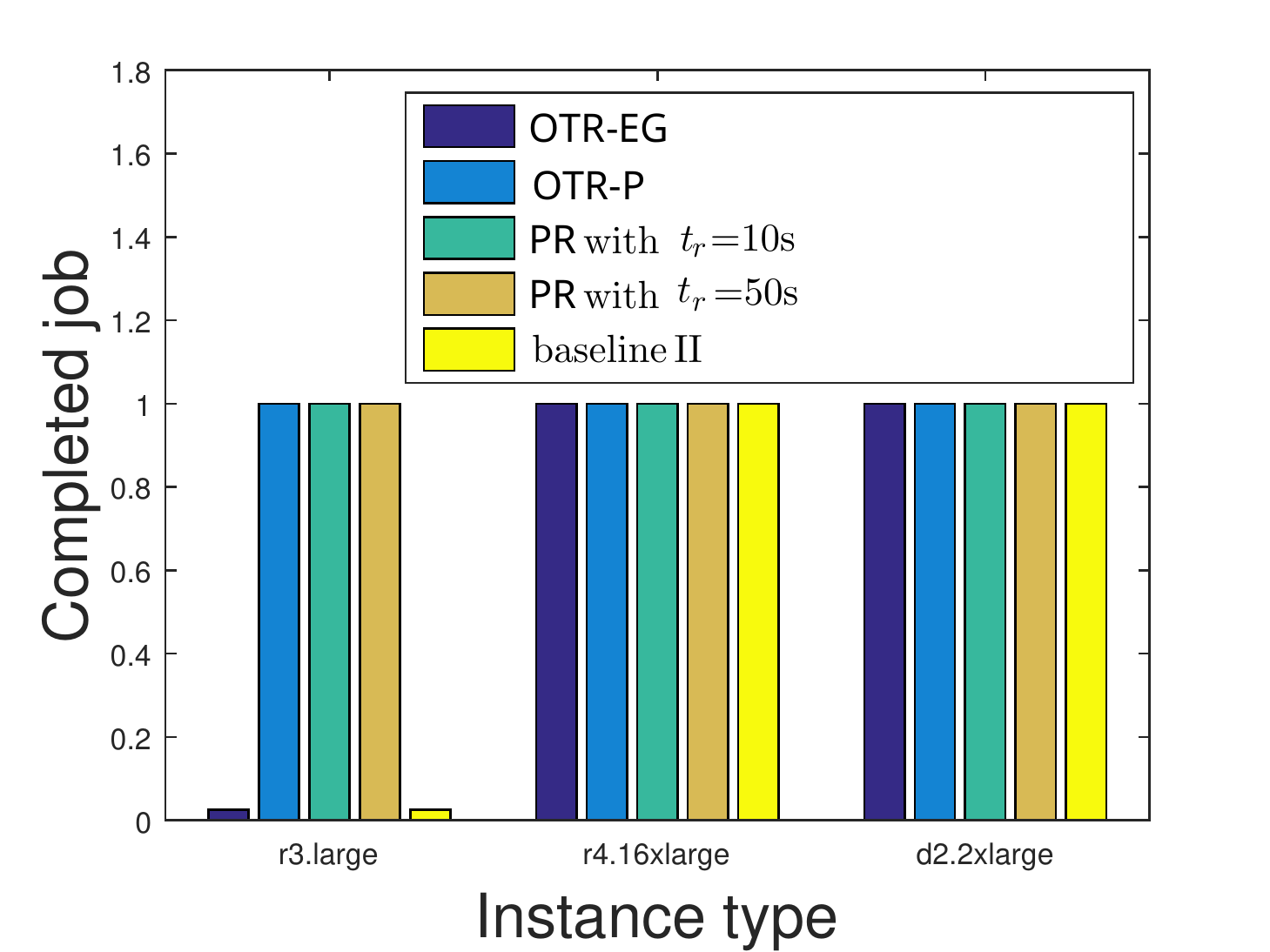}}
\subfigure[Job completion time for different instance types]{\label{fig:7.2}\includegraphics[width=57mm]{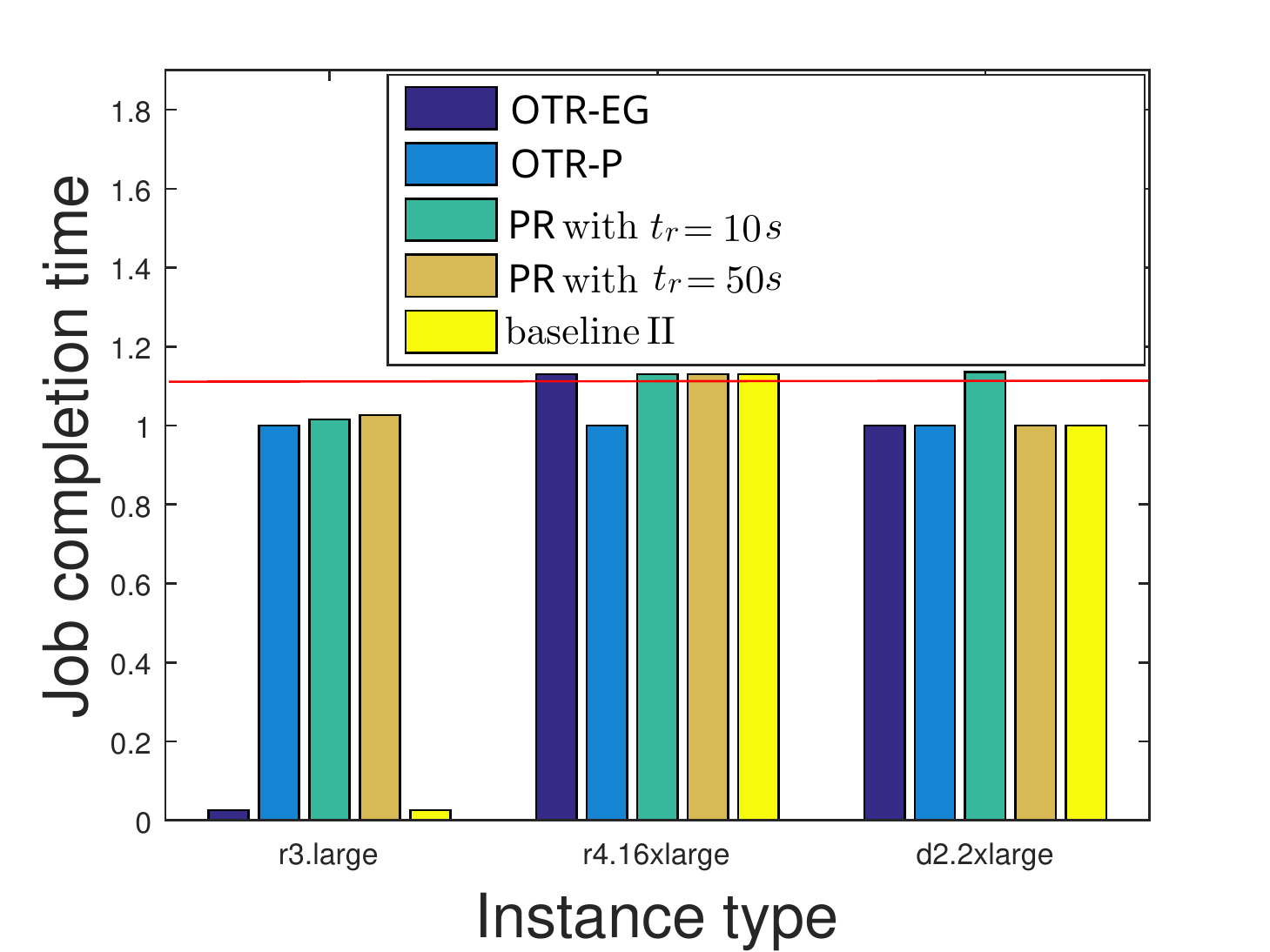}}
\caption{Comparison among different bidding strategies for smaller value of deadlines}
\label{fig:7}
\end{figure*}

\section{Conclusions and Future Work}\label{con}

In this paper, we develop optimization models to minimize the cost by randomizing between the on-demand, and spot instance. We also provide a mechanism to bid optimally in the spot instance. We consider three  different strategies: one-time requests with expected guarantee, one-time requests with penalty, and persistent requests. We characterize the optimal portion of the job that should be run in the on-demand instance. Our analytical result shows that a user should never opt for on-demand instance if the execution time is smaller than the deadline. However, when the the execution time is larger than the deadline, the user should select the on-demand instance with a non-zero portion of the job. Additionally, the bid price also increases as the execution time increases. However, the portion to be run on the on-demand instance never exceeds half for one-time request without penalty and the persistent request scenario.  Our numerical experiments shows the trade-off between higher prices to avoid interruptions (for one-time request without penalty), higher prices to avoid penalty (for one-time request with penalty) and lower prices to save money (for persistent request) on the condition that the deadline requirement is satisfied. It also shows that the user's cost is the lowest in the persistent request scenario.  Finally, we use real world data to test our model and analytical results to verify our models and shows that the user can significantly reduces its cost both in the one-time and the persistent request scenario. 

Some cloud service providers may give a two-minute warning to the user before the instance is revoked \cite{wang17}. The user can use re-bidding strategy to avoid termination of the job. We do not consider such scenario in this paper, however, we believe that such analysis constitutes an interesting future direction. In order to prevent the job from revocation, some user may choose to run each spot request on multiple machines. However, it is not clear how many machines the user should choose and how much to bid. Thus, it can be another interesting research direction.


%



%
%

\ifCLASSOPTIONcaptionsoff
  \newpage
\fi



%

\bibliographystyle{IEEEtran}
\bibliography{acmsmall-sample-bibfile}

\begin{thebibliography}{10}
\providecommand{\url}[1]{#1}
\csname url@samestyle\endcsname
\providecommand{\newblock}{\relax}
\providecommand{\bibinfo}[2]{#2}
\providecommand{\BIBentrySTDinterwordspacing}{\spaceskip=0pt\relax}
\providecommand{\BIBentryALTinterwordstretchfactor}{4}
\providecommand{\BIBentryALTinterwordspacing}{\spaceskip=\fontdimen2\font plus
\BIBentryALTinterwordstretchfactor\fontdimen3\font minus
  \fontdimen4\font\relax}
\providecommand{\BIBforeignlanguage}[2]{{%
\expandafter\ifx\csname l@#1\endcsname\relax
\typeout{** WARNING: IEEEtran.bst: No hyphenation pattern has been}%
\typeout{** loaded for the language `#1'. Using the pattern for}%
\typeout{** the default language instead.}%
\else
\language=\csname l@#1\endcsname
\fi
#2}}
\providecommand{\BIBdecl}{\relax}
\BIBdecl

\bibitem{web3}
L.~Columbus, ``Roundup of cloud computing forecasts,'' 2017.

\bibitem{lz16}
L.~Zheng, C.~J. Wong, C.~G. Brinton, C.~W. Tan, S.~Ha, and M.~Chiang, ``On the
  viability of cloud virtual service provider,'' in \emph{SIGMETRICS}, France,
  2016.

\bibitem{web}
Amazon, \url{https://aws.amazon.com/ec2/pricing/}, amazon EC2 Pricing.

\bibitem{google}
Google, \url{https://cloud.google.com/}, google Cloud Platform.

\bibitem{azure}
Microsoft, \url{https://azure.microsoft.com/en-us/?v=17.36}, microsoft windows
  azure.

\bibitem{web2}
Amazon,
  \url{http://docs.aws.amazon.com/AWSEC2/latest/UserGuide/spot-requests.html},
  amazon EC2 spot requests.

\bibitem{jin}
H.~Jin, X.~Wang, S.~Wu, S.~Di, and X.~Shi, ``Towards optimized fine-grained
  pricing of iaas cloud platform,'' \emph{IEEE Transactions on cloud
  computing}, vol.~3, no.~4, pp. 436--448, october 2015.

\bibitem{feng}
G.~Feng, S.~Garg, R.~Buyya, and W.~Li, ``Revenue maximization using adaptive
  resource provisioning in cloud computing environments.''\hskip 1em plus 0.5em
  minus 0.4em\relax ACM/IEEE 13th International Conference on Grid Computing,
  2012, pp. 192--200.

\bibitem{wang13}
P.~Wang, Y.~Qi, D.~Hui, L.~Rao, and X.~Liu, ``Present or future: optimal
  pricing for spot instances.''\hskip 1em plus 0.5em minus 0.4em\relax IEEE
  33rd International Conference on Distributed Computing Systems, 2013, pp.
  410--419.

\bibitem{xu13}
H.~Xu and B.~Li, ``Dynamic cloud pricing for revenue maximization,'' \emph{IEEE
  Transactions on Cloud Computing}, vol.~1, no.~2, pp. 158 -- 171, November
  2013.

\bibitem{sha12}
B.~Sharma, R.~K. Thulasiram, P.~Thulasiraman, and S.~K. Garg, ``Pricing cloud
  compute commodities: A novel financial economic model.''\hskip 1em plus 0.5em
  minus 0.4em\relax Ottawa, ON, Canada: 2012 12th IEEE/ACM International
  Symposium on Cluster, Cloud and Grid Computing, 2012, pp. 451--457.

\bibitem{lz15}
L.~Zheng, C.~J. Wong, C.~W. Tan, M.~Chiang, and X.~Wang, ``How to bid the
  cloud,'' in \emph{SIGCOMM}, London, United Kingdom, 2015, pp. 71--84.

\bibitem{dli16}
D.Li, C.~Chen, Y.~Zhang, J.~Zhu, and R.Yu, ``Dcloud: Deadline-aware resource
  allocation for cloud computing jobs,'' pp. 2248--2260, 2016.

\bibitem{brown08}
B.~M. E, ``Famine early warning systems and remote sensing data,'' in
  \emph{Springer}, 2011, p. 121.

\bibitem{wang155}
W.~Wang, M.~Barnard, and L.~Ying, ``Decentralized scheduling with data locality
  for data-parallel computation on peer-to-peer networks.''\hskip 1em plus
  0.5em minus 0.4em\relax Monticello, IL, USA: 2015 53rd Annual Allerton
  Conference on Communication, Control, and Computing, 2015, pp. 337--344.

\bibitem{tan12}
J.~Tan, X.~Meng, and L.~Zhang, ``Delay tails in mapreduce scheduling.''\hskip
  1em plus 0.5em minus 0.4em\relax London, England, UK: SIGMETRICS, 2012, pp.
  5--16.

\bibitem{chen}
Y.~Chen, ``Optimal dynamic auctions for display advertising,'' pp. 897--913,
  2017.

\bibitem{hp}
C.~Bilington, ``Hp cuts risk with portfolio approach,'' \emph{Purchasing.com},
  February 2002.

\bibitem{sim05}
M.~de~Albeniz and S.-L. D., ``A portfolio approach to procurement contracts,''
  \emph{Production and Operations Management}, vol.~14, no.~1, pp. 90--114,
  2005.

\bibitem{qi10}
Q.~Fu, C.-Y. Lee, and C.~piaw Teo, ``Procurement management using option
  contracts: random spot price and the portfolio effect,'' \emph{IIE
  Transactions}, vol.~42, no.~11, pp. 793--811, 2010.

\bibitem{li16}
D.~Li, C.~Chen, Y.~Zhang, J.~Zhu, and R.~Yu, ``Dcloud: Deadline--aware resource
  allocation for cloud computing jobs,'' \emph{IEEE Transactions on Parallel
  and Distributed Systems}, vol.~27, no.~8, pp. 2248--2260, August 2016.

\bibitem{maria14}
M.~A. Rodriguez and R.~Buyya, ``Deadline based resource provisioning and
  scheduling algorithm for scientific workflows on clouds,'' \emph{IEEE
  Transactions on Cloud Computing}, vol.~2, no.~2, pp. 222--235, 2014.

\bibitem{rod13}
R.~N. Calheiros and R.~Buyya, ``Meeting deadlines of scientific workflows in
  public clouds with tasks replication,'' \emph{IEEE Transactions on Parallel
  and Distributed Systems}, vol.~25, no.~7, pp. 1787--1796, September 2013.

\bibitem{abb88}
R.~Abbott and H.~Garcia-Molina, ``Scheduling real-time transactions: a
  performance evaluation.''\hskip 1em plus 0.5em minus 0.4em\relax Los Angeles,
  California, USA: Proceedings of the 14th VLDB Conference, 1988, pp. 1--12.

\bibitem{zhou17}
R.~Zhou, Z.~Li, C.~Wu, and Z.~Huang, ``An efficient cloud market mechanism for
  computing jobs with soft deadlines,'' \emph{IEEE/ACM Transactions on
  Networking}, vol.~25, no.~2, pp. 793--805, November 2017.

\bibitem{mak11}
M.~Hadji, W.~Louati, and D.~Zeghlache, ``Constrained pricing for cloud resource
  allocation.''\hskip 1em plus 0.5em minus 0.4em\relax Cambridge, MA, USA:
  Network Computing and Applications (NCA), 2011 10th IEEE International
  Symposium on, 2011, pp. 359--365.

\bibitem{ash09}
A.~A. Daoud, S.~Agarwal, and T.~Alpcan, ``Brief announcement: Cloud computing
  games: Pricing services of large data centers.''\hskip 1em plus 0.5em minus
  0.4em\relax Berlin, Heidelberg: Distributed Computing, Spring 2009, pp.
  309--310.

\bibitem{danilo11}
D.~Ardagna, B.~Panicucci, and M.~Passacantando, ``A game theoretic formulation
  of the service provisioning problem in cloud systems.''\hskip 1em plus 0.5em
  minus 0.4em\relax India: WWW 2011, 2011, pp. 177--186.

\bibitem{ard13}
------, ``Generalized nash equilibria for the service provisioning problem in
  cloud systems,'' \emph{IEEE TRANSACTIONS ON SERVICES COMPUTING}, vol.~6,
  no.~4, pp. 429--442, October 2013.

\bibitem{kut99}
E.~Kutanoglus and D.~Wu, ``On combinatorial auction and lagrangean relaxation
  for distributed resource scheduling,'' \emph{IIE Transactions}, vol.~31,
  no.~9, pp. 813--826, 1999.

\bibitem{song13}
K.~Song, Y.~Yao, and L.~Golubchik, ``Exploring the profit-reliability trade-off
  in amazon's spot instance market: A better pricing mechanism.''\hskip 1em
  plus 0.5em minus 0.4em\relax Montreal, QC, Canada: Quality of Service
  (IWQoS), 2013 IEEE/ACM 21st International Symposium on, 2013.

\bibitem{qwang13}
Q.~Wang, K.~Ren, and X.~Meng, ``When cloud meets ebay: Towards effective
  pricing for cloud computing.''\hskip 1em plus 0.5em minus 0.4em\relax 2012
  Proceedings IEEE INFOCOM, 2012, pp. 936--944.

\bibitem{kojima2001complexity}
M.~Kojima and A.~Takeda, ``Complexity analysis of successive convex relaxation
  methods for nonconvex sets,'' \emph{Mathematics of Operations Research},
  vol.~26, no.~3, pp. 519--542, 2001.

\bibitem{ec2ins}
Amazon, ``Easy amazon ec2 instance comparison,''
  \emph{https://www.ec2instances.info/}.

\bibitem{wang17}
C.~Wang, B.~Urgaonkar, A.~Gupta, G.~Kesidis, and Q.~Liang, ``Exploiting spot
  and burstable instances for improving the cost-efficacy of in-memory caches
  on the public cloud.''\hskip 1em plus 0.5em minus 0.4em\relax Belgrade,
  Serbia: In EuroSys'17, 2017, pp. 620--634.

\end{thebibliography}


\begin{thebibliography}{10}
\providecommand{\url}[1]{#1}
\csname url@samestyle\endcsname
\providecommand{\newblock}{\relax}
\providecommand{\bibinfo}[2]{#2}
\providecommand{\BIBentrySTDinterwordspacing}{\spaceskip=0pt\relax}
\providecommand{\BIBentryALTinterwordstretchfactor}{4}
\providecommand{\BIBentryALTinterwordspacing}{\spaceskip=\fontdimen2\font plus
\BIBentryALTinterwordstretchfactor\fontdimen3\font minus
  \fontdimen4\font\relax}
\providecommand{\BIBforeignlanguage}[2]{{%
\expandafter\ifx\csname l@#1\endcsname\relax
\typeout{** WARNING: IEEEtran.bst: No hyphenation pattern has been}%
\typeout{** loaded for the language `#1'. Using the pattern for}%
\typeout{** the default language instead.}%
\else
\language=\csname l@#1\endcsname
\fi
#2}}
\providecommand{\BIBdecl}{\relax}
\BIBdecl

\bibitem{ad09}
A.~Luca and M.~Bhide, \emph{Storage virtualization for dummies, Hitachi Data
  Systems Edition}.\hskip 1em plus 0.5em minus 0.4em\relax John and Wiley
  Publishing, 2009.

\bibitem{94}
Y.~Xiang, T.~Lan, V.~Aggarwal, and R.~Chen, ``Joint latency and cost
  optimization for erasure-coded data center storage,'' \emph{IEEE/ACM Trans.
  Netw}, vol.~24, no.~4, pp. 2443--2457, 2016.

\bibitem{gu11}
J.~Guerra, H.~Pucha, W.~J.Glider, and R.~Rangaswami, ``Cost effective storage
  using extent based dynamic tiering,'' in \emph{In Proceedings of the 9th
  USENIX Conference on File and Stroage Technologies}.\hskip 1em plus 0.5em
  minus 0.4em\relax Usenix Association, 2011, pp. 20--20.

\bibitem{kim14}
H.~Kim, S.~Seshadri, C.~Dickey, and L.~Chiu, ``Evaluating phase change memory
  for enterprise storage systems: study of caching and tiering approaches,'' in
  \emph{In Proceedings of the 12th USENIX Conference on File and Storage
  Technologies}.\hskip 1em plus 0.5em minus 0.4em\relax Santa Clara, CA, USA:
  Usenix Association, 2014, pp. 33--45.

\bibitem{li14}
Z.~Li, A.~Mukker, and E.~Zadok, ``On the importance of evaluating storage
  systems’ $\$$ costs,'' in \emph{In Proceedings of the 6th USENIX Conference
  on Hot Topics in Storage and File Systems}.\hskip 1em plus 0.5em minus
  0.4em\relax Philadelphia PA USA: Usenix Association, 2014, pp. 6--6.

\bibitem{wang14}
H.~Wang and P.~Varman, ``Balancing fairness and efficiency in tiered storae
  systems with bottleneck-aware allocation,'' in \emph{In Proceedings of the
  12th USENIX Conference on File and Storage Technologies}.\hskip 1em plus
  0.5em minus 0.4em\relax Santa Clara, CA, USA: Usenix Association, 2014, pp.
  229--242.

\bibitem{oa11}
O.~Ben-Yehuda, M.~Ben-Yehuda, A.~Schuster, and D.~Tsafrir, ``Deconstructing
  amazon ec2 spot instance pricing,'' in \emph{In Proceedings of the IEEE 3rd
  International Conference on Cloud Computing Technology and Science}.\hskip
  1em plus 0.5em minus 0.4em\relax Athens, Greece: CloudCom 2011, 2011, pp.
  304--311.

\bibitem{i12}
I.~Drago, M.~Mellia, M.~Munaf\'{o}, A.~Sperotto, R.~Sadre, and A.~Pras,
  ``Inside dropbox: understanding personal cloud storage services,'' in
  \emph{In Proceedings of the 12th ACM SIGCOMM Conference on Internet
  Measurement}.\hskip 1em plus 0.5em minus 0.4em\relax Boston, MA, USA: IMC 12,
  2012, pp. 481--494.

\bibitem{tier_store}
L.~Youseff, M.~Butrico, and D.~D. Silva, ``Toward a unified ontology of cloud
  computing,'' in \emph{2008 Grid Computing Environments Workshop}, Nov 2008,
  pp. 1--10.

\bibitem{Naldi13}
\BIBentryALTinterwordspacing
M.~Naldi and L.~Mastroeni, ``Cloud storage pricing: A comparison of current
  practices,'' in \emph{Proceedings of the 2013 International Workshop on Hot
  Topics in Cloud Services}, ser. HotTopiCS '13.\hskip 1em plus 0.5em minus
  0.4em\relax New York, NY, USA: ACM, 2013, pp. 27--34. [Online]. Available:
  \url{http://doi.acm.org/10.1145/2462307.2462315}
\BIBentrySTDinterwordspacing

\bibitem{amazonaws}
Amazon, \url{https://aws.amazon.com/s3/pricing/}, 2017, accessed 8th Feb,2017.

\bibitem{drop}
Dropbox, \url{https://www.dropbox.com/business/pricing}, 2017, accessed 8th
  Feb,2017.

\bibitem{google}
Google, \url{https://cloud.google.com/storage/pricing}, 2017, accessed 8th
  Feb,2017.

\bibitem{xu}
H.~Xu and B.~Li, ``Dynamic cloud pricing for revenue maximization,'' \emph{IEEE
  Transactions on Cloud Computing}, vol.~1, no.~2, pp. 158--171, July 2013.

\bibitem{ebay}
Q.~Wang, K.~Ren, and X.~Meng, ``When cloud meets ebay: Towards effective
  pricing for cloud computing,'' in \emph{2012 Proceedings IEEE INFOCOM}, March
  2012, pp. 936--944.

\bibitem{double_auction}
H.~Zhang, H.~Jiang, B.~Li, F.~Liu, A.~V. Vasilakos, and J.~Liu, ``A framework
  for truthful online auctions in cloud computing with heterogeneous user
  demands,'' \emph{IEEE Transactions on Computers}, vol.~65, no.~3, pp.
  805--818, March 2016.

\bibitem{random_auction}
L.~Zhang, Z.~Li, and C.~Wu, ``Dynamic resource provisioning in cloud computing:
  A randomized auction approach,'' in \emph{IEEE INFOCOM 2014 - IEEE Conference
  on Computer Communications}, April 2014, pp. 433--441.

\bibitem{online_auction}
W.~Shi, L.~Zhang, C.~Wu, Z.~Li, and F.~C. Lau, ``An online auction framework
  for dynamic resource provisioning in cloud computing,'' in \emph{The 2014 ACM
  International Conference on Measurement and Modeling of Computer Systems},
  ser. SIGMETRICS '14.\hskip 1em plus 0.5em minus 0.4em\relax New York, NY,
  USA: ACM, 2014, pp. 71--83.

\bibitem{lin}
W.~Y. Lin, G.~Y. Lin, and H.~Y. Wei, ``Dynamic auction mechanism for cloud
  resource allocation,'' in \emph{2010 10th IEEE/ACM International Conference
  on Cluster, Cloud and Grid Computing}, May 2010, pp. 591--592.

\bibitem{varian2014vcg}
H.~R. Varian and C.~Harris, ``The vcg auction in theory and practice,''
  \emph{The American Economic Review}, vol. 104, no.~5, pp. 442--445, 2014.

\bibitem{latency}
N.~Shalom, “Amazon found every 100ms of latency cost them 1\% in sales.''
  \url{http://blog.gigaspaces.com/amazon-found-every-100ms-of-latency-cost-them-1-in-sales/},
  2008, accessed 8th Feb. 2017.

\bibitem{chan1997pollaczek}
W.~Chan, T.-C. Lu, and R.-J. Chen, ``Pollaczek-khinchin formula for the m/g/1
  queue in discrete time with vacations,'' \emph{IEE Proceedings-Computers and
  Digital Techniques}, vol. 144, no.~4, pp. 222--226, 1997.

\bibitem{equivalent_deterministic}
R.~J.-B. Wets, ``Stochastic programs with fixed recourse: The equivalent
  deterministic program,'' \emph{SIAM Review}, vol.~16, no.~3, pp. 309--339,
  1974.

\bibitem{conopt}
A.~S. Drud, ``Conopt: A large-scale grg code,'' \emph{ORSA Journal on
  Computing}, vol.~6, no.~2, pp. 207--216, 1994.

\bibitem{lz04}
W.~Chen and L.~Sha, ``An energy-aware data-centric generic utility based
  approach in wireless sensor networks,'' in \emph{In Proceedings of the third
  international symposium on Information processing in sensor networks}.\hskip
  1em plus 0.5em minus 0.4em\relax Berkeley California USA: IPSN, 2004, pp.
  215--224.

\end{thebibliography}




\appendices
\section{Proof of Lemma \ref{lemma3}}\label{prolemma3}

Suppose $g(p) =  \mathbb{E}(\pi | \pi \leq p) = \frac{ \int_{\underline{\pi}}^{p}x f_{\pi}(x)dx}{F_{\pi}(p)}$, we take the first-order derivative of $g(p)$ over $p$ and get $\frac{\partial g(p)}{\partial p} = \frac{f_{\pi}(p)}{F_{\pi}(p)^{2}}( pF_{\pi}(p)- \int_{\underline{\pi}}^{p}x f_{\pi}(x)dx)$. Suppose $h(p) =pF_{\pi}(p)- \int_{\underline{\pi}}^{p}x f_{\pi}(x)dx $, we take the first derivative of $h(p)$ over $p$ and get $\frac{\partial h(p)}{\partial p} = F_{\pi}(p) \geq 0$, which means $h(p)$ monotonically increases with $p$ and the minimum value of $h(p)$ is $h_{min}(p)=h(\underline{\pi}) = 0$. Because $\frac{f_{\pi}(p)}{F_{\pi}(p)^{2}}\geq 0$, $\frac{\partial g(p)}{\partial p} \geq 0$. Therefore, $g(p)$ monotonically increase with $p$. The maximum value of $g(p)$ is $g_{\max}(p) = g(\bar{\pi}) =  \frac{\int_{\underline{\pi}}^{\bar{\pi}}x f_{\pi}(x)dx} {F_{\pi}(\bar{\pi})}\leq \frac{\underline{\pi}+\bar{\pi}}{2}$. When $\underline{\pi}$ is closed to 0, $g_{max}(p) \leq \frac{\bar{\pi}}{2}$.

\section{Proof of Claim \ref{cla1}}\label{claim1}

We will prove this result by contradiction. Suppose $\pi_{1}$ is the expected optimal bid price and $q_{1} > \frac{1}{2}$ is the optimal portion of job running on on-demand instance, and the associated optimal objective value is $obj_{1} = q_{1}t_{e}\bar{\pi}+(1-q_{1})t_{e}E[\pi|\pi\leq \pi_1]$.  

Now, we show that we can achieve a lower value by employing a strategy different to the above one. Note that since $q_1>1/2$. Thus, the value of $obj_1$ is at least $t_e\bar{\pi}/2$.

Thus, there exists a solution $q\leq 1/2$ such that $qt_e\bar{\pi}+(1-q)t_e\bar{\pi}/2=obj_1$. Now, consider the strategy $q^{*}=q-\epsilon$, and the bidding price $p=\bar{\pi}$. Since $p=\bar{\pi}$, thus, $t_n=0$.  Now consider $q=q_1-\dfrac{t_s}{t_e}$. Since $t_e/2<t_s<t_e$, thus, $0<q<1/2$. The bid price be $\bar{\pi}$. The above strategy satisfies all the constraints. The objective value is thus at most 
$obj_2=(q-\epsilon)t_e\bar{\pi}+(1-q+\epsilon)t_e\bar{\pi}/2$ which is less than $obj_{1}$. Hence, we obtain a lower value by  employing a different strategy. Hence, the strategy is not optimal.

Therefore, $q^{*}\leq \frac{1}{2}$. 

Now, we show that $F_{\pi}(p)\geq 1/2$ for an optimal bidding. From constraint (\ref{P1-1})
\begin{equation}\label{cc1}
\begin{split}
&(1-q^{*})t_{e} \leq \frac{t_{k}}{1-F_{\pi}(p^{*})} \\
\iff & (1-q^{*})t_{e}F_{\pi}(p^{*}) \geq (1-q^{*})t_{e}-t_{k}\\
\iff & t_{k}(1-F_{\pi}(p^{*})) + (1-q^{*})t_{e}F_{\pi}(p^{*}) \geq (1-q^{*})t_{e} - t_{k} F_{\pi}(p^{*})
\end{split}
\end{equation}

From constraint (\ref{P1-3})
\begin{equation}\label{cc2}
\begin{split}
& t_{k}(\frac{1}{F_{\pi}(p^{*})}-1) +(1-q^{*})t_{e} \leq t_{s}  \\
\iff & t_{k}(1-F_{\pi}(p^{*})) + (1-q^{*})t_{e}F_{\pi}(p^{*}) \leq t_{s}F_{\pi}(p^{*})
\end{split}
\end{equation}

According to (\ref{cc1}) and (\ref{cc2}), we can get that 
\begin{equation}
\begin{split}
(t_{s}+t_{k})F_{\pi}(p^{*}) \geq (1-q^{*})t_{e} \geq \frac{t_{e}}{2}
\end{split}
\end{equation}

Because $t_{s}+t_{k} \leq t_{e}$, we can get $F_{\pi}(p^{*}) \geq \frac{1}{2}$.



\section{Proof of Proposition \ref{pro1}}\label{apdxpro1}

By taking the first-order derivative of $\Phi_{1}(p,q)$ over $q$, we have

\begin{equation}\label{mono}
 \frac{\partial \Phi(p,q)}{\partial q} = t_{e}(\bar{\pi} - \frac{\int_{\underline{\pi}}^{p}x f_{\pi}(x)dx }{F_{\pi}(p)}) \geq 0.
\end{equation}

\begin{color}{black} Therefore, $\Phi_{1}(p,q)$ increases monotonically with $q$, the user can minimize his expected total cot by choosing the smallest possible $q$ in the feasible set. First we consider constraints (2) and (4), we get $q^{*} = \max \{ 1-\frac{t_{k}}{t_{e}(1-F_{\pi}(p))}, 1- \frac{  t_{s}-t_{k}( \frac{1}{F_{\pi}(p^{*})} -1  )   }{t_{e}} \}$, then we will go back to check constraints (7) and (9). \end{color}

When
\begin{equation}\label{condition}
1-\frac{t_{k}}{t_{e}(1-F_{\pi}(p))} \leq 1- \frac{  t_{s}-t_{k}( \frac{1}{F_{\pi}(p^{*})} -1  )   }{t_{e}}
\end{equation}

$q^{*} = 1- \frac{  t_{s}-t_{k}( \frac{1}{F_{\pi}(p^{*})} -1  )   }{t_{e}}$. 


Then we substitute $q$ in (P1) with $q^{*}$, (P1)  becomes

\begin{align}
\begin{split}
{\text{(P1') \quad   min}} &
\quad G(p) =  [t_{s}-t_{k}(\frac{1}{F_{\pi}(p)}-1)][ \frac{  \int_{\underline{\pi}}^{p}x f_{\pi}(x)dx  }{  F_{\pi}(p) } - \bar{\pi}     ]\\
& + t_{e} \bar{\pi}
  \end{split}
\label{obj11} 
\\[2ex]
\text{subject to}\qquad & t_{s} - t_{k}(\frac{1}{F_{\pi}(p)} -1)\leq \frac{t_{k}}{1-F_{\pi}(p)} 
\label{P1-111} 
\\
& 1-\frac{t_{s}}{t_{e}}+\frac{t_{k}}{t_{e}}(\frac{1}{F_{\pi}(p)} -1) \leq \frac{t_{s}}{t_{e}} 
\label{P1-222}
\\
& \underline{\pi} \leq  p \leq \bar{\pi}\\
& 0 \leq 1-\frac{t_{s}}{t_{e}}+\frac{t_{k}}{t_{e}}(\frac{1}{F_{\pi}(p)} -1)  \leq 1 \label{P1-333}
\end{align}

In terms of constraint (\ref{P1-222}), 

\begin{equation}\label{P1-2122}
\begin{split}
& 1- \frac{  t_{s}-t_{k}( \frac{1}{F_{\pi}(p^{*})} -1  )   }{t_{e}} \leq \frac{t_{s}}{t_{e}} \\
\iff & \frac{t_{k}}{t_{k}-t_{e}+2t_{s}} \leq F_{\pi}(p^{*})\\
\iff & \frac{t_{k}}{2t_{s}-t_{e}+t_{k}} \leq \frac{t_{k}}{t_{k}+t_{k}} = \frac{1}{2} \leq F_{\pi}(p^{*}) 
\end{split}
\end{equation}

which has been proved in claim \ref{cla1}. Thus, $$q^{*}=1- \frac{  t_{s}-t_{k}( \frac{1}{F_{\pi}(p^{*})} -1  )   }{t_{e}} \leq \frac{t_{s}}{t_{e}}\leq 1.$$

Also $$q^{*}=1-\frac{t_{s}}{t_{e}}+\frac{t_{k}}{t_{e}}(\frac{1}{F_{\pi}(p^{*})} -1) \geq 1-\frac{t_{s}}{t_{e}} \geq 0$$ because $0 \leq F_{\pi}(p^{*}) \leq 1$. 

Thus constraints (\ref{P1-222}) and (\ref{P1-333}) hold when the optimization model (P1') takes the optimal value. 

Next, we will consider objective function $G(p)$ in (P1'). 

\begin{equation}
\begin{split}
\frac{\partial G(p)}{\partial p}&=\frac{t_{k}f_{\pi}(p)}{F_{\pi}(p)^{2}}[ \frac{  \int_{\underline{\pi}}^{p}x f_{\pi}(x)dx  }{  F_{\pi}(p) } - \bar{\pi} ] \\
&+ \frac{f_{\pi}(p)[pF_{\pi}(p)-\int_{\underline{\pi}}^{p}xf_{\pi}(x)dx]}{F_{\pi}^{2}(p)}[t_{s}-t_{k}(\frac{1}{F_{\pi}(p)}-1)]\\
&=\frac{f_{\pi}(p)}{F_{\pi}(p)^{2}}[  \frac{  \int_{\underline{\pi}}^{p}x f_{\pi}(x)dx  }{  F_{\pi}(p) }t_{k} - \bar{\pi}t_{k} \\
& + (pF_{\pi}(p)-\int_{\underline{\pi}}^{p}xf_{\pi}(x)dx) (t_{s}-t_{k}(\frac{1}{F_{\pi}(p)}-1))    ]
\end{split}
\end{equation}

Suppose $g(p) =  \frac{  \int_{\underline{\pi}}^{p}x f_{\pi}(x)dx  }{  F_{\pi}(p) }t_{k} - \bar{\pi}t_{k}  + (pF_{\pi}(p)-\int_{\underline{\pi}}^{p}xf_{\pi}(x)dx) (t_{s}-t_{k}(\frac{1}{F_{\pi}(p)}-1)) $

\begin{equation}
\begin{split}
&\frac{\partial g(p)}{\partial p} \\
&= \frac{2t_{k}f_{\pi}(p)}{F_{\pi}(p)} ( p - \frac{\int_{\underline{\pi}}^{p}xf_{\pi}(x)dx}{F_{\pi}(p)}   )+ F_{\pi}(p)(t_{s}-t_{k}(\frac{1}{F_{\pi}(p)}-1)) \\
\end{split}
\end{equation}

It is clear that $p  \geq \frac{\int_{\underline{\pi}}^{p}xf_{\pi}(x)dx}{F_{\pi}(p)} $, and since $F_{\pi}(p*)\geq \frac{1}{2}$, $t_{s}-t_{k}(\frac{1}{F_{\pi}(p)}-1) \geq t_{s}-t_{k} \geq 0$, $\frac{\partial g(p)}{\partial p} \geq 0$. Thus $g(p)$ monotonically increases with $p$.
Because $\lim_{p \to \underline{\pi} }g(p) = t_{k}\underline{\pi}-t_{k}\bar{\pi} <0 $, and $g(\bar{\pi})=( \bar{\pi}-\int_{\underline{\pi}}^{\bar{\pi}}xf_{\pi}(x)dx )(t_{s}-t_{k})>0$, $g(p)$ increases monotonically from a negative value to a positive value. The term $\frac{t_{k}f_{\pi}(p)}{F_{\pi}(p)^{2}}> 0$ in $\frac{\partial G(p)}{\partial p}$, thus $\frac{\partial G(p)}{\partial p}$ also increases monotonically from a negative value to a positive value, i.e., $G(p)$ first decreases and then increases with $p$. Thus $G(p)$ is minimized when $\frac{\partial G(p)}{\partial p}=0$. Letting $\frac{\partial G(p)}{\partial p} = 0$, we thus deduce 

\begin{equation}
\begin{split}
&\psi_{1}(p) \\
&=\frac{ 2t_{k} \int_{\underline{\pi}}^{p}x f_{\pi}(x)dx  }{  F_{\pi}(p) } + 2pt_{s}F_{\pi}(p) - pt_{s} - (t_{s}+t_{k}) \int_{\underline{\pi}}^{p}x f_{\pi}(x)dx\\
&=  t_{k}\bar{\pi}
\end{split}
\end{equation}

Thus the objective function $G(p)$ takes the optimal value at $$\hat{p_{1}}^{*} = \psi_{1}^{-1} (t_{k}\bar{\pi})$$ regardless of other constraints. 

Further, we will analyze constraint (\ref{P1-111}), which is equivalent to 

\begin{equation}\label{con111}
\psi_{2}(p) = (t_{s}+t_{k})F_{\pi}(p) - (t_{s}+t_{k})F_{\pi}(p)^{2} - t_{k} \leq 0
\end{equation}

Because

\begin{equation}\label{const}
\frac{\partial \psi_{2}(p)}{\partial p} = (t_{s}+t_{k})f_{\pi}(p)(1-2F_{\pi}(p))\leq 0
\end{equation}
$\psi_{2}(p)$ is monotone decreasing with $p$.

 Also because $\lim_{p: F_{\pi}(p) \to \frac{1}{2}} \psi_{2}(p) \geq 0$ and $\psi_{2}(  \bar{\pi}) \leq 0$, there exists one and only one $\hat{p_{2}}^{*} $ such that $\psi_{2}(\hat{p_{2}}^{*} ) = 0$, so constraint (\ref{P1-111}) ( or (\ref{con111}))is equivalent to

\begin{equation}\label{con1111}\nonumber
p \geq \hat{p_{2}}^{*}
\end{equation}

In order to get the optimal solution for (P1'), we need to take the maximum of $ \hat{p_{1}}^{*}$ and $ \hat{p_{2}}^{*}$, i.e., $p^{*}= \max \{  \hat{p_{1}}^{*},  \hat{p_{2}}^{*}   \}= \max \{\psi_{1}^{-1} (t_{k}\bar{\pi}), \psi_{2}^{-1} (0)\} $. 

Finally we will go back to check whether condition (\ref{condition}) can be satisfied with $p^{*} = \max \{\psi_{1}^{-1} (t_{k}\bar{\pi}), \psi_{2}^{-1} (0)\} $. 

We need to consider the following two cases:

Case 1: when $\psi_{2}^{-1} (0) \geq \psi_{1}^{-1} (t_{k}\bar{\pi})$, $p^{*} = \psi_{2}^{-1} (0)$, that is $(t_{s}+t_{k})F_{\pi}(p^{*}) - (t_{s}+t_{k})F_{\pi}(p^{*})^{2} - t_{k}=0$, which is equivalent to $1-\frac{t_{k}}{t_{e}(1-F_{\pi}(p))} = 1- \frac{  t_{s}-t_{k}( \frac{1}{F_{\pi}(p^{*})} -1  )   }{t_{e}}$. Thus the condition (\ref{condition}) holds in case 1. 

Case 2: when $\psi_{1}^{-1} (t_{k}\bar{\pi}) \geq \psi_{2}^{-1} (0)$, $p^{*}= \psi_{1}^{-1} (t_{k}\bar{\pi}) $. According to (\ref{const}), we know that $\psi_{2}(p)$ is monotone decreasing with $p$. Thus  $\psi_{2}(\psi_{1}^{-1} (t_{k}\bar{\pi})) <  \psi_{2}(\psi_{2}^{-1} (0))=0$, which is equivalent to $1-\frac{t_{k}}{t_{e}(1-F_{\pi}(p))} < 1- \frac{  t_{s}-t_{k}( \frac{1}{F_{\pi}(p^{*})} -1  )   }{t_{e}}$. Therefore, the condition (\ref{condition}) also holds in case 2.

\begin{color}{black}
Then we will consider when
\begin{equation}\label{condition}
1-\frac{t_{k}}{t_{e}(1-F_{\pi}(p))} \geq 1- \frac{  t_{s}-t_{k}( \frac{1}{F_{\pi}(p^{*})} -1  )   }{t_{e}}
\end{equation}

$q^{*} = 1-\frac{t_{k}}{t_{e}(1-F_{\pi}(p))}$

Then we substitute $q$ in (P1) with $q^{*}$, (P1) becomes

\begin{align}
\begin{split}
{\text{(P1'') \quad   min}} &
\quad G_{2}(p) = (1-\frac{t_{k}}{t_{e}(1-F_{\pi}(p))})t_{e}\bar{\pi} + \frac{  t_{k}  \int_{\underline{\pi}}^{p}x f_{\pi}(x)dx                   }{ F_{\pi}(p)( 1-F_{\pi}(p)  ) }
  \end{split}
\label{obj111} 
\\[2ex]
\text{subject to}\qquad & 1-\frac{t_{k}}{t_{e}(1-F_{\pi}(p))} \geq 1- \frac{  t_{s}-t_{k}( \frac{1}{F_{\pi}(p^{*})} -1  )   }{t_{e}}
\label{P11-111} 
\\
& 1-\frac{t_{k}}{t_{e}(1-F_{\pi}(p))} \leq \frac{t_{s}}{t_{e}} 
\label{P11-222}
\\
& \underline{\pi} \leq  p \leq \bar{\pi}\\
& 0 \leq 1-\frac{t_{k}}{t_{e}(1-F_{\pi}(p))}  \leq 1 \label{P11-333}
\end{align}

Next, we take the first derivative of the objective function in (P1''), 
\begin{equation}
\begin{split}
&\frac{\partial G_{2}(p)}{\partial p} \\
&=\frac{f_{\pi}(p)t_{k} [  -\bar{\pi} F_{\pi}(p)^{2}  +pF_{\pi}(p)(1-F_{\pi}(p)) -(1-2F_{\pi}(p)) \int_{\underline{\pi}}^{p}x f_{\pi}(x)dx ]   }{F_{\pi}(p)^{2}(1-F_{\pi}(p))^{2}}
\end{split}
\end{equation}

Suppose $g_{2}(p) =   -\bar{\pi} F_{\pi}(p)^{2}  +pF_{\pi}(p)(1-F_{\pi}(p)) -(1-2F_{\pi}(p)) \int_{\underline{\pi}}^{p}x f_{\pi}(x)dx  $, 
Now, we will prove $g_{2}(p) \leq 0$
\begin{equation}
\begin{split}
&pF_{\pi}(p) (1-F_{\pi}(p))+ (2F_{\pi}(p) - 1)\int_{\underline{\pi}}^{p}x f_{\pi}(x)dx\\
=& F_{\pi}(p)[ p(1-F_{\pi}(p)) + (2F_{\pi}(p) - 1 ) \frac{  \int_{\underline{\pi}}^{p}x f_{\pi}(x)dx    }{ F_{\pi}(p) }]\\
\leq & F_{\pi}(p)[ \bar{\pi} (1-F_{\pi}(p))  +    \bar{\pi}( 2F_{\pi}-1  )   ]\\
= & \bar{\pi}F_{\pi}(p)^{2}
\end{split}
\end{equation}

Thus, $g_{2}(p) = pF_{\pi}(p) (1-F_{\pi}(p))+ (2F_{\pi}(p) - 1)\int_{\underline{\pi}}^{p}x f_{\pi}(x)dx -  \bar{\pi}F_{\pi}(p)^{2} \leq 0$
%
%
%
Then we can get $\frac{\partial G_{2}(p)}{\partial p} \leq 0 $, thus $G_{2}(p)$ decreases monotonically with $p$, and the optimal solution is the largest possible $p$ in its feasible set. 

Because $ 1- \frac{  t_{s}-t_{k}( \frac{1}{F_{\pi}(p^{*})} -1  )   }{t_{e}} \geq 1-\frac{t_{s}}{t_{e}} > 0$. Thus, if we can meet constraint (\ref{P11-111}), constraint (\ref{P11-333}) will be met. 

As we analyzed in the last part, constraint (\ref{P11-111}) is equivalent to \begin{equation}\label{con111}
\psi_{2}(p) = (t_{s}+t_{k})F_{\pi}(p) - (t_{s}+t_{k})F_{\pi}(p)^{2} - t_{k} \leq 0
\end{equation}
which  is monotone decreasing with $p$. Thus the largest feasible solution will be obtained when $\psi_{2}(p) = 0$. Therefore, the optimal solution $p^{*} = \psi_{2}^{-1}(0)$. 

In summary, 

\begin{itemize}
\item When $1-\frac{t_{k}}{t_{e}(1-F_{\pi}(p))} \leq 1- \frac{  t_{s}-t_{k}( \frac{1}{F_{\pi}(p^{*})} -1  )   }{t_{e}}$, 

$p^{*} = \max \{\psi_{1}^{-1} (t_{k}\bar{\pi}), \psi_{2}^{-1} (0)\} $, $q^{*} =1- \frac{  t_{s}-t_{k}( \frac{1}{F_{\pi}(p^{*})} -1  )   }{t_{e}}$;

\item When $ 1-\frac{t_{k}}{t_{e}(1-F_{\pi}(p))} \geq 1- \frac{  t_{s}-t_{k}( \frac{1}{F_{\pi}(p^{*})} -1  )   }{t_{e}} $, 

$p^{*} = \psi_{2}^{-1}(0)$, $q^{*} = 1-\frac{t_{k}}{t_{e}(1-F_{\pi}(p))} = 1- \frac{  t_{s}-t_{k}( \frac{1}{F_{\pi}(p^{*})} -1  )   }{t_{e}}$
\end{itemize}

We combine the above two cases: $p^{*} = \max \{\psi_{1}^{-1} (t_{k}\bar{\pi}), \psi_{2}^{-1} (0)\} $, $q^{*} =1- \frac{  t_{s}-t_{k}( \frac{1}{F_{\pi}(p^{*})} -1  )   }{t_{e}}$. Thus Proposition \ref{pro1} is proved. 
\end{color}


\section{Proof of Proposition \ref{pro2}}\label{apdxpro2}

Recalling (\ref{mono}) in Proposition \ref{pro1}, we know that $\Phi_{1}(p,q)$ increases monotonically with $q$, the user can minimize his expected total cost by choosing the smallest possible $q$ in the feasible set. That is, $q^{*} = \max \{ 1-\frac{t_{k}}{t_{e}(1-F_{\pi}(p))}, 1- \frac{  t_{s}-t_{k}( \frac{1}{F_{\pi}(p^{*})} -1  )   }{t_{e}}, 0 \}$. 

When 
\begin{equation}\label{cond1}
1-\frac{t_{k}}{t_{e}(1-F_{\pi}(p))}\leq 0
\end{equation}
and 
\begin{equation}\label{cond2}
1- \frac{  t_{s}-t_{k}( \frac{1}{F_{\pi}(p^{*})} -1  )   }{t_{e}} \leq 0
\end{equation}

$q^{*}=0$

When $q=0$, the objective function of (P1) becomes $$G(p) = \frac{t_{e}  \int_{\underline{\pi}}^{p}x f_{\pi}(x)dx}{F_{\pi}(p)}$$
Then we take the first derivative of $G(p)$ and get $$  \frac{\partial G(p)}{\partial p} =t_{e}f_{\pi}(p)(  \frac{p - \frac{\int_{\underline{\pi}}^{p}xf_{\pi}(x)dx}{F_{\pi}(p)}       }{  F_{\pi}(p)  }   ) \geq 0, $$ which means the $G(p)$ monotonic increasing with $p$.  Minimizing $G(p)$ is equivalent to finding the minimum $p$ in its feasible set. Suppose $t_{s}-t_{e}\geq t_{k}$ and $t_{e}\geq 2t_{k}$, then
$p^{*} = \max\{  F_{\pi}^{-1}(1-\frac{t_{k}}{t_{e}}), F_{\pi}^{-1}(\frac{t_{k}}{t_{s}-t_{e}+t_{k}}), \underline{\pi} \}= F_{\pi}^{-1}(1-\frac{t_{k}}{t_{e}})$, which is equivalent to $$1-\frac{t_{k}}{t_{e}(1-F_{\pi}(p))}=0,$$ thus condition (\ref{cond1}) is satisfied. 

Next we will check whether condition (\ref{cond2}) will be satisfied. 

\begin{equation}\nonumber
\begin{split}
1- \frac{  t_{s}-t_{k}( \frac{1}{F_{\pi}(p^{*})} -1  )   }{t_{e}} &  \leq 1- \frac{t_{e}+t_{k}}{t_{e}} + \frac{t_{k}}{t_{e}}(\frac{t_{k}}{t_{e}}-1 )    \\
&= \frac{t_{k}(2t_{k}-t_{e}) }{t_{e}(t_{e}-t_{k})} \\
&\leq 0
\end{split}
\end{equation}

thus condition (\ref{cond2}) is also satisfied. This proves the result as in the statement of Proposition \ref{pro2}.

\section{Proof of Claim \ref{cla2}}\label{claim2}

Recall that with bid price $p$ and deadline $t_{s}$, $F_{\pi}(p)$ denotes the probability that the bid price $p \geq \pi(t)$, the spot price, the job's expected running time on spot instance is $t_{s}F_{\pi}(p)$. In order to guarantee the job can be finished before deadline, $t_{s}F_{\pi}(p^{*}) \geq (1-q^{*})t_{e}$, 

\begin{equation}
\begin{split}
& t_{s}F_{\pi}(p^{*}) \geq (1-q^{*})t_{e} \\
\iff & F_{\pi}(p^{*}) \geq (1-q^{*})\frac{t_{e}}{t_{s}}\\
\iff & F_{\pi}(p^{*}) \geq 1-q^{*}\\
\iff & F_{\pi}(p^{*}) \geq \frac{1}{2}. 
\end{split}
\end{equation}

\section{Proof of Proposition \ref{pro3}}\label{apdxpro3}

When $\frac{t_{e}}{2} < t_{s} < t_{e}$, we take the first-order derivative of $G(p,q)$ over $q$ and get

\begin{equation}
\begin{split}
\frac{\partial \Phi_{3}(p,q)}{\partial q}  &= t_{e}( \bar{\pi} -\frac{ 1 }{1-  \frac{t_{r}}{t_{k}} (1-F_{\pi}(p) )   } \frac{ \int_{\underline{\pi}}^{p}x f_{\pi}(x)dx}{F_{\pi}(p)}   )\\
& \geq t_{e}(\bar{\pi} -  \frac{ 2\int_{\underline{\pi}}^{p}x f_{\pi}(x)dx}{F_{\pi}(p)}     )
\end{split}
\end{equation}
Suppose $g(p) = \bar{\pi} -  \frac{ 2\int_{\underline{\pi}}^{p}x f_{\pi}(x)dx}{F_{\pi}(p)}$, and take the first-order of derivative of $g(p)$, we get 

\begin{equation}\nonumber
\frac{\partial g(p)}{\partial p} = -\frac{2f_{\pi}(p)}{F_{\pi}(p)^{2}}(pF_{\pi}(p) -  \int_{\underline{\pi}}^{p}x f_{\pi}(x)dx  ) \leq 0
\end{equation}
thus $g(p)$ monotonically decrease with $p$. The minimum value of $g(p)$ is $g_{min}(p) = g(\bar{\pi}) = \bar{\pi} -  \frac{ 2\int_{\underline{\pi}}^{\bar{\pi}}x f_{\pi}(x)dx}{F_{\pi}(p)} \geq 0$. Then we can get

\begin{equation}\label{phi3mono}
\frac{\partial \Phi_{3}(p,q)}{\partial q} \geq 0. 
 \end{equation}

 Therefore, $\Phi_{3}(p,q)$ monotonically increases with $q$. Minimizing $\Phi_{3}(p,q)$ is equivalent to finding the minimum $q$ in its feasible set, that is, 

\begin{equation}\nonumber
\begin{split}
q^{*}  =\max \{1-\frac{t_{s}F_{\pi}(p)( 1-\frac{t_{r}}{t_{k}}(1-F_{\pi}(p))   )  }{t_{e}}, 0  \}
\end{split}
\end{equation}

Because 
\begin{equation}\label{condd}
 1-\frac{t_{s}F_{\pi}(p)( 1-\frac{t_{r}}{t_{k}}(1-F_{\pi}(p))   )  }{t_{e}} \geq 1-\frac{t_{s}}{t_{e}} \geq 0,
 \end{equation}
 
\begin{equation}\label{P3-q}
 q^{*} = 1-\frac{t_{s}F_{\pi}(p)( 1-\frac{t_{r}}{t_{k}}(1-F_{\pi}(p))   )  }{t_{e}}
\end{equation}

Then substitute $q$ using $q^{*}$ in (P3), we will get a new optimization problem (P3')

\begin{align}
\begin{split}
{\text{(P3') \quad   min}} &
\quad  G(p) = t_{e}\bar{\pi}  - t_{s}F_{\pi}(p)\bar{\pi}[ 1-\frac{t_{r}}{t_{k}}(1-F_{\pi}(p))] \\
&+t_{s} \int_{\underline{\pi}}^{p}f_{\pi}(x)dx
  \end{split}
\label{green} 
\\
\text{s.t.}\qquad &  (1-\frac{t_{r}+t_{s}F_{\pi}(p)( 1-\frac{t_{r}}{t_{k}}(1-F_{\pi}(p))   )  }{t_{e}} ) t_{e} \leq t_{s}
\label{gc1}
\\
& t_{r} < \frac{t_{k}}{2(1-F_{\pi}(p))}
\label{gc2}
\\
& \underline{\pi} \leq  p \leq \bar{\pi}
\label{gc3}
\\
& 1-\frac{t_{s}F_{\pi}(p)( 1-\frac{t_{r}}{t_{k}}(1-F_{\pi}(p))   )  }{t_{e}}  \geq 0
\label{gc4}
\end{align}

Take the first-order derivative of $G(p)$, we have 
$$ \frac{\partial G(p)}{\partial p}= (p-\bar{\pi})t_{s}f_{\pi}(p) + \frac{t_{r}t_{s}}{t_{k}}\bar{\pi}f_{\pi}(p)(1-2F_{\pi}(p)) <0, $$ so $G(p)$ monotonic decreasing with $p.$ Thus, the optimal solution is $p^{*}= \bar{\pi}$. In addition, the constraints (\ref{gc1}) and (\ref{gc4}) are satisfied at optimality. 

Substitute $p$ with $p^{*}=\bar{\pi}$ in (\ref{P3-q}), we will get $ q^{*} = 1- \frac{t_{s}  }{t_{e}} \geq 0$. This proves the result given in the statement of Proposition \ref{pro3}.

\section{Proof of Proposition \ref{pro4}}\label{apdxpro4}

From (\ref{phi3mono}) we know that $\Phi_{3}(p,q)$ monotonically increases with $q$. 

\begin{equation}\nonumber
\begin{split}
q^{*}  =\max \{1-\frac{t_{s}F_{\pi}(p)( 1-\frac{t_{r}}{t_{k}}(1-F_{\pi}(p))   )  }{t_{e}}, 0  \}
\end{split}
\end{equation}

When 
\begin{equation}\label{cont}
1-\frac{t_{s}F_{\pi}(p)( 1-\frac{t_{r}}{t_{k}}(1-F_{\pi}(p))   )  }{t_{e}}\leq 0,
\end{equation}
$$q^{*}=0,$$

then optimization problem (P3) will become 

\begin{align}
\begin{split}
{\text{(P3'') \quad   min}}  \quad &  \Phi(p) = \frac{ t_{e} }{1-  \frac{t_{r}}{t_{k}} (1-F_{\pi}(p) )   } \frac{ \int_{\underline{\pi}}^{p}x f_{\pi}(x)dx}{F_{\pi}(p)}
\end{split}
\\
\text{subject to}&   \frac{t_{e}}{  1- \frac{t_{r}}{t_{k}}(1-F_{\pi}(p))   }\frac{1}{F_{\pi}(p)}  \leq t_{s} 
\label{P3111}
\\
& \underline{\pi} \leq  p \leq \bar{\pi} \label{ppp}
\end{align}

By taking the first derivative of $\Phi(p)$ in (P3''), we will have 

$$\frac{\partial \Phi(p)}{\partial p  } =  \frac{ (t_{e})f_{\pi}(p)( 1-\frac{t_{r}}{t_{k}}+2\frac{t_{r}}{t_{k}}F_{\pi}(p)  )   }{ h(p)^{2}   }g(p)    $$

where 
$$g(p)=-\int_{\underline{\pi}}^{p} x f_{\pi}(x)dx+p \frac{(1-\frac{t_{r}}{t_{k}})F_{\pi}(p)+\frac{t_{r}}{t_{k}}(F_{\pi}(p))^{2} }{1-\frac{t_{r}}{t_{k}}+2\frac{t_{r}}{t_{k}}F_{\pi}(p)}$$

and 
$$
h(p)=(1-\frac{t_{r}}{t_{k}})F_{\pi}(p)+\frac{t_{r}}{t_{k}}(F_{\pi}(p))^{2}
$$

Because $\frac{ (t_{e})f_{\pi}(p)( 1-\frac{t_{r}}{t_{k}}+2\frac{t_{r}}{t_{k}}F_{\pi}(p)  )   }{ h(p)^{2}   }>0$, in order to show the positivity of $\Phi(p)$, we take the first derivative of $g(p)$ and then we have 

$$\frac{\partial g(p)}{\partial p}=\frac{1-\frac{t_{r}}{t_{k}}+2 \frac{t_{r}}{t_{k}}(F_{\pi}(p)-pf_{\pi}(p))}{(1-\frac{t_{r}}{t_{k}}+2\frac{t_{r}}{t_{k}}F_{\pi}(p))^{2} }h(p)$$

Because $F_{\pi}(p)$ is concave and $F_{\pi}(p)-pf_{\pi}(p) \geq 0$, we have $\frac{\partial g(p)}{\partial p} \geq 0$. Thus, $g(p)$     monotonically increases with $p$. By the fact that $g(\underline{\pi})=0$, then $g(p) \geq 0$ and  $\frac{\partial \Phi(p)}{\partial p  } \geq 0 $, i.e., $\Phi(p)$ increases monotonically with $p$. 

In order to minimize the total cost, we just need to choose the lowest feasible bid price. Constraint (\ref{P3111}) in (P3'') is equivalent to 
\begin{equation}
g(p) =  \frac{t_{r}}{t_{k}}F_{\pi}(p)^{2} + (1- \frac{t_{r}}{t_{k}})F_{\pi}(p) \geq \frac{t_{e}}{t_{s}}.
\end{equation}

The axis of symmetry of $g(F_{\pi}(p))$ is $F_{\pi}(p) = -\frac{t_{k}-t_{r}}{2t_{r}} <0$, thus $g(F_{\pi}(p))$ monotonically increases with $F_{\pi}(p)$ on the condition that $0\leq F_{\pi}(p)\leq 1$ and $g(p)$ monotonically increases with $p$ on the condition that $\underline{\pi}\leq p\leq \bar{\pi}$. So the minimum value to satisfy constraint (\ref{P3111})  is $g(p^{*})=\frac{t_{e}}{t_{s}}$. Because $0 \leq \frac{t_{e}}{t_{s}} \leq 1$, the maximum and minimum value of $g(p)$ is $g_{max}(p) = g(\bar{\pi})=1 $ and $g_{min}(p) = g(\underline{\pi})=0$ respectively. Therefore, constraint (\ref{ppp}) is satisfied. $g(p^{*}) = \frac{t_{e}}{t_{s}}$ is equivalent to $1-\frac{t_{s}F_{\pi}(p^{*})( 1-\frac{t_{r}}{t_{k}}(1-F_{\pi}(p^{*}))   )  }{t_{e}}= 0$, thus the condition 
(\ref{cont}) is satisfied. This proves the result as in the statement of Proposition \ref{pro4}.

\begin{color}{black}
\section{Proof of Lemma \ref{lemmaaa}}\label{apdxlemmaaa}

Recall that when $t_{s}< t_{e} \leq 2t_{s}$, $F_\pi(p^{*}) \geq \frac{1}{2}$. The difference optimal portions of job to run on on-demand instance in Proposition 1 and Proposition 3 is 

\begin{equation}
\begin{split}
&1- \frac{  t_{s}-t_{k}( \frac{1}{F_{\pi}(p^{*})} -1  )   }{t_{e}} - (1-\frac{t_{s}}{t_{e}}) \\
=& \frac{t_{k}}{t_{e}}( \frac{1}{F_{\pi}(p^{*})} -1  ) \\
\leq & \frac{t_{k}}{t_{e}}
\end{split}
\end{equation}
Note the last step is because $F_\pi(p^{*}) \geq \frac{1}{2}$. 
\end{color}

\end{document}